\documentclass[letter, 11pt]{article}

\usepackage{natbib}
\usepackage{amssymb,amsmath,graphicx,xspace,colortbl,rotating} %
\usepackage[raggedrightboxes]{ragged2e} 
\usepackage{textcomp}
\usepackage{appendix}  
\usepackage{bm}  
\usepackage{boxedminipage}  
\usepackage[usenames, dvipsnames]{xcolor} 
\usepackage{endnotes}  
\usepackage{ragged2e}  
\usepackage[onehalfspacing]{setspace}  
\usepackage{tabulary}  
\usepackage{varioref}  
\usepackage{wrapfig}  
\usepackage{float}
\usepackage{comment}

\usepackage{multimedia}  
\usepackage{pgf,tikz,pgfplots}
\pgfplotsset{width=5cm,compat=1.15}
\usepackage{mathrsfs}
\usetikzlibrary{arrows}
\usepackage{caption,subcaption}
\usepgfplotslibrary{fillbetween}
\definecolor{zzttqq}{rgb}{0.6,0.2,0.0}

\graphicspath{{Stochastic_Path 17 new_graphics/}{Stochastic_Path 17 new_tcache/}{Stochastic_Path 17 new_gcache/}}
\DeclareGraphicsExtensions{.pdf,.eps,.ps,.png,.jpg,.jpeg}
\usepackage[margin=1in]{geometry}
\usepackage {bm}
\usepackage {tabulary}
\usepackage[normalem]{ulem}

\newcommand{\bl}[1]{{\color{blue} [BL: #1]}}
\newcommand{\newbl}[1]{{\color{blue}  #1}}

\newtheorem {theorem}{Theorem}

\newtheorem {corollary}{Corollary}

\newtheorem {definition}{Definition}

\newtheorem {lemma}{Lemma}

\newtheorem {proposition}{Proposition}

\newenvironment {proof}[1][Proof]{\noindent \textbf {#1.} }{\ \rule {0.5em}{0.5em}}
\newcommand{\coef}{\text{coef}}
\newcommand{\diag}{\text{diag}}
\newcommand{\sgn}{\text{sgn}}
\newcommand{\R}{\mathbb{R}}

\begin{document}

\title{Social Learning under Platform Influence:\\ Consensus and Persistent Disagreement}

\author{Ozan Candogan\protect\footnote{  University of Chicago Booth School of Business,  e-mail: \textsf{ozan.candogan@chicagobooth.edu}}   ~  Nicole Immorlica\protect\footnote{  Microsoft Research, e-mail: \textsf{nicimm@gmail.com}} ~ Bar Light\protect\footnote{ Business School and Institute of Operations Research and Analytics, National University of Singapore,  e-mail: \textsf{barlight@nus.edu.sg}} ~ Jerry Anunrojwong\protect\footnote{  Columbia University,  e-mail: \textsf{jerryanunroj@gmail.com}\ }~ ~} 
\maketitle
\thispagestyle{empty}

\begin{abstract}

Individuals increasingly rely on social networking platforms to form opinions. However, these platforms typically aim to maximize engagement, which may not align with social good. In this paper, we introduce an opinion dynamics model where agents are connected in a social network, and update their opinions based on their  neighbors' opinions and on the content shown to them by the platform. 
We focus on a  stochastic block model with two blocks, where the initial opinions of the individuals in different blocks are different.
We prove that for large and dense enough networks
the trajectory of opinion dynamics in such networks
can be approximated well by a simple two-agent system. The latter  admits tractable analytical analysis, which we leverage to provide  interesting insights into the platform's impact on the social learning outcome in our original  two-block model.  Specifically, by using our approximation result, we show that agents' opinions approximately converge to some limiting opinion, which is either: consensus, where all agents agree, or persistent disagreement, where agents' opinions differ.
We find that when the platform is weak and there is a high number of connections between agents with different initial opinions, a consensus equilibrium is likely. In this case, even if a persistent disagreement equilibrium arises, the polarization in this equilibrium, i.e., the degree of  disagreement, is low. When the platform is strong,  a persistent disagreement equilibrium is likely and the equilibrium polarization is high. A moderate platform typically leads to a persistent disagreement equilibrium with moderate polarization. We analyze the effect of initial polarization on consensus and  explore various extensions including a three block stochastic model and a correlation between initial opinions and agents' connection probabilities, thereby providing a more robust and comprehensive analysis of the social learning dynamics. 
\end{abstract}

\newpage
\section{Introduction}\label{sec:intro}




Social ties -- friends and acquaintances -- shape our opinions and beliefs.  With the rise of major social networking platforms such as Facebook, these social interactions are increasingly mediated by a platform. According to the Pew Research Center survey in 2016 \citep{pew-news-social-media}, 62\% of adults get their news from social media, a sharp rise from 49\% just four years earlier. As individuals increasingly rely on social networking platforms to get informed about politics and form opinions \citep{bakshyMA15-ideologically-diverse-facebook}, there are concerns over  whether or how social media lead to political polarization \citep{boxellGS17-pnas, levy2021social}, selective exposure \citep{boxellGS18-internet-2016-election}, extremism \citep{benigniKC17-extremism-twitter}, and disagreement \citep{klofstadSM17-ajps-disagree, barnidge15-news-disagree}.  On the flip side, disengaging with social networking platforms has been shown to decrease polarization and exposure to polarizing news~\citep{allcott2020welfare}.

While offline social learning is often unplanned and decentralized, online social learning is significantly shaped by key decisions made by a centralized platform. For instance, your friends write posts and share articles on a platform, and the News Feed algorithm decides which posts and articles are shown to you. The platform has its own objective -- for example, maximizing engagement -- and it uses its algorithms, which take into account existing user behavior, to choose content for your consumption. While these algorithms are proprietary, there is evidence that they are more likely to expose individuals to news matching their ideology~\citep{bakshyMA15-ideologically-diverse-facebook, levy2021social}; indeed individuals are more likely to share posts matching their ideology,  suggesting that algorithms that promote such posts increase engagement~\citep{halberstam-knight-16}. 
In a recent theoretical paper, \citep{acemoglu2022model} find that  social media platforms interested in maximizing engagement tend to design their algorithms to create more homophilic communication patterns. In a recent empirical paper,  \cite{levy2021social} finds that exposure to social media substantially affects online news consumption and that Facebook's algorithm is less likely to supply individuals with content from counter-attitudinal media outlets. 
The objective of this paper is to investigate the effect of such platform decisions on opinion dynamics. In particular, we explore whether participants reach  consensus and, if not, how much individuals disagree.

We introduce a naive social learning model that explicitly captures the platform's influence. Each agent's opinion is modeled as a scalar, where a lower (higher) number means a more left-leaning (right-leaning) opinion. We focus on a simple average-based opinion updating process, in which agents set their next-period opinion based on the weighted average of sources that influence them. In standard social learning, those sources are the opinions of the agents and their neighbors in a network, and the opinion dynamics, commonly known as \textit{DeGroot}, converge to a weighted average of the initial opinions. In our model, we add an additional source of influence: the \textit{platform}. We assume that there are media outlets publishing content with known opinions, a left-leaning opinion and a right-leaning opinion. At each time step, the platform can give a personalized content recommendation to each agent. Consistent with the evidence 
discussed above, we assume that
the agents are more likely to engage more with content that is close to their opinion, so that at each time period, the platform shows content from the media outlet whose stance is closest to the agent's opinion. In other words, it shows left-leaning content to left-leaning agents, and right-leaning content to right-leaning agents.

Importantly, the platform's content recommendation depends on the agent's \textit{current} opinion, so the platform's content recommendation can change from left to right, say, if the peer influence  brings the agent's opinion from left-leaning to right-leaning. We therefore see that the two sources of influence on opinions may counteract each other: social learning brings   opinions closer together, while the platform's content recommendation pushes opinions toward the media outlets' opinions.  

We analyze the effect of the platform on the social learning outcome, and study how its strength, the agents' initial opinions, and the structure of the underlying social network, impact the outcome. We first analyze a two-agent system with only one left-leaning and one right-leaning agent. We later show that this simple two-agent system approximates a general stochastic block model with two blocks. We categorize the equilibria into two separate groups. The first group is the \textit{consensus} equilibria where the agents agree:  either all are left-leaning agents or all are right-leaning agents. The second group is the \textit{persistent disagreement} equilibria, which are defined to be all the other equilibria that are not consensus. Agents have different limiting opinions in persistent disagreement equilibria and the degree of disagreement can be thought of as the \textit{polarization} of the population. In Theorem \ref{thm:equal-a-init-char} we provide conditions on the two-agent system's primitives that characterize which type of equilibrium arises as a function of the primitives.

We then analyze  a general stochastic block model.
 In this model there are  two blocks of agents: a block of left-leaning agents and a block of  right-leaning agents. Each block has $n$ agents. Each pair of agents from the same block are connected with probability $p$, and
each pair of agents from different blocks are connected with probability $q$. The initial opinions of agents  in different blocks are assumed to be random variables. 
Stochastic block models are prevalent  in the study of opinion dynamics in networks. 
Having two blocks is also a natural starting point,  considering the bimodal distribution of opinions in various settings of practical interest,  including, e.g., political discourse.

 We focus on a regime with a large number of agents; i.e., $n$ is large. This regime is typical of many social networks where there are many interacting agents. Our main theoretical contribution in this part (see Theorem \ref{Thm:Main})  is to show that when the network is dense enough and the number of agents is large, the two-block  model behaves approximately like the two-agent system. Therefore, the results and insights that we obtain on the two-agent system, apply (approximately) to the general two-block  model. This result on the stochastic block model with two blocks is of independent interest as 
 it is often challenging to theoretically analyze 
 dynamic processes on  large (random) networks.
 Our paper contributes to the nascent literature on 
  analyzing such processes through the  mean-field approximations of such networks/processes.\footnote{The most related work in this direction is \cite{huangMW19-random-network-price-disc}, which shows that price discrimination on \textit{realizations} of the Erdos-Renyi network adds negligible value to the optimal strategy on the mean-field regular graph approximation: uniform pricing. The emerging literature on graphon games \citep{pariseO19-graphon-games} compares equilibria of games on networks with their continuous limits. All of these works study static settings , whereas our work study a dynamic setting.}

Leveraging this approximation result and the in-depth analysis of the two-agent system we derive insights into the effect of both the platform and the agents' initial opinions on opinion dynamics and limiting opinions. 
 We show that when the platform is weak and there is a high number of connections between agents with different initial opinions, a consensus equilibrium is likely as the model behaves in a similar manner to a DeGroot model. In this case, even if a persistent disagreement equilibrium arises, the polarization in this equilibrium, i.e., the degree of  disagreement, is low. When the platform is strong,  a persistent disagreement equilibrium is likely and the equilibrium polarization is high. A moderate platform typically leads to a persistent disagreement equilibrium with moderate polarization. 
Given the level of the platform's influence, we show that there are two types of initial opinions that lead to persistent disagreement: (i) the low-polarization initial opinions, and (ii) the high-polarization initial opinions that are ``balanced'' (that is, the agents have opposite opinions that average out to roughly neutral). The remaining type of initial opinion, the high-polarization and imbalanced type, leads to consensus. For example, if we start with far-right and moderate-left agents, both agents become right in agreement. We also show that the polarization is monotonic as it converges to the limiting opinion's polarization.

We then provide two important extensions to our baseline model.  Firstly, while our baseline model is based on a binary media outlet system, we recognize that real-world scenarios often feature a diverse spectrum of media sources, each bringing its own perspective. Hence, we introduce a model with three media outlets, allowing for a more comprehensive representation of the multitude of influences agents might be exposed to. We introduce a moderation parameter $c$ that determines how the platform classifies agents to right-leaning, moderate and left-leaning agents and derive comparative results regarding this parameter. 
We provide both analytical and numerical results that show that our approximation results continue to hold in this more complex setting with three media outlets. Building on techniques developed for the two-block model, we establish conditions under which the dynamics of a stochastic block model with three blocks can be approximated by a three-agent deterministic system. The extension, however, involves more intricate analysis due to the richer interaction structure among the groups and the platform. 
 We then examine the impact of platform recommendations on shaping final polarization. In particular, we show in numerical simulations that a higher moderation parameter leads to lower final polarization.
Secondly, drawing inspiration from observable patterns in social media platforms, we introduce a correlation between agents' initial opinions and the likelihood of them forming connections. We analyze the effect of this correlation on the final polarization and provide a numerical comparative statics analysis.

Lastly, we address the convergence of opinions to an equilibrium. Convergence of opinions is \textit{a priori} not obvious in our setting; it is conceivable that the platform's influence can alternately push agents to the left and to the right when social learning changes the agents' opinions back and forth to the right and to the left, leading to some form of limit cycles or chaotic dynamics. Nonetheless, we prove that under a symmetric influence matrix, i.e., agent $i$'s influence on agent $j$'s opinion equals agent $j$'s influence on agent $i$'s opinion, the opinions of all agents converge to some limiting opinion (see Theorem \ref{thm:opinions-converge}). To show this, we explicitly construct a  Lyapunov function that   decreases along the trajectory, which  establishes that dynamics converge to an equilibrium. When the influence matrix is not symmetric, we provide a simple numerical example that shows that the agents' opinions do not necessarily converge to an equilibrium. 
Interestingly, for a typical realization of  our stochastic block model, the influence structure is not necessarily symmetric. That said, in expectation the symmetry condition holds. Our results indicate that this fact is  sufficient to show that the  opinions of most agents converge to an arbitrarily small neighborhood of a limiting opinion (characterized in terms of the approximating two-agent system) in the large network regime.

In Section \ref{sec:sim}, we illustrate our results via numerical simulations. We consider two sets of simulations.  The first set is based on a stochastic block model with two blocks. We consider a small stochastic block model with random initial opinions. The second set is based on the political blogs data collected by \cite{adamic2005political}. \cite{adamic2005political} labeled the agents to liberal (left) and conservative (right) and constructed the connections between agents using blog activity. Our simulations suggest that the main theoretical predictions and insights we obtain for large and dense stochastic block models hold also for small stochastic block models and for the political blogs network.  Hence, our simulations show that our results, derived via an approximation to  the two-agent system, are predictive even after substantially relaxing  the assumptions made in the analysis. 

Section \ref{subsec:sim-disagree-polarize} shows that, conditional on persistent disagreement, the mean equilibrium opinion and polarization quantitatively agree with the theoretical predictions of Section \ref{sec:theory} and that polarization is monotonic along the trajectory. 
Section \ref{subsec:sim-balance} analyzes the effect of initial opinions on polarization.
Lastly, Section \ref{subsec:extremism} analyzes the extremism of the limiting opinion as a function of the platform's influence.

\subsection{Related Work}


This paper is closely related to a growing literature on opinion dynamics in social
networks with boundedly rational agents. Most works in this strand of literature \citep{ellisonF93, ellisonF95, dmz03, acemogluDLO11-bayesian-learning, golub2012homophily,  jadbabaieMST12-non-bayesian-social-learning, shi2016evolution, anunrojwongS18-naive-bayesian, grabisch2018strategic} are based on a model of DeGroot \citep{degroot74} and assume that individuals use heuristic rules to update opinions such as taking the weighted average of beliefs. Under standard assumptions, the limiting opinion is a consensus where all agents have the same opinion, which is a weighted average of the agents' initial opinions. Therefore, we cannot meaningfully analyze polarization, persistent disagreement, and extremism in such models. By contrast, in our model, agents generally have different limiting opinions and can also have opinions that are higher or lower than any agent's initial opinion. We can therefore analyze the conditions that lead to an increase (or a decrease) in extremism and polarization. Moreover, standard DeGroot-based models give rise to the same limiting opinion equilibrium independently of the initial opinions whereas our model admits multiple equilibria depending on the initial opinions. We can therefore analyze what distribution of beliefs in a society gives rise to less extremism or less polarization. Lastly, \cite{golubJ10-wisdom-of-crowds} show that homophily is an important network attribute in DeGroot learning. In our model, between-group strength determines the extremism and polarization of the limiting opinion, whereas homophily does not, since it depends also on within-group strength, which is irrelevant in our model.

We also depart from the social learning literature by introducing the platform as a major player in opinion dynamics. The prevalence of the platform in our model  departs from DeGroot opinion dynamics models with stubborn agents (for example, see \cite{mobilia2003does}, \cite{ghaderi2014opinion}, \cite{wai2016active}, and \cite{hunter2022optimizing}) or non-linear models that produce polarization (see \cite{baumann2020modeling} and references therein).   The platform is different from a stubborn agent in the DeGroot dynamics because it influences the agents' opinions not through its own fixed opinion but through personalized content recommendation that differs between agents and can change over time with the evolution of the agents' opinions. \cite{perraR19-scientific-reports} also model opinion dynamics under algorithmic personalization. Our paper differs from theirs in a few ways. First, their focus is on the effect of different recommendation algorithms used by the platform, whereas we assume a simpler platform behavior and focus more on the effects of the strength of the platform, initial opinions, and network structure. Second, they model opinions as binary, whereas we model opinions as scalars, so we can talk not only about agreement but also about whether agents have moderate or extreme beliefs. Lastly, they show all of their results through simulations without theoretical analysis, whereas we theoretically analyze a general stochastic block model with two blocks and also provide insights using simulations.

Our paper is also related to the literature on the spread of information and polarization in social networks. \cite{dandekarGL13-biased-assimilation-pnas} show that when people weight confirming evidence more than disconfirming evidence when they update their beliefs, society becomes polarized, whereas homophily alone cannot increase polarization.  Their key theoretical result is based on the two-island model, a deterministic variant of the two-block stochastic block model. \cite{amelkinBS17} analyze polar opinion dynamics in social networks. As in our work, they formulate the opinion dynamics in continuous time with differential equations and prove convergence with a Lyapunov function and LaSalle's invariance principle. However, they focus only on convergence whereas we mainly focus on analyzing the stochastic block model with two blocks and on properties of the limiting opinion and its dependence on the model's primitives.

There is also a growing body of empirical literature on social media, news consumption, and polarization. \cite{flaxman2016filter}  find that social networks are associated with an increase in the ideological distance between individuals. \cite{levy2021social} finds that Facebook's algorithm is more likely to expose individuals to news matching their ideology using experimental variation. \cite{davisGK20}  investigate how different types of social information affect how firms compete through service quality.  \cite{allcott2020welfare} find that Facebook may increase polarization and \cite{lelkes2017hostile} find that access to the Internet may increase polarization.  \cite{bakshyMA15-ideologically-diverse-facebook} study how much cross-cutting content a
user sees in their Facebok's News Feed and show that agents are more likely to see content that matches their own opinions.  Overall, this empirical literature supports our assumption that the platform shows the agents content that matches their opinions. 

Lastly, our model has no notion of truth, only opinions, so we can talk about extremism and polarization but not about whether agents learn the true state of the world, or about misinformation and ``fake news.'' Nevertheless, the literature on truth and misinformation in social networks is conceptually related to our work. \cite{azzimontiF18-fake-news-polarization-nber} study the effect of internet bots on the spread of fake news in social networks. \cite{blochDK18-rumors-social-networks} analyze the strategic transmission of possibly false information by rational agents who seek the truth. \cite{papanastasiou18-fake-news-sequential} analyzes the propagation of fake news and derives a detection policy. \cite{mostagirOS20} consider a social learning model where the principal can also send signals to agents, and characterize when the principal can get the agents to take an action in the principal's interest. \cite{acemogluMMO19} analyze how the design of rating systems affect whether and how quickly customers can learn about product quality.

\section{Opinion Dynamics}\label{sec:model}



Agents are in a social network, modeled as a graph $G$ with $k \geq 2$ nodes. The influence matrix of this graph is given by a matrix $A$. We interpret the $(i,j)$-th entry $a_{ij} \geq 0$ of $A$ as the strength of agent $j$'s influence on agent $i$.  In this paper we mainly focus on the stochastic block model with two blocks where the network is random and the number of agents is large (see Section \ref{sec:theory}). In this section we start by describing the evolution of the agents' opinions for a given realized influence matrix $A$.

We model each agent $i$'s opinion at time $t \geq 0$ as a scalar $x_i(t) \in \mathbb{R}$. Let $\mathbf{x}(t) = (x_1(t), \dots, x_n(t))^\top \in \R^{n}$ be the vector of all agents' opinions at time $t$. Time is continuous and we denote the vector of initial opinions at $t=0$ by $\mathbf{x}(0)$.

Each agent's opinion is influenced by her neighbors and by a platform. Let $b_i \geq 0$ be the strength of the platform on agent $i$. We assume that the opinions evolve according to the equation
\begin{align}\label{eqn:main-dynamics}
    \frac{d}{dt} x_i(t) \equiv \dot{x}_i(t) = \sum_{j \neq i} a_{ij} (x_j(t) - x_i(t)) + b_i (s(x_i(t)) - x_i(t)),
\end{align}
for some function $s: \R \to \R$, for all $t \geq 0$. We interpret $s(x_i)$ as the opinion of the media outlet shown to agent $i$ by the platform.

Qualitatively, the above opinion dynamics imply that the change in agent $i$'s opinion comes from two sources: social learning and the platform's influence,  which are the first and second terms on the right-hand side, respectively. 
The first term alone constitutes a continuous version of naive (DeGroot-like) social learning, and, in fact,
having only this term recovers a  formulation that is studied in the literature \citep{abelson64,proskurnikovT17-tutorial-dynamic-networks}.\footnote{In the discrete-time DeGroot learning model, for any agent $i$, the next period opinion is a linear combination of the neighboring opinions in the current period, i.e., 
\begin{align*}
    x_i(t+1) = \sum_{j=1}^{n} a_{ij} x_j(t).
\end{align*}
The influence matrix is often assumed to be row-stochastic, i.e., $\sum_{j=1}^{n} a_{ij} = 1$ for each $i$. Substituting $a_{ii} = 1 - \sum_{j \neq i} a_{ij}$ in the above equation we get
\begin{align*}
    x_i(t+1) - x_i(t) = \sum_{j \neq i} a_{ij} (x_j(t) - x_i(t)).
\end{align*}
By interpreting the difference $x_i(t+1) - x_i(t)$ as a time derivative $\dot{x}_i(t)$, we get the continuous-time DeGroot dynamics. }
Note that in \eqref{eqn:main-dynamics}
the term $x_j(t) - x_i(t)$ brings agent $i$'s opinion $x_i$ closer to her neighbor $j$'s opinion $x_j$ because if $x_j(t) > x_i(t)$, the aforementioned term is positive, which contributes positively to the time derivative of $i$'s opinion (and vice versa for $x_j(t) < x_i(t)$). In the subsequent analysis, unless stated otherwise, we allow for  $a_{ij}\geq 0$  to take an arbitrary nonnegative value.\footnote{Thus, there are two important differences relative to the  discrete-time DeGroot dynamics: (i) the term $a_{ii}$ does not appear in the equations describing the continuous-time opinion dynamics,  and (ii) the sum of $\{a_{ij}\}_j$ is no longer normalized to one.} 

Analogously, the second term $b_i(s(x_i) - x_i)$ can be interpreted as the platform's influence. We let $b_i$ be the strength of the platform's influence on agent $i$. When agent $i$'s opinion is $x_i$, the platform shows content that nudges her opinion toward $s(x_i)$, which depends on $x_i$.

The linear update rules for neighbors and for the platform have the same form, so we can think of the platform as a ``neighbor'' to everyone, with an important difference: while each agent has one opinion at any given time and nudges all her neighbors toward that opinion, the platform does not have its own opinion; it simply nudges each agent's opinion toward the content that matches their current opinion.

We assume that all right-leaning media outlets have the same opinion and all left-leaning media outlets have the same opinion. Without loss of generality let the left-leaning and right-leaning outlets' opinions be at $-1$ and $1$, respectively. We further assume that
the platform shows to an agent the content from the media outlet whose opinion is closest to the agent's opinion.
As we discussed in the Introduction, this assumption is motivated by a growing body of evidence that suggests that in order to increase engagement, social media platforms often show  agents content that aligns with their current opinions. Hence,   $s$  can be modeled as the sign function, i.e.,  $s(x_{i})= 1$ if $x_{i}$ is positive and $s(x_{i})= -1$ if $x_{i}$ is negative. This means that the platform shows content from right-leaning (left-leaning) outlets to right-leaning (left-leaning) agents. 
For technical reasons (in particular to ensure that the right-hand side of \eqref{eqn:main-dynamics} is continuous in the opinions), we consider  the following continuous interpolation of the sign function for a small $\epsilon > 0$:
\begin{equation}\label{eqn:sgn-eps-def}
    \sgn_{\epsilon}(x_i) = \begin{cases}
    -1 &\text{ if } x_i \leq -\epsilon, \\
    x_i/\epsilon &\text{ if } -\epsilon \leq x_i \leq \epsilon, \\
    1 &\text{ if } x_i \geq \epsilon.
    \end{cases}
\end{equation}
For the majority of our theoretical results, we will assume that $s(x_i) = \sgn_{\epsilon}(x_i)$. We numerically study a more complicated case that includes a centrist media platform in Section \ref{sec:three}.

Under this functional-form assumption on $s(x_i)$, Equation (\ref{eqn:main-dynamics}) reduces to 
\begin{equation}\label{eqn:main-dynamics-b-original}
    \dot{x}_i(t) = \sum_{j \neq i} a_{ij} (x_j(t) - x_i(t)) + b_i (\sgn_{\epsilon}(x_i(t)) - x_i(t)),
\end{equation}
Note that even though the initial opinions $\mathbf{x}(0) \in \R^{n}$ can be arbitrary scalars, it is easy to see that the limiting opinions $\mathbf{x}(\infty) = \lim_{t \to \infty} \mathbf{x}(t)$ (provided they exist), must be in between
the opinions of the media outlets
($-1$ and $1$) because of averaging that takes place due to the social learning term.\footnote{To see this, suppose that  agent $i$ has the highest opinion across all agents at equilibrium. If agent $i$'s opinion is above $1$, then all the other agents and the platform strictly push agent $i$'s opinion lower contradicting the equilibrium assumption.} Also note that if for some period $t$ all the agents' opinions are positive (negative) then all the agents' limiting opinions should  be exactly $1$ ($-1$) as the platform pushes the agents toward the opinions of the media outlets.

Lastly, in most of the subsequent analysis, our focus will be on understanding  how  the platform influences opinion dynamics. To provide cleaner insights into this question,  we will often assume that $b_i = b$ for all agents $i$. Accordingly, we will write the opinion dynamics as 
\begin{align}\label{eqn:main-dynamics-b}
    \dot{\mathbf{x}} = - L \mathbf{x} + b(\sgn_{\epsilon}(\mathbf{x}) - \mathbf{x}), 
\end{align}
where $L$ is the Laplacian of\footnote{Recall that given a network $A$, the Laplacian $L \equiv L[A]$ of $A$ is given by 
\begin{align*}
    L[A]_{ij} = \begin{cases}
   -a_{ij} &\text{ for } j \neq i, \\
    \sum_{j' \neq i} a_{ij'} &\text{ for } j = i.
  \end{cases}
\end{align*}} $A$.

\subsection{Remarks On The Assumptions}

We now provide a few remarks on our assumptions.

\noindent \textbf{DeGroot social learning model.}
We assume a DeGroot social learning model as opposed to other social learning models such as Bayesian learning. The first reason to study a DeGroot model is that it provides a tractable way to describe learning dynamics. Its simplicity typically allows for mathematical analysis, especially in network settings. Second, the DeGroot model elegantly captures the concept of individuals learning from their peers. Each individual assigns a certain weight to the opinions of their neighbors, and these weights dictate the influence each peer exerts. Third, this model is considered the canonical non-Bayesian social learning model (see \cite{molavi2018theory}). Our goal in this paper is to integrate a platform as a major player within a foundational and tractable social learning model. Hence, the DeGroot model is a reasonable choice for our purposes. Future research exploring whether our insights remain valid in more complex social learning models would be intriguing.

\noindent \textbf{Platform's objective.} In our model we assume that the platform's objective is maximizing engagement. This assumption is grounded  in empirical observations detailed in the related literature section. Many contemporary platforms, especially in the social media domain, rely on user engagement as a primary metric to maximize. Higher engagement often translates to increased screen time, more frequent user visits, and greater user activity, all of which are correlated with higher revenues for the platform. Therefore, it is reasonable to assume that the platform maximizes engagement. 
However, we acknowledge that platforms may have  multifaceted objectives. For instance, a platform might also be concerned about its reputation and reliability especially in the longer term. A platform that is only seen as promoting extreme or polarized views might face backlash from certain user groups or even regulatory bodies. In such scenarios, while short-term engagement might be maximized by reinforcing users' existing beliefs, in the longer term, a platform might prioritize its reputation, which could mean promoting diverse views. We think that a more complicated platform's objective is an  interesting future research direction.

\noindent \textbf{Showing a mix of left and right opinions.} In the model formulation above, we assume that if $x_i \geq \epsilon$, the platform always shows the right-leaning content $s(x_i) = 1$, and if $x_i \leq -\epsilon$, the platform always shows the left-leaning content $s(x_i) = -1$. If, instead, we let the platform show a mix of different content types with bias toward each agent's opinion, we can reinterpret $s(x_i)$ as the average slant of the content shown by the platform, and the same framework still applies. In particular, let $0 \leq \alpha \leq 1$ be a constant, and suppose that the platform shows one unit of content depending on each agent $i$'s opinion $x_i$ as follows. For $x_i \geq \epsilon$, the platform shows $(1+\alpha)/2$ units of right-leaning content $+1$ and $(1-\alpha)/2$ units of left-leaning content $-1$. Similarly, for $x_i \leq \epsilon$, the platform shows $(1-\alpha)/2$ units of right-leaning content $+1$ and $(1+\alpha)/2$ units of left-leaning content $-1$. Then the average content slant $s(x_i)$ is $\alpha$ for $x_i > \epsilon$ and $-\alpha$ for $x_i < -\epsilon$. We can then write $s(x_i) = \alpha \mathrm{sgn}_{\epsilon}(x_i)$ and the opinion dynamics becomes $\dot{\mathbf{x}} = - L \mathbf{x} + B(\alpha \mathrm{sgn}_{\epsilon}(\mathbf{x}) - \mathbf{x})$, where $B=\diag ( { \boldsymbol{b} })$ is a diagonal matrix with entries $b_{i}$. This is the same dynamics as in (\ref{eqn:main-dynamics-b-original}) but scaled by $\alpha$, that is, if we let $\mathbf{x} = \alpha \tilde{\mathbf{x}}$ and $\epsilon = \alpha \tilde{\epsilon}$, the dynamics becomes $\dot{\tilde{\mathbf{x}}} = - L \tilde{\mathbf{x}} + B( \mathrm{sgn}_{\tilde{\epsilon}}(\tilde{\mathbf{x}}) - \tilde{\mathbf{x}})$, which has the same form as (\ref{eqn:main-dynamics-b-original}). Hence   all our results hold for this case as well.  One can consider a more complicated structure for the function $s$ where there are different opinions between the  right-leaning and left-leaning outlets.

\noindent \textbf{Two media outlets. }
We study a model with two media outlets. This choice  allows us to provide a  characterization of the system's dynamics and equilibria as we show in Section \ref{sec:theory}. In addition, in some real world applications, 
 there is noticeable polarization, often split into two major ideological camps. The binary model captures the dynamics in such polarized settings in a stylized way. On the other hand, the two media outlets model comes at some expense of realism. In Section \ref{sec:three} we extend our model and consider a more realistic scenario where the media outlets have more diverse opinions.

\section{Theoretical Analysis} \label{sec:theory}

In this section we present our main theoretical results.  We study a stochastic block model with two blocks, which is commonly employed    in the study of opinion dynamics in networks. 
Agents in different blocks differ in their initial opinions.
Motivated by 
social media platforms that have a large number of agents,
we focus on the case of large networks  with many nodes and obtain asymptotic results. 
More specifically, we show that when the network is large and dense enough, the opinion dynamics for the stochastic block model with two blocks can be approximated by a simple two-agent system (see Theorem \ref{Thm:Main} for a precise statement).   We show that the two-agent system  admits a tractable theoretical  analysis and provide a detailed characterization of the limiting opinions in this system.
Leveraging this characterization, we then provide sharp results and predictions for the general stochastic block model. 

We start by presenting the results for  the two-agent system (see Section \ref{subsec:2-agent}). In Section \ref{sec:main-result} we formally show that this system can be used to approximate and study the stochastic block model with two blocks. 
Lastly, we discuss the convergence of opinions in more general settings (see Section \ref{sec:converge}).

\subsection{The Two-Agent System} \label{subsec:2-agent}

We first analyze the system with   only two agents, referred to as   left and right agents. The agents are connected and symmetric in their influence, so each agent's influence on the other agent is  $a_{12}=a_{21}=a$ for some $a> 0$.  We refer to the quadrants $\{x_1 > 0, x_2 < 0\}$ and $\{x_1 < 0, x_2 > 0\}$ as $(+,-)$ and $(-,+)$ quadrants, and
denote by   $(x_{1}(t),x_{2}(t))$ the agents' opinions at time $t$.
We show that the two-agent system always has an equilibrium. That is,
the limit
$\lim_{t \rightarrow \infty } (x_{1}(t),x_{2}(t)) = (e_{1},e_{2})$ exists, and  $(e_{1},e_{2})$ is called the equilibrium of the dynamical system describing the opinion dynamics. We refer to the equilibria 
 in 
 $(+,-)$ and $(-,+)$ quadrants
 as   $(+,-)$ and $(-,+)$ equilibria, respectively.

We distinguish between two types of equilibria. In the \textit{consensus} equilibria, the agents agree, and the consensus opinion is $+1$ or $-1$, which is equal to one of the media outlets' opinions.\footnote{In any equilibria where all agents' opinions are the same, all the opinions must equal to $1$, $0$ or $-1$. Theorem \ref{thm:equal-a-init-char} shows that for almost all possible model's primitives,  an equilibrium where the opinion of each agent is $0$ does not arise. Hence, consensus equilibria refer to those  equilibria where all agents have the same opinion.} We call all other non-consensus equilibria \textit{persistent disagreement} equilibria because the agents  do not have the same opinion in the limit.\footnote{This is in contrast to most naive and DeGroot learning models where opinions converge to an equilibrium where all the agents have the same opinion. Agents' opinions in a persistent disagreement equilibrium are also generally not equal to the media outlets' opinions, meaning that they are bounded away from $-1$ and $+1$.   } In addition, in any  persistent disagreement equilibria, one agent is a left agent and the other agent is a right agent. This follows because if both agents have positive (negative) opinions at some period $t$, then the  consensus equilibrium $(1,1)$ ($(-1,-1)$) would arise.

The main result of this section is the characterization of conditions on the parameters that describe the two-agent system $(a, b, \mathbf{x}(0))$ such that the trajectory converges to a persistent disagreement equilibrium or a consensus equilibrium. For ease of exposition, we first state the conditions.

\begin{definition}\label{def:cond-pd-ec}
Let  $a,b>0$ and $\mathbf{x}(0) = (x_{1}(0),x_{2}(0)) \in \R^2$ be such that $x_1(0) x_2(0) \neq 0$.
\begin{itemize}
    \item $(a,b,\mathbf{x}(0))$ satisfies {\em condition PD} if $x_1(0)x_2(0) < 0$ and at least one of C1 or C2 below holds
\begin{itemize}
    \item[C1.] $(2a+b) \left| x_1(0) - x_2(0) \right| - 2b < b|x_1(0)+x_2(0)|$,
    \item[C2.] $(2a+b)|x_1(0)-x_2(0)|-2b > b^{1-2a/b} a^{2a/b} |x_1(0)+x_2(0)|^{1+2a/b}$.

\end{itemize}
    \item $(a,b,\mathbf{x}(0))$ satisfies {\em condition CO} if either $x_1(0)x_2(0) > 0$ or $$b|x_1(0)+x_2(0)| < (2a+b)|x_1(0)-x_2(0)|-2b < b^{1-2a/b} a^{2a/b} |x_1(0)+x_2(0)|^{1+2a/b}.$$
\end{itemize}

\end{definition}

Theorem~\ref{thm:equal-a-init-char} shows that trajectories with conditions PD and CO converge to persistent disagreement and consensus equilibria, respectively. Note that conditions PD and CO together cover all possible values of $(a,b,\mathbf{x}(0))$ except where the above inequalities are replaced by equalities. The theorem, therefore, gives an essentially complete characterization of the two-agent system.

\begin{theorem}\label{thm:equal-a-init-char}

Consider a two-agent system $(a,b,\mathbf{x}(0))$; then the dynamical system is
\begin{align*}
    \dot{x}_1 &= a(x_2-x_1) + b(\textup{sgn}_{\epsilon}(x_1) - x_1) \\
    \dot{x}_2 &= a(x_1-x_2) + b(\textup{sgn}_{\epsilon}(x_2) - x_2).
\end{align*}
 If $(a, b, \mathbf{x}(0))$ satisfies condition PD,
then the trajectory converges to a persistent disagreement equilibrium for all sufficiently small\footnote{When we say ``for all sufficiently small $\epsilon > 0$,"  we mean that there exists $\delta > 0$ such that the statement holds for all $\epsilon \in (0, \delta)$ (note that $\delta$ can depend on the parameters $(a,b,\mathbf{x}(0))$). } $\epsilon > 0$. If $(a, b, \mathbf{x}(0))$ satisfies condition CO, then the trajectory converges to a consensus equilibrium for all sufficiently small $\epsilon > 0$.

\end{theorem}


The proof sketch is as follows. If the starting opinion is in the $(+,+)$ or $(-,-)$ quadrants, the trajectory will lead to a consensus at $(1,1)$ or $(-1,-1)$. We can therefore focus the characterization effort on the $(-,+)$ and the $(+,-)$ quadrant. By symmetry, and without loss of generality, we focus on the $(-,+)$ quadrant. In the ``interior'' region $x_1 \leq -\epsilon, x_2 \geq \epsilon$ the system is linear with $\sgn_{\epsilon}(x_1) = -1, \sgn_{\epsilon}(x_2) = 1$ so we can explicitly compute the closed form of the trajectory. If the trajectory is always confined in this interior $(-,+)$ region, then we can use linear ODE theory to conclude that the trajectory converges to an equilibrium inside this region: a persistent disagreement equilibrium. If the trajectory is not always confined in this region it will eventually hit the ``$\epsilon$-band'' that surrounds the region and the axis; without loss of generality this is the point where $x_1=-\epsilon$. Inside this band with $-\epsilon \leq x_1 \leq \epsilon$, $x_2 \geq \epsilon$, we have a different linear ODE with $\sgn_{\epsilon}(x_1) = x_1/\epsilon, \sgn_{\epsilon}(x_2) = 1$. We then analyze the behavior of the trajectory once it enters the band. 

We can show that once the trajectory enters the band from the $(-,+)$ interior region, it will always cross the band to the adjacent $(+,+)$ interior region and hence converges to consensus at $(1,1)$. We conclude that if the trajectory determined by the linear system in $(-,+)$ is always bounded away from the two axes $x_1=0$ and $x_2=0$, then it will converge to a persistent disagreement, and if it is not, then it will converge to a consensus. We have a closed-form formula for the above trajectory, and we can use this formula to show that the conditions for persistent disagreement and consensus correspond to conditions PD and CO in Definition~\ref{def:cond-pd-ec} respectively, and we are done.

We observe that the conditions in Theorem \ref{thm:equal-a-init-char}  depend only on $b/a$ and $\mathbf{x}(0)$. This is expected, because we can rescale the solution along with the time by a factor of $a$. To facilitate an understanding of Theorem \ref{thm:equal-a-init-char}, we  plot the region of all initial opinions that lead to persistent disagreement as a function of $b/a$. See Figure~\ref{fig:plot-init-disagree} for plots when $b/a=0.5$ and $b/a=1.0$. We see that there are two different regions of initial opinions that lead to persistent disagreement: the green region corresponding to the first condition C1 and the blue region corresponding to the second condition C2. 

For the rest of the section we assume that the initial opinions have different signs, i.e., $x_{1}(0)x_{2}(0) < 0$. As mentioned above this is the interesting case in our model as if the signs are the same the equilibrium would be a consensus equilibrium. 

In the two-agent system we can define \textit{polarization} at time $t$ by $|x_2(t)-x_1(t)|$, and \textit{imbalance} at time $t$ by $|x_1(t)+x_2(t)|$. Let $\mu^* = b/(2a+b)$ and $p^* = 2\mu^* = 2b/(2a+b)$. As described in the proof of Theorem \ref{thm:equal-a-init-char}, the persistent disagreement equilibria are $\left(  \mu^*, - \mu^* \right)$ and $\left(  -\mu^*, \mu^* \right)$, and the polarization in these equilibria is $p^*$. Intuitively, initial opinions satisfying condition C1 have low initial polarization when compared to the imbalance. In particular, the initial polarization is lower than the persistent disagreement equilibrium polarization, i.e., $|x_1(0)-x_2(0)| < p^* = 2b/(2a+b)$. 
Initial opinions satisfying condition C2 have high polarization but are ``balanced'' (imbalance is small relative to initial polarization). When the average opinion of the left and right agents is close to the neutral opinion of $0$, the imbalance measure is small, and when their average opinion is far from $0$, say,  left-leaning with $x_1$ slightly negative and $x_2$ positive and large, the imbalance measure is high. Both the green and blue regions are symmetric with respect to the axis $x_1(0)=-x_2(0)$, where opinions are perfectly balanced. In particular, when opinions are perfectly balanced, i.e., $x_1(0) = -x^*$ and $x_2(0) = x^*$ for some $x^*$, we always have a persistent disagreement equilibrium. Overall,   low-polarization initial opinions, and high-polarization initial opinions that are ``balanced'' lead to persistent disagreement. We show in Theorem \ref{Thm:Main} that this insight also further generalizes to the  stochastic block model with two blocks for large and dense networks. 

For a fixed influence parameter $a$ and initial opinions, we now describe the effect of the platform's influence on the limiting opinions. When the platform's    influence is weak ($b$ is small),  the polarization in the persistent disagreement equilibria $(\pm \mu^*, \mp \mu^*)$ is close to $0$ because $\mu^* = b/(2a+b)$ is small, so the difference   between the agents' opinions is small conditional on persistent disagreement. In addition, the consensus equilibria are likely under a weak platform for a wide range of initial opinions that are not too polarized. When the platform's influence is strong ($b$ is large), the  polarization in the persistent disagreement equilibria $(\pm \mu^*, \mp \mu^*)$ is large because  $\mu^* = b/(2a+b)$ can be large, so the difference between the agents' opinions is large   conditional on persistent disagreement.  In addition,  persistent disagreement equilibria are likely under a strong platform for a wide range of initial opinions that are not too imbalanced.

To conclude, either when the platform is very weak or very strong, the society becomes more ``extreme" in the sense that opinions are further away from the origin, but for different reasons. When the platform is weak, all the agents' opinions are the same as that of a single media outlet and we have a consensus equilibrium.  When the platform is strong, we have a disagreement between the agents, one agent is left-leaning and the other agent is right-leaning and polarization, i.e., the degree of  disagreement is high.  Arguably, both situations are undesirable from a societal perspective. By contrast, when the platform is between weak and strong, we get  ``moderate disagreement'': agents typically disagree in the limit, but not by much (so polarization is low). We consider this a good outcome as diversity of opinions  typically fosters healthy debate and civic engagement when the degree of disagreement is not too high.

We illustrate the insights we obtain via simulations (see Section \ref{sec:sim}) and  summarize them as the following takeaways. 
\begin{itemize}
\item [\textbf{Takeaways}]
    \item Assume $a = 1$. When the platform is weak (i.e., $b$ is small such that condition CO holds), the system admits a consensus equilibrium. In this regime, even if the PD condition is also satisfied, the resulting equilibrium exhibits low polarization, specifically $2b/(2 + b)$.

In contrast, when the platform is strong (i.e., $b$ is large), the system admits a persistent disagreement equilibrium with a higher polarization level of $2b/(2 + b)$.

    \item When the polarization of the initial opinions, $|x_{1} - x_{0}|$, is small relative to the imbalance $|x_{1} + x_{0}|$ so condition C1 is satisfied then the system converges to a persistent disagreement equilibrium. 

    \item When the imbalance $|x_{1} + x_{0}|$ is small relative to the polarization $|x_{1} - x_{0}|$, so that condition C2 is satisfied,  the system again converges to a persistent disagreement equilibrium.

\end{itemize}

Lastly, as we mentioned above, the amount of polarization in a persistent disagreement equilibrium depends only on the platform's influence and the agents'  influence and not on the initial opinions (as long as they generate a persistent disagreement equilibrium). Therefore, if the initial opinion polarization is lower than the persistent disagreement polarization, the final polarization is necessarily higher than the initial polarization, whereas if the initial opinion polarization is higher than the persistent disagreement polarization, the final polarization is lower than the initial polarization. The next proposition shows that a stronger statement is true: \textit{polarization is monotonic}. If the initial polarization is lower than the persistent disagreement polarization, polarization is always increasing, whereas if the initial polarization is higher than the persistent disagreement polarization, polarization is always decreasing. Under the conditions of Theorem \ref{Thm:Main}, a similar version of this insight also generalizes to the stochastic block model (see Corollary \ref{Cor:polarization} for a precise statement). 

\begin{proposition}\label{prop:polar-monotonic}
Consider a two-agent system $(a,b,\boldsymbol{x}(0))$ that satisfies the PD condition. 
If $|x_1(0)-x_2(0)| < 2b/(2a+b)$ then the polarization $|x_1(t)-x_2(t)|$ is increasing in $t$. If $|x_1(0)-x_2(0)| > 2b/(2a+b)$ then the polarization $|x_1(t)-x_2(t)|$ is decreasing in $t$. 
\end{proposition}

\begin{figure}
\begin{subfigure}{.5\textwidth}
  \centering
  \includegraphics[width=1.0\linewidth]{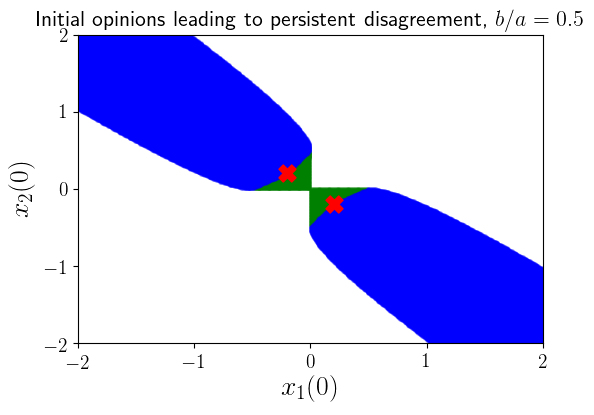}
  \subcaption{$b/a=0.5$}
  \label{fig:sfig1}
\end{subfigure}%
\hfill
\begin{subfigure}{.5\textwidth}
  \centering
  \includegraphics[width=1.0\linewidth]{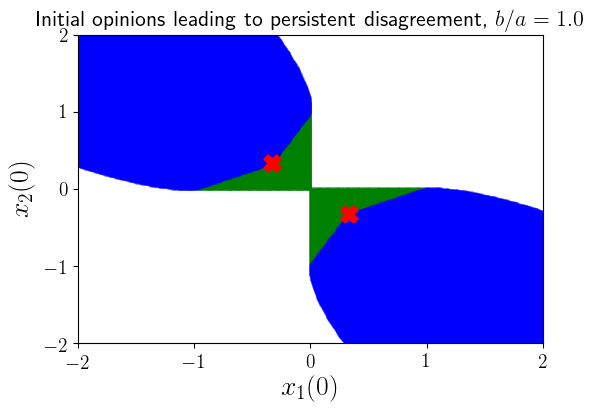}
  \subcaption{$b/a=1.0$}
  \label{fig:sfig2}
\end{subfigure}
\caption{Initial opinions $\mathbf{x}(0) = (x_1(0),x_2(0))$ that lead to persistent disagreement when $b/a=0.5$ and $b/a=1.0$. The red crosses represent the persistent disagreement equilibria. The green and blue regions correspond to condition C1 (low initial polarization) and condition C2 (balanced initial opinion) respectively. The white regions correspond to consensus equilibrium. When the initial polarization is lower (higher) than persistent disagreement polarization, polarization increases (decreases) as the opinion trajectory converges to the equilibrium at the red cross.  }
\label{fig:plot-init-disagree}
\end{figure}

\subsection{The Stochastic Block Model with Two Blocks}\label{sec:main-result}

In this section  we study the stochastic block model with two blocks and state our main theorem. We show that the two-agent system that we analyzed in depth in the last section approximates the stochastic block model with two blocks when the number of agents is large and the network is dense enough.

In the stochastic block model with two blocks there are two sets of agents: A set of left agents and a set of right agents. Each set consists of $n$ agents. The probability of a connection between the same type of agents, i.e.,  a connection between left and left agents and a connection between right with right agents, is given by\footnote{All our results can be extended to the case where a connection between left and left agents has a probability of $p_{L}(n)$ and a connection between right and right agents has a probability of $p_{R}(n)$ where $p_{L}(n) \neq p_{R}(n)$. To ease exposition, we assume that these probabilities are the same.} $p(n)$ and the probability of a connection between different types of agents, i.e., a connection between left and right agents, is given by $q(n)$. We denote by $A_{n}$  a typical adjacency matrix,  
 where an adjacency matrix satisfies $A_{n,ij}=1$ if there is a connection between agent $i$ and agent $j$ and zero otherwise, and by $\boldsymbol{A}_{n}$ the set of all possible adjacency matrices when there are $n$ left agents and $n$ right agents. We denote by $\mathbb{P}_{n}(A_{n})$ the probability that a matrix $A_{n}$ is realized. With slight abuse of notation we denote $\mathbb{P}_{n}$ by $\mathbb{P}$.  
 The initial opinion of each left agent is a random variable on $X_{L} \subset \mathbb{R}_{-} = (-\infty,0]$ and the initial opinion of each right agent is a random variable on $X_{R} \subset \mathbb{R}_{+} = [0,\infty)$. Let $\mu_{0}$ be the probability measure that describes the initial opinions' distribution where we again omit the index $n$ for notational convenience. For simplicity we assume that the support of $\mu_{0}$ is given by a finite set $\boldsymbol{X}_{0}^{n}$  in $\mathbb{R}^{2n}$ where $\boldsymbol{X}_{0} = X_{L} \times X_{R}$. We denote by  $(\boldsymbol{A}_{n} \times \boldsymbol{X}_{0}^{n},2^{ \boldsymbol{A}_{n}} \times 2^{\boldsymbol{X}_{0}^{n}},\mathbb{P} \otimes \mu_{0})$ the probability space that is generated by the stochastic block model with two blocks described above where $\otimes$ denotes the product measure between two probability measures and by $\mathbb{E}$ the expectation operator with respect to the probability measure $\mathbb{P} \times \mu_{0}$ and $2^{D}$ is the set of all subsets of a set $D$. Note that we assume in this section that the initial opinions are independent of the adjacency matrices. We consider the case where they are dependent in Section \ref{sec:correlated}.

 Given a realized adjacency matrix $A$ and a realized initial opinion vector $\boldsymbol{x}_{0}$, each agent's opinion is dynamically evolving and is influenced by the platform and by the opinions of her connections as we described in Section \ref{sec:model} (see Equation (\ref{eqn:main-dynamics-b})).  In order to study a mean field regime where the number of agents tends to infinity we normalize the influence matrix.   We assume that if agent $i$ has $d_{i}$ connections then each connection has an influence weight of $a/d_{i}$, i.e.,  $a_{ij} = a/d_{i}$ is the influence of agent $j$ on agent $i$ if agent $j$ is a connection of agent $i$ and $0$ otherwise. 
 Hence, the opinions of agent $i$ evolve according to the equation
    \[
    \dot{x}_{i,n}(t) = \frac{a}{|N(i)|}\sum_{j\in N(i)} (x_{j,n}(t)-x_{i,n}(t)) +b (sgn_\epsilon(x_{i,n}(t))-x_{i,n}(t))
    \]
where $N(i)$ is the set of agents that are connected to agent $i$ and $|N(i)|$ is the cardinality of $N(i)$, i.e., the total number of connections that agent $i$ has.\footnote{If agent $i$ has no connections, i.e., $N(i)$ is an empty set then the other agents do not influence agent $i$'s opinions.} Note that the opinion of agent $i$ in some period $t$ is a random variable that depends on the realized adjacency matrix and realized initial opinions.

 We now state the main result of this subsection, Theorem \ref{Thm:Main}. Informally, this result shows that in dense networks as the number of agents approaches to infinity, over time
the opinion of every left (right) agent
approaches to the equilibrium opinion of the left (right) agent in a corresponding  two-agent system under certain conditions on the supports of agents' initial opinions.  To prove this theorem, we leverage the comprehensive analysis presented in Theorem \ref{thm:equal-a-init-char} that characterizes the possible equilibrium that can emerge in this system as a function of the agents' initial opinions.\footnote{\cite{golub2012homophily} provide a ``representative-agent" theorem that allows them to analyze the convergence of a multi-type random network by studying a
 smaller network in which there is only one node for each type of agent.  In contrast to our setting, in \cite{golub2012homophily} there is no platform and the dynamical system is linear so techniques from random matrices theory can be leveraged to prove their  representative-agent theorem that don't apply to our setting. } 
More precisely, this claim holds 
for any adjacency matrix belonging to a   set of adjacency matrices that approaches probability that is tending to  $1$ when $n$ tends to infinity. We now introduce notations that are used to state Theorem \ref{Thm:Main}. 

Recall that for two functions $f,g$ from the set of positive integers to the set of positive numbers, we say that $f(n)$ is $\omega(g(n))$ if for every $c>0$, there exists $n_{0} \geq 1$ such that $f(n) > cg(n)$ for all $n \geq n_{0}$. Let $X_{PD}(a,b)$   be the set in $\mathbb{R}^{2}$ that includes the initial opinions that satisfy the PD condition (see Definition \ref{def:cond-pd-ec}) for $a,b>0$. Similarly, let $X_{COP} (a,b)$  ($X_{CON}(a,b)$) be the set  in $\mathbb{R}^{2}$ that includes the initial opinions that satisfy the CO condition and leads to a positive, i.e., $(e_{L},e_{R})=(1,1)$, (negative, i.e., $(e_{L},e_{R})=(-1,-1)$) consensus equilibrium (see Definition \ref{def:cond-pd-ec} and the proof of Theorem \ref{thm:equal-a-init-char} for an exact characterization of these sets). 

We denote by $x_{i,n,A_{n},\boldsymbol{x}_{0} }(t) $ the opinion of agent $i$ when the realized adjacency matrix is $A_{n}$ and the realized initial opinions vector is $\boldsymbol{x}_{0}$. 

\begin{theorem} \label{Thm:Main}
Suppose that $p(n)$ and $q(n)$ are $\omega(\ln (n) / n)$. Assume that $\epsilon$ is small so the characterization of the two agent system  given in Theorem \ref{thm:equal-a-init-char} holds.  Let  $\beta = \lim_{n \rightarrow \infty} q(n)/(q(n)+p(n)) $.

 Consider a two agent system $(a \beta ,b, (x_{L},x_{R}) )$ that satisfies the PD  condition with an equilibrium $\boldsymbol{e} = (e_{L},e_{R})$.  If $X_{L} \times X_{R} \subseteq X_{PD}(a\beta,b)$,  
 then for every $\epsilon' >0$, there exist $\delta>0$, $N>0$, $T>0$, and sets $C_{\delta,n} \subset \boldsymbol{A}_{n}$ such that  $\lim  _{n \rightarrow \infty} \mathbb{P}(C_{\delta,n}) = 1$ and for all $n \geq N$, all $t \geq T$, all $A_{n} \in C_{\delta,n}$, all $\boldsymbol{x}_{0} \in (X_{L} \times X_{R})^{n}$, we have 
 \begin{equation}\label{Eq:MainThmStatement}
|  x_{i,n,A_{n},\boldsymbol{x}_{0} }(t) - e_{L}| \leq \epsilon' \text{ and } |  x_{j,n,A_{n},\boldsymbol{x}_{0}}(t) - e_{R}| \leq \epsilon'
\end{equation}
 for every left agent $i$ and every right agent $j$. 

 Similarly, when $(a \beta ,b, (x_{L},x_{R}) )$ satisfies the CO condition with equilibrium $(1,1)$ ($(-1,-1)$) then Inequality (\ref{Eq:MainThmStatement}) holds  when $X_{L} \times X_{R} \subseteq X_{COP}(a\beta,b)$ ($X_{L} \times X_{R} \subseteq X_{CON}(a\beta,b)$).

\end{theorem}

When the numbers of agents is large, the influence of  an agent on her connections is small with high probability. Hence, Theorem \ref{Thm:Main} can be seen as a mean-field approximation for the opinion dynamics that are generated by the stochastic block model. This theorem shows that the agents' opinions in this general stochastic block  model can be approximated by the agents' opinions in the  much simpler two-agent system that we solve analytically in  Section \ref{subsec:2-agent}.  


To gain some intuition for the proof of Theorem \ref{Thm:Main}, note that the expected number of right (left) connections that a left (right) agent has is $q(n)n$ and the expected total number of connections is $q(n)n + p(n)n$.   When the network is dense enough, ($p(n)$ and $q(n)$ are $\omega(\ln (n) / n)$) and $n$ is large; then, using a concentration argument, we show that all agents have $q(n)n$ plus a small deviation error connections of a different type and  $p(n)n+q(n)n$ plus a small deviation error total connections with high probability. The key step in proving Theorem \ref{Thm:Main} is to construct auxiliary   differential equations to provide bounds on the trajectories of left and right agents that have close to $q(n)n$ connections of a different type and $q(n)n + p(n)n$ total connections. These bounds imply that for large $t$, the opinions of all   left (right) agents are  close to the equilibrium opinion of a single left (right) representative agent. Hence, the opinion of each left (right) agent is approximately not influenced by the other left (right) agents and is approximately influenced   by the single representative right (left) agent
at a rate of $aq(n)n/(p(n)n+q(n)n) \approx a \beta$. Intuitively, this is the influence parameter because $p(n)n+q(n)n$ is the expected total number of connections and $q(n)n$ is the expected number of connections of a different type an agent has. Thus, the stochastic block model with two blocks is approximated by a two-agent system for large $n$ and large $t$.  

We assume that $p(n)$ and $q(n)$ are $\omega(\ln (n) / n)$ which implies that the network is connected. To guarantee connectedness, it is enough to assume that  $p(n)$ and $q(n)$ are $\Omega(\ln (n) / n)$  (see Chapter 4 in  \cite{blum2020foundations}). We note that Theorem \ref{Thm:Main} can be extended for the case that $p(n)$ and $q(n)$ are $\Omega(\ln (n) / n)$. Suppose that $p(n)=c_{p}\ln(n)/n$ and $q(n)= c_{q}\ln(n)/n$. Then an inspection of the proof of Theorem \ref{Thm:Main} shows that we can choose $\delta = d/\min \{c_{q},c_{p} \}^{1/2}$ for some $d>0$ and the proof still follows so  the approximation result holds for high enough $c_{p}$ and $c_{q}$.

When all the left (right) agents start with the same initial opinion, the proof of Theorem \ref{Thm:Main} shows that, with high probability, for every finite $t$, the opinion of every left (right) agent in the stochastic block model with two blocks is arbitrarily close to the left (right) agent's opinion in the $(a\beta,b,x_{L},x_{R})$ two-agent system,  respectively. Hence, in this case, the insights and analysis we provided in Section \ref{subsec:2-agent} for the opinions' trajectories of the simple two-agent system apply  also to the stochastic block model with two blocks. For example,  the analysis that described  the platform's influence on the convergence of opinions and polarization (see the discussion in Section \ref{subsec:2-agent}) carries over to  the stochastic block model with two blocks. The next corollary that shows that polarization is monotonic extends Proposition \ref{prop:polar-monotonic} to the stochastic block model.

\begin{corollary} \label{Cor:polarization}
Consider a two-agent system $(a \beta ,b, (x_{L},x_{R}) )$ that satisfies the PD condition with a solution $(x_{1}(t),x_{2}(t))$ and suppose that the Assumptions of Theorem \ref{Thm:Main} hold with $X_{L}= \{x_{L} \}$ and $X_{R} = \{ x_{R} \} $. Assume WLOG that $x_{L} > x_{R}$.  
 
Let $t_{2} > t_{1} > 0$, $t>0$, and $\delta' >0$. If $x_L(0)-x_R(0) <  2b/(2a\beta+b)$ then there exists a $\delta>0$, and sets $C_{\delta,n} \subset \boldsymbol{A}_{n}$ such that  $\lim  _{n \rightarrow \infty} \mathbb{P}(C_{\delta,n}) = 1$ and 

(i) $$|x_{i,n,A_{n} } (t) - x_{1}(t)| \leq \delta' \text { and } |x_{j,n,A_{n} } (t) - x_{2}(t)| \leq \delta'$$

(ii) 
$$x_{i,n,A_{n} } (t_{2}) -  x_{j,n,A_{n} }(t_{2}) \geq x_{i,n,A_{n} } (t_{1}) -  x_{j,n,A_{n} }(t_{1}) $$
 for every left agent $i$, every right agent $j$, and all $A_{n} \in C_{\delta,n}$,  where $x_{i,n,A_{n}}(t)$ is agent $i$'s opinion  at time $t$ when $A_{n} \in C_{\delta,n}$ is the realized adjacency matrix. Hence, the polarization between any left and any right agent is increasing. The opposite result holds, i.e., polarization is decreasing, for the case that $x_L-x_R > 2b/(2a\beta+b)$.
\end{corollary}

An important parameter in analyzing the stochastic block model that is absent in the two-agent system is $\beta$. Intuitively, a large (small)  $\beta$ means that the number of connections between different types of agents is large (small) relative to the number of connections between the same types of agents. In terms of analysis,  $1/\beta$ plays a similar role to that of  the platform's influence parameter $b$. Hence, for example, a small $\beta$ typically leads to persistent disagreement equilibrium with high polarization.  This is intuitive because when the number of connections between different types of agents is relatively small, we would expect  polarization to be higher.

We note that Theorem \ref{Thm:Main} does not imply that the opinions converge for a given realized adjacency matrix $A$. In Section \ref{sec:converge} we shed light on convergence of opinions for a fixed adjacency matrix.

\subsection{Three Media Outlets} \label{sec:three}
In the original model in Section \ref{sec:model} we study a binary framework where there are two media outlets. This choice simplified our analysis and allowed us to characterize the system's dynamics and equilibria. In this section we consider a more general setting where the media outlets have more diverse opinions. 

 We extend our baseline model to incorporate a third, centrist, media outlet. This outlet neither leans right nor left. This introduction brings about new dynamics, complexities, and richer behaviors to our model. As before, we assume that the platform exposes individuals to news matching their own opinions, that is, we assume that the platform shows content from right-leaning (left-leaning) outlets to right-leaning (left-leaning) agents and content from the centrist media outlet to moderate agents. Hence, the platform has to classify the moderate agents. We introduce a moderation parameter $c$ where a larger value of 
$c$ indicates that the platform considers a broader range of opinions as moderate. We extend the sign function introduced in Section \ref{sec:model} to capture this. For technical reasons (see Section \ref{sec:model}), we consider  the following continuous interpolation of a step function for a small $\epsilon > 0$:

\begin{equation}\label{eqn:sgn-eps-c-def}
    \sgn_{\epsilon,c}(x_i) = \begin{cases}
    -1 &\text{ if } x_i \leq -c -\epsilon, \\
    (x_i+c)/\epsilon &\text{ if }  -c - \epsilon  \leq x_i \leq -c, \\
    0 &\text{ if } -c \leq x_i \leq c \\
    (x_i-c)/\epsilon &\text{ if }  c \leq x_i \leq  c+ \epsilon \\
    1 &\text{ if } x_i \geq c+ \epsilon.
    \end{cases}
\end{equation}
Hence, the platform shows the content from the centrist media outlet to moderate agents that are defined by agents with opinion between $-c$ and $c$ where $c$ is the moderation parameter. Under this functional-form assumption, we study the social dynamics generated by Equation (\ref{eqn:main-dynamics}). 
The opinions evolve according to $ \dot{x}_i(t) = \sum_{j \neq i} a_{ij} (x_j(t) - x_i(t)) + b (\sgn_{\epsilon,c}(x_i(t)) - x_i(t))$. We denote by $R$ the right agents with opinions above $c+\epsilon$, $M$ the moderate agents with opinions between $-c$ and $c$ and by $L$ the left agents with opinions below $-c-\epsilon$. 

A first question we address is whether a version of Theorem~\ref{Thm:Main} also holds for the three-outlet model. Our numerical analysis suggests that the mean-field approximation result presented in Section~\ref{sec:main-result} continues to apply in the three-outlet setting across various model specifications.
 To illustrate this, consider a three block stochastic model (3BSM) with $n$ agents, consisting of $n/3$ left-leaning agents, $n/3$ moderate agents, and $n/3$ right-leaning agents. The initial opinions of left agents are drawn uniformly from $[-1, c]$, those of moderate agents from $[-c, c]$, and those of right agents from $[c, 1]$. We assume that the probability of an in-block connection (i.e., between agents within the same group) is $1/4$, while the probability of an out-of-block connection is $1/8$.

\begin{figure}[H]
  \centering

  \begin{subfigure}[b]{0.4\linewidth}
    \includegraphics[width=\linewidth]{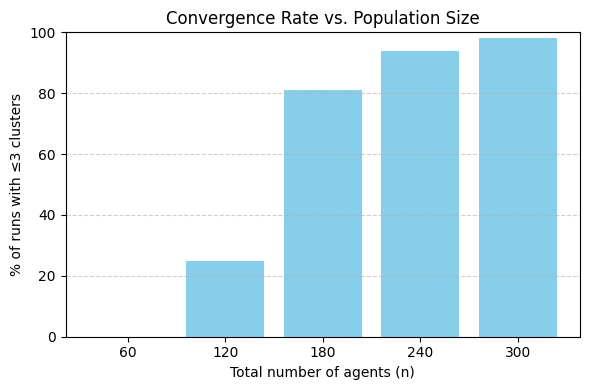}
    \label{fig:sample1}
  \end{subfigure}
  \hfill
  \begin{subfigure}[b]{0.4\linewidth}
    \includegraphics[width=\linewidth]{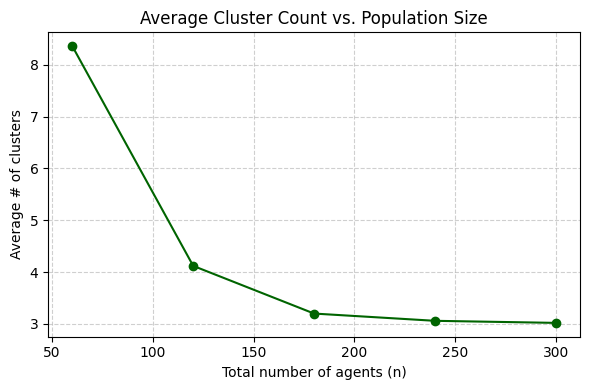}
    \label{fig:sample2}
  \end{subfigure}

  \vspace{0.3cm}

  \begin{subfigure}[b]{0.4\linewidth}
    \includegraphics[width=\linewidth]{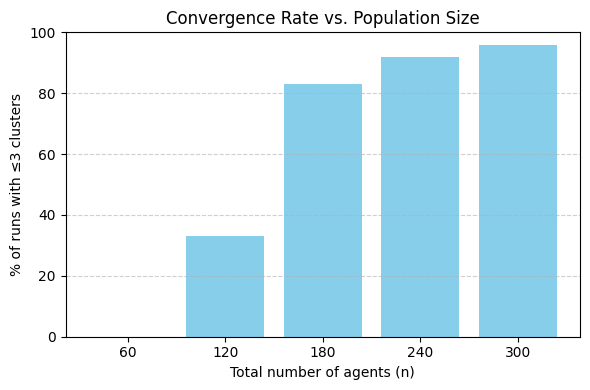}
    \label{fig:sample3}
  \end{subfigure}
  \hfill
  \begin{subfigure}[b]{0.4\linewidth}
    \includegraphics[width=\linewidth]{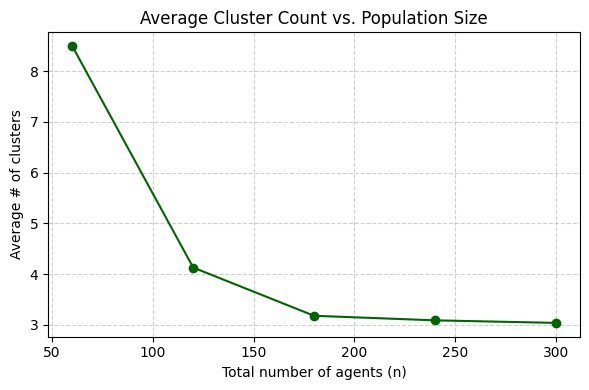}
    \label{fig:sample4}
  \end{subfigure}

\caption{
For each value of $n \in \{60, 120, 180, 240, 300\}$, we simulate the three block stochastic model 100 times with parameters $a = 1$, $b = 1$, and $\epsilon = 0.0001$. The two top figures use $c=0.2$ and two bottom figures use $c=0.35$. Clusters in final opinions are identified using mean-shift clustering with bandwidth $0.05$. The left panel shows the number of simulations (out of 100) that resulted in 3 or fewer final opinion clusters. We observe that this number significantly increases as $n$ increases, suggesting that opinions are more likely to form 3 or less distinct clusters in larger networks. The right panel reports the average number of clusters across the 100 runs. These figures show that as the network size increases, final opinions tend to cluster into three distinct groups.
}
  \label{fig:4samples}
\end{figure}

The numerical results in Figure \ref{fig:4samples} show that as the number of agents in each block increases from 60 to 300, opinions increasingly cluster into three distinct groups under the parameters we consider. This behavior suggests the dynamics may be well approximated by a three-agent system, similar in spirit to the mean field approximation we analytically established in the two-block setting. This empirical observation motivated us to undertake a theoretical analysis of the stochastic block model with three blocks. While we build on techniques developed for the two-agent system, the extension to three outlets introduces substantial additional complexity; see Appendix~\ref{Sec:Appendix3BSM} for further details. Nevertheless, we characterize conditions under which a form of persistent disagreement emerges, namely, convergence to three distinct opinion clusters, where left, moderate, and right agents each concentrate around separate points.
 To establish this, in Theorem~\ref{Thm:3-agentsystem} in Appendix \ref{Sec:Appendix3BSM}, we prove an analogue of Theorem~\ref{thm:equal-a-init-char} for the case of persistent disagreement in a deterministic three-agent system, which serves as a candidate approximation for the stochastic block model with three blocks. We then derive a mean-field approximation result in Theorem~\ref{Theorem: 3BSMmain}, analogous to Theorem~\ref{Thm:Main}, showing that the behavior of the stochastic block model with three blocks can be captured by the deterministic three-agent system analyzed in Theorem~\ref{Thm:3-agentsystem}. Regarding the platform's influence, we note that similar intuition from the two-agent system holds. For example, when the platform is weak (i.e., 
$b$ is small), the conditions for persistent disagreement are harder to satisfy. And even when such disagreement does arise, the resulting polarization is low. In contrast, when the platform is strong (i.e., 
$b$ is large), the system tends to exhibit persistent disagreement with high polarization. The effect of the moderation parameter $c$ on the emergence of persistent disagreement is more complicated. When $c$ is very small, the model effectively behaves like a two-agent system. In such cases, disagreement can emerge between two camps, similar to the two-agent model. On the other hand, when $c$ is relatively large, the platform tends to show moderate content to most agents, which can lead to convergence toward consensus or collapse to a single dominant opinion. As a result, the condition for persistent disagreement is generally not monotone in $c$ as illustrated in Figure \ref{fig:PD3_region}. 
We provide the formal details and proofs of Theorems \ref{Thm:3-agentsystem} and \ref{Theorem: 3BSMmain} in Appendix \ref{Sec:Appendix3BSM}.

\begin{figure}[htbp]
    \centering
    \includegraphics[width=0.55\textwidth]{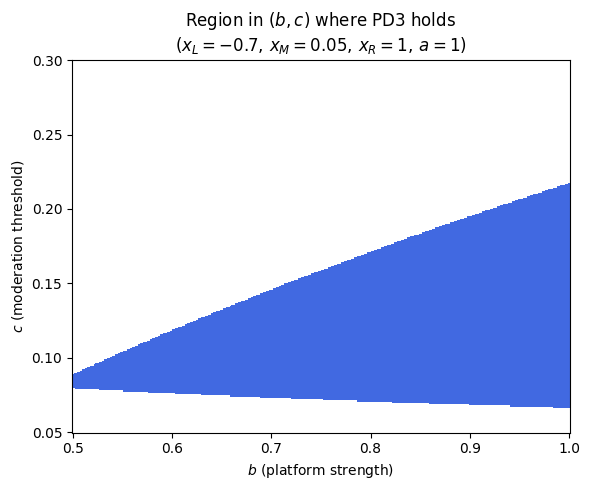} 
    \caption{
    Region in the $(b,c)$ parameter space where the persistent disagreement (the PD3 condition described formally in Theorem \ref{Thm:3-agentsystem}) holds for the three-agent system with initial opinions $x_L = -0.7$, $x_M = 0.05$, $x_R = 1$, and $a = 1$. Each point in the plot corresponds to a unique pair of platform strength $b$ and moderation threshold $c$. The blue region indicates parameter combinations for which the PD3 condition holds. 
    }
    \label{fig:PD3_region}
\end{figure}

In addition, we study numerically the impact of the platform's moderation parameter 
$c$ on the final polarization levels among the agents. The parameter 
$c$ dictates how broadly or narrowly a platform defines moderate opinions, influencing the type of content recommendations an agent receives based on their current opinion. As 
$c$ increases, the platform labels a wider spectrum of opinions as moderate and consequently promoting more centrist content to a larger fraction of users. Our simulations reveal a direct relationship between a higher 
$c$ value and reduced polarization. Specifically, as the platform adopts a broader definition of moderation (higher 
$c$), the final opinions are more moderate or less polarized. This shows the influence of platform recommendation algorithms on shaping public discourse.  

We define the final polarization by the difference between the average opinions of left and right agents, i.e., $(\sum _ { x_{j} > 0} x_{j} - \sum _{x_{i} < 0} x_{i} )/ n$ when agents opinions do not equal zero.\footnote{In the rare cases where agents' final opinions equaled exactly $0$, we assigned them to be left and right agents evenly when possible.} For the plots in Figure \ref{fig:moderation} we also assume that the probability of in-block connection is $1/4$ and the probability of out of block connection is $1/8$.

\begin{figure}
\centering
\includegraphics[width=0.65\linewidth]{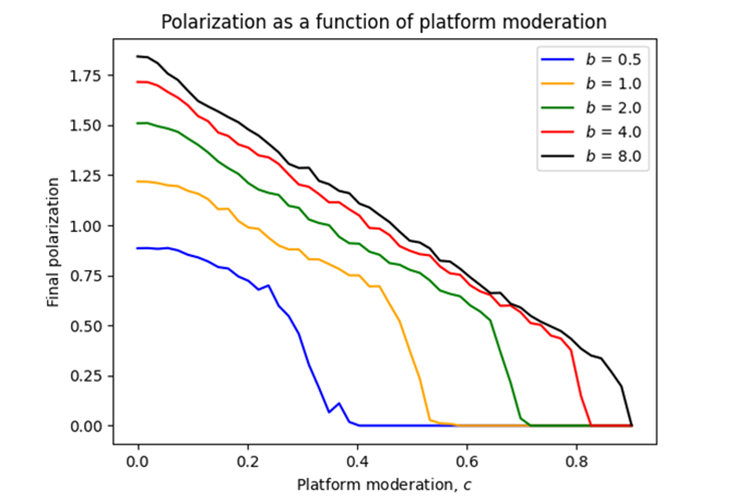}
  \caption{Parameters for the model were set to $a=1, \epsilon = 0.0001$, $60$ agents, $30$ with a random initial opinion that is a uniform random variable on $[0,1]$ and $30$ agents with a random initial opinion that is a uniform random variable on $[-1,0]$. We ran 1000 samples for each $b,c$ value and the final polarization was calculated as the average final polarization over all 1000 samples. 
  The plot shows that higher platform moderation leads to smaller final polarization for different platform influence parameters.}
  \label{fig:moderation}
\end{figure}

\subsection{Correlated Initial Opinions and  Connection Probabilities} \label{sec:correlated}

In the baseline model, we made the simplifying assumption that the initial opinions of agents and the probabilities of connections between them are uncorrelated. This assumption, although mathematically tractable and used to prove Theorem \ref{Thm:Main}, might not always hold in real-world settings, especially in social media platforms where users tend to connect with others who share similar viewpoints.

To accommodate this more realistic scenario, in this section we introduce a correlation between the agents'  initial opinions and their connection probabilities. Specifically, the probability that two agents form a connection increases as their initial opinions become more similar. For two agents $i,j$ with initial opinions $x_{i}(0)$ and $x_{j}(0)$, we assume that the probability of a connection between them is given by $1/(1+\gamma|x_{i}(0)-x_{j}(0)|)$. Hence, when the difference between the initial opinions is higher then the probability of a connection between these two agents is lower.\footnote{From our experiments, other functions that are decreasing in the difference of opinions $|x_{i}(0)-x_{j}(0)|$ instead of the function $1/(1+\gamma|x_{i}(0)-x_{j}(0)|)$ provided similar insights. } Intuitively, the parameter $\gamma$
 dictates the strength of the correlation.

In our simulations (see Figure \ref{fig:Correlated}) we observe that introducing a correlation between the initial opinions of agents and their connection probabilities tends to amplify the final polarization across different platform moderation parameter $c$
 and varying levels of platform influence parameter $b$. This phenomenon can be attributed to the fact that when agents are more likely to connect with others harboring similar initial opinions, it fosters an environment where echo chambers can proliferate. In such scenarios, agents predominantly interact with peers who reinforce their existing beliefs, overturning the potential moderating effects that interactions with individuals with differing opinions could have fostered. This accentuates the necessity for platforms to not only consider content moderation and influence strategies but also to actively encourage a diversity of connections in order to reduce final polarization. 

To add richness to the analysis, we run the simulations described in Figure \ref{fig:Correlated} with the three media outlets model described in Section \ref{sec:three} and varied model parameters. In particular, we consider $30$ agents with initial opinion that is uniformly distributed on $[-1,0]$ and $30$ agents with initial opinion that is uniformly distributed on $[0,1]$. Similar qualitative  results were observed with different parameters. 

\begin{figure}
\centering
\begin{subfigure}{.5\textwidth}
  \centering
  \includegraphics[width=1.0\linewidth]{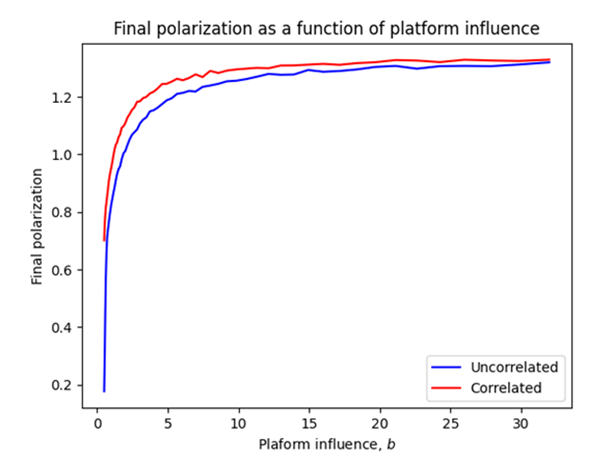}
  \subcaption{Final Polarization as a function of $b$}
  \label{fig:Cofig1}
\end{subfigure}%
\hfill
\begin{subfigure}{.5\textwidth}
  \centering
  \includegraphics[width=1.0\linewidth]{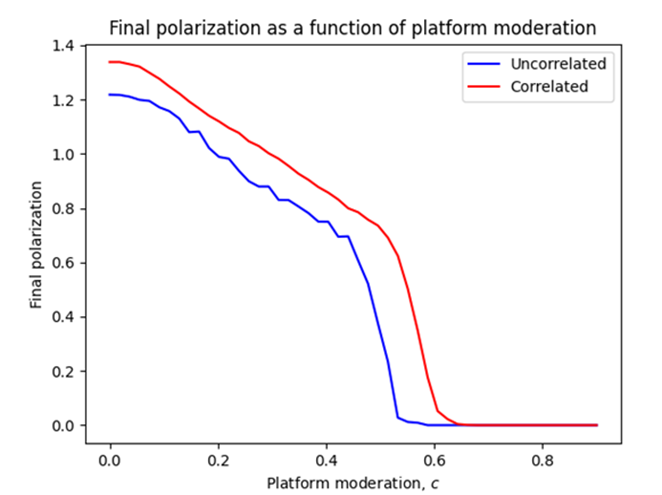}
  \subcaption{Final Polarization as a function of $c$}
  \label{fig:Cofig2}
\end{subfigure}
\caption{ In the left figure we let $a=1$, $c=0.33$ and in the right figure we use $a=b=1$. We consider $\gamma =15$ so the probability of a connection between two agents with initial opinions $x_{i}(0),x_{j}(0)$ is $1/(1+15|x_{i}(0)-x_{j}|)$. Other parameters led to similar qualitative results.  }
\label{fig:Correlated}
\end{figure}

\subsection{Convergence of Opinions}\label{sec:converge}

We next show that when the influence matrix $A$ is symmetric  (note that $A$ here is the influence matrix like in Section \ref{sec:model} and not necessarily an adjacency matrix) the opinions of agents always converge to an equilibrium. We then provide a numerical example that shows that when the influence matrix $A$ is not symmetric, the opinions do not necessarily converge. 

Intuitively, there are two main forces in our model. The first is social learning, which tends to bring opinions closer together. The second is the platform's influence, which tends to push opinions toward the outlets aligned with the agents. The extent of both of these effects also depends on each agent's current opinion. With social learning alone, opinions will only become closer as time passes, and they will eventually converge. Indeed there are standard convergence results under weak connectedness assumptions \citep{proskurnikovT17-tutorial-dynamic-networks}. With the platform alone, opinions also converge. Agents' opinions will equal one of the media outlets' opinions ($-1$ or $1$) depending on their initial opinions: $x_i(\infty) \equiv \lim_{t \to \infty} x_i(t) = \sgn(x_i(0))$. 

With both social learning and the platform's influence, it is not even intuitively clear that opinions will converge to a limit. Because the push of the platform's influence depends on the agent's current opinion, agent $i$'s opinion might start off as somewhat right-leaning, i.e., $x_i > 0$ and be pushed by the platform to the right, but social learning brings her opinion back to being left-learning, i.e., $x_i < 0$. At the same time, the opinions of the other agents also change, and now social learning can bring agent $i$'s opinion to the right again. The vector fields arising from the platform are different for each orthant of the opinion space, so it is not inconceivable that the trajectory, in principle, can converge to a limit cycle. We show in a numerical example that this cyclical behavior can arise even in a simple network with 4 agents.  

Nevertheless,
as we establish in Theorem \ref{thm:opinions-converge} below,  such situations do not arise as long as the influence matrix $A$ is symmetric, i.e., $a_{ij} = a_{ji}$ for all $i,j$. In this case, the opinions' trajectory always converges to a limiting opinion or an equilibrium. The convergence result holds for a Lipschitz  continuous function $s$, under mild conditions (that are satisfied by $\mathrm{sgn}_\epsilon)$. 

\begin{theorem}\label{thm:opinions-converge}
Suppose that the influence matrix $A$ is a symmetric matrix of nonnegative entries, $L=L[A]$ is the Laplacian of $A$, $B= \diag {(\boldsymbol{b})}$ is a diagonal matrix with positive entries, $\mathbf{x}_0 \in \R^{n}$, and $s: \R \to [-1,1]$ is a Lipschitz continuous function. Let $K = \max(\max_{i} |x_i(0)|, 1)$ and $\Omega = [-K, K]^{n}$. Then, there exists a unique continuous solution $\{\mathbf{x}(t)\}_{t\geq 0}$ to the differential equation $\dot{\mathbf{x}} = - L \mathbf{x} + B(s(\mathbf{x}) - \mathbf{x})$ with initial condition $\mathbf{x}(0) = \mathbf{x}_0$, and $\mathbf{x}(t) \in \Omega$ for all $t \geq 0$.

Suppose further that
 $E=\{\mathbf{x}\in \Omega|(L+B) \mathbf{x} = B s(\mathbf{x}) \} $ is a finite set.
Then the limit $\mathbf{x}_\infty \equiv \lim_{t \to \infty} \mathbf{x}(t)$ exists and is such that
$\mathbf{x}_\infty \in E$. 
\end{theorem}

Theorem \ref{thm:opinions-converge} shows agents' opinions  converge to a limiting vector of opinions for any continuous function $s(\cdot)$ and any symmetric influence matrix $A$. When $A$ is not symmetric then the convergence of opinions is not guaranteed as the following simple numerical example shows. Suppose that  $\dot{\mathbf{x}} = - L \mathbf{x} + b(\sgn_{\epsilon}(\mathbf{x}) - \mathbf{x})$ as in Equation (\ref{eqn:main-dynamics-b}). Suppose that the parameters are given by $\epsilon = 0.1$, $b=0.6$ and there are only four agents. For $i=1,2,3$ assume that $a_{ij}=1$ if $j=i+1$ and $0$ otherwise, and for $i=4$ assume that $a_{ij}=1$  for $j=1$ and zero otherwise; i.e., the network is a directed cycle with an influence parameter of $1$. In Figure \ref{fig:example} we plot the agents' opinions over time when the initial opinion vector is given by $(x_{1}(0),x_{2}(0),x_{3}(0),x_{4}(0))=(-1/2,1,1/2,-1)$. We see in Figure \ref{fig:example} that the agents' opinions are cyclic and do not converge to an equilibrium.

\begin{figure}
\centering
\includegraphics[width=0.65\linewidth]{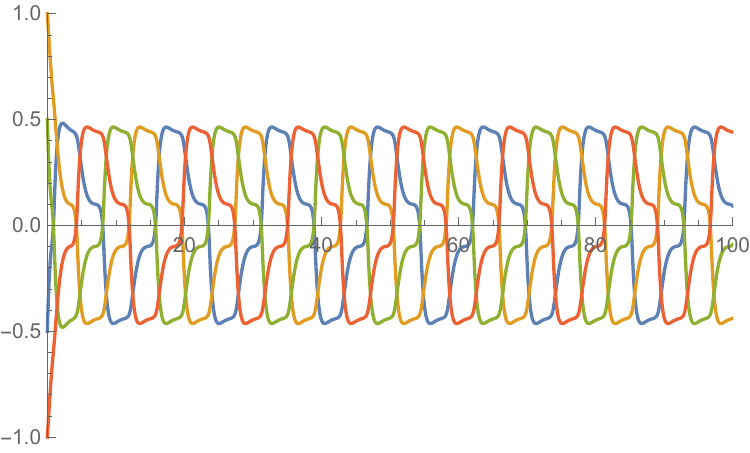}
  \caption{Initial opinions: $(x_{1}(0),x_{2}(0),x_{3}(0),x_{4}(0))=(-1/2,1,1/2,-1)$. Each curve depicts the evolution of the opinion of a different agent over time. The agents'  opinions are cyclic and do not converge to an equilibrium.}
  \label{fig:example}
\end{figure}

We note that the example provided in Figure \ref{fig:example} does not contradict Theorem \ref{Thm:Main}. While a given realization in the stochastic block model typically leads to an influence matrix that is not symmetric, the influence matrix is  symmetric in expectation. In Theorem \ref{Thm:Main} we show that for connected networks, this is suffice to show that opinions approximately converge when the number of agents approaches infinity.  

\section{Extensions and Simulations}\label{sec:sim}

 Throughout this section, we consider two set of simulations.  As in Section \ref{sec:main-result} we assume that if agent $i$ has $|N(i)|$ connections then each connection has an influence weight of $1/|N(i)|$, i.e., $a_{ij} = 1/|N(i)|$ is the influence of agent $j$ on agent $i$. For the  model with two media outlets, we denote by $L$ and $R$ the sets of $L$ agents (agents with negative initial opinions) and $R$ agents (agents with positive initial opinions). Each agent in the $L$ block ($R$ block) has an initial opinion that is drawn independently from a known distribution $\mathcal{D}_L$ ($\mathcal{D}_R$).  As in Equation (\ref{eqn:main-dynamics-b}), opinions evolve according to $ \dot{x}_i(t) = \sum_{j \neq i} a_{ij} (x_j(t) - x_i(t)) + b (\sgn_{\epsilon}(x_i(t)) - x_i(t))$ for sufficiently small $\epsilon$. 

The first set of simulations that we consider is based on simulating a stochastic block model with two blocks. Each block has $n$ agents. Each pair of agents from the same block is connected with probability $p_s$, and each pair of agents from different blocks is connected with probability $p_d$. Our simulations suggest that the predictions and insights our theory provide for large stochastic networks apply also for small networks where the initial opinions are random. In Appendix B we provide simulations for a non-normalized stochastic block model with two blocks and show that similar results hold for the non-normalized model also. 

The second set of simulations that we consider is based on simulating a given network that is constructed from real-world data. We use the political blogs data collected by  \cite{adamic2005political}. In this data agents are labeled to liberal (left) and conservative (right) agents based on their political blogs.  \cite{adamic2005political} measure the degree
of interaction between left and right agents and construct the connection between agents  using  the activity in the political  blogs.\footnote{We removed from the original data nodes without connections.}  The
 network has 1,222 nodes and 33,558 edges, There are 586 left-leaning nodes and 636 right-leaning nodes in the 
network. The average degree (the average number of connections that an agent has) is 13.73.  Interestingly, most nodes have a small number of connections and the variation
is quite large. Approximately $1/3$ of nodes have less than 3 connections, and a few nodes have more
than 200 connections. Our simulations suggest that our theoretical insights hold also for this political blogs network.

\subsection{Persistent Disagreement and Polarization}\label{subsec:sim-disagree-polarize}

In Theorem  \ref{thm:equal-a-init-char} we show that the limiting opinions' polarization is increasing with the platform's influence for the two-agent system. We extended this result in Theorem \ref{Thm:Main} for large stochastic block models with two blocks. We now show that this result extends for the political blogs network and for a small stochastic block model with random initial opinions. We also show that the polarization trajectories converge monotonically to the limiting polarization as Proposition \ref{prop:polar-monotonic} and Corollary \ref{Cor:polarization} suggest for the two-agent system and the large stochastic block model, respectively. 

For the stochastic block model we consider the parameters $(n, p_s, p_d) = (32, 1/4, 1/8)$. The initial opinions for $L$ and $R$ agents are drawn independently from $\mathcal{D}_L \equiv \text{Unif}[-2,0]$ and $\mathcal{D}_R \equiv \text{Unif}[0,2]$. We define the polarization as $\frac{1}{n} \sum_{j \in R} x_j - \frac{1}{n} \sum_{i \in L} x_i $, which is the difference between average opinions of the two groups. For each $b \in \{0.25, 0.5, 1, 2, 4, 8, 16, 32\}$, we compute the distribution of final polarization \textit{conditional on reaching persistent disagreement} by simulating the opinion trajectory from random graphs and initial opinions $N = 1000$ times. 
We plot the 5th, 50th, and 95th percentiles of the final polarization in blue error bars in Figure~\ref{fig:polarization-empirical-theoretical}. Alongside the empirical distributions of final polarization, we plot the values $2b/(2p_d/(p_{d}+p_{s}) + b)$, which we call the theoretical final polarization. The theoretical final polarization follows from the characterization of the persistent disagreement equilibria in the two-agent system (see Section \ref{subsec:2-agent}) and our approximation result (see Theorem \ref{Thm:Main}). Figure~\ref{fig:polarization-empirical-theoretical} shows that \textit{the distributions of final polarization are concentrated around their theoretical values}, demonstrating our theory's predictive power. 


\begin{figure}
\begin{subfigure}{.47\textwidth}
  \centering
  \includegraphics[width=1.0\linewidth]{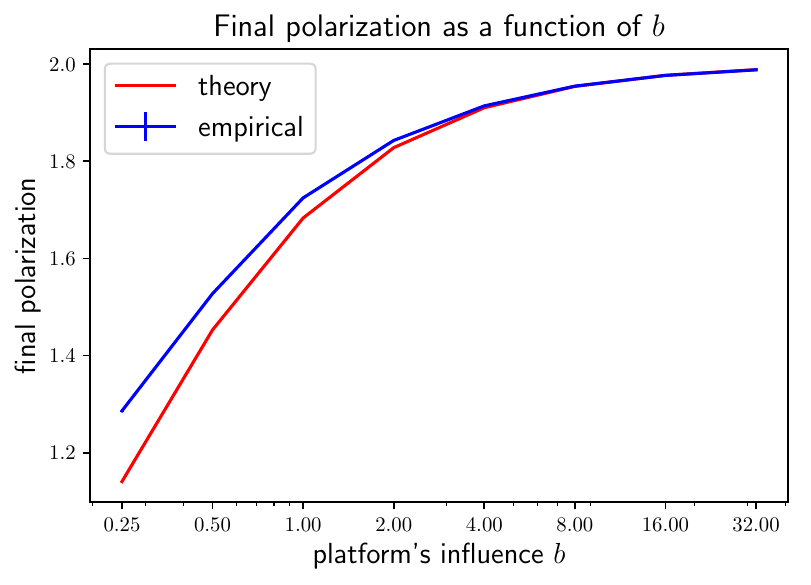}
  \subcaption[1.25\linewidth]{{Political blogs network}}
\end{subfigure}%
\hfill
\begin{subfigure}{.47\textwidth}
  \centering
  \includegraphics[width=1.02\linewidth]{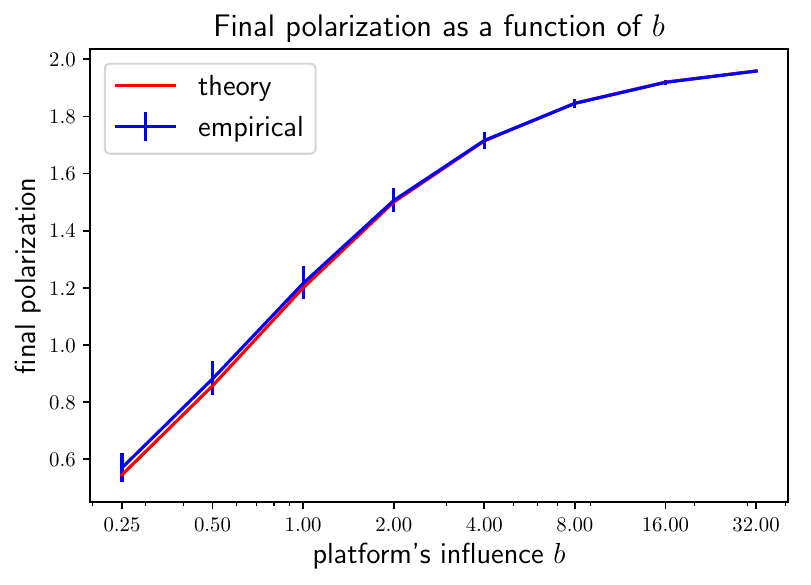}
  \subcaption{Stochastic block model}
\end{subfigure}
\caption{The limiting opinions' polarization is concentrated around the theoretical values.}
 \label{fig:polarization-empirical-theoretical}
\end{figure}

We next consider the trajectory of how polarization evolves from the initial condition to the limiting opinions. We use different values of $h$ where $\mathcal{D}_L \equiv \text{Unif}[-h,0], \text{ and }  \mathcal{D}_R \equiv \text{Unif}[0,h]$, because a higher $h$ means a higher initial polarization, and we  illustrate the trajectory monotonicity for a wide range of initial polarization values. Figure~\ref{fig:polarization-trajectory} plots the means and standard deviations of the trajectory for each $h$ and each time step. We can see that \textit{the polarization trajectory is monotonic, increasing for low initial polarization states and decreasing for high initial polarization states, and they all converge to roughly the same limit polarization}. 

\begin{figure}
\begin{subfigure}{.47\textwidth}
  \centering
  \includegraphics[width=1.0\linewidth]{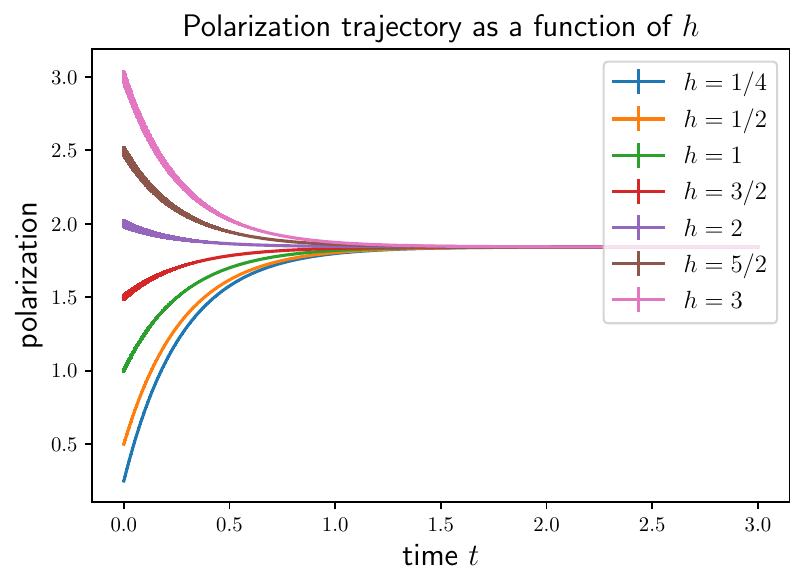}
  \subcaption[1.25\linewidth]{{Political blogs network}}
\end{subfigure}%
\hfill
\begin{subfigure}{.47\textwidth}
  \centering
  \includegraphics[width=1.02\linewidth]{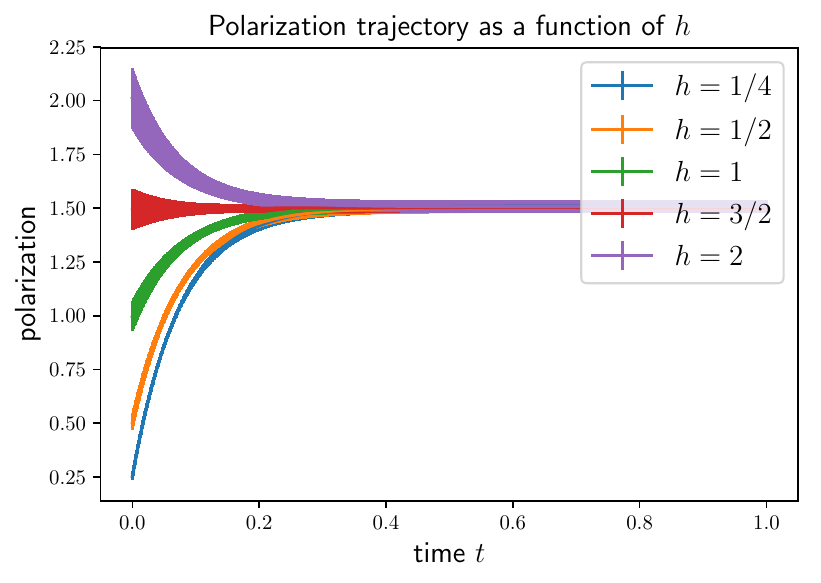}
  \subcaption{Stochastic block model}
\end{subfigure}
\caption{The polarization trajectories are monotonic and converge to the same limit polarization.}
  \label{fig:polarization-trajectory}
\end{figure}

\subsection{Impact of Initial Polarization on the Likelihood of Consensus}\label{subsec:sim-balance}

 In this section we show that \textit{low-polarization initial opinions generally lead  to lower consensus probabilities} for a stochastic block network with random initial opinions and for the political blogs  network (see the discussion after Theorem 
 \ref{thm:equal-a-init-char}). We consider a stochastic block model with parameters $(n,p_s,p_d) = (32, 1/4, 1/8)$. We use the parameters  $b = 0.05$, and $\mathcal{D}_L \equiv \text{Unif}[-h,0], \mathcal{D}_R \equiv \text{Unif}[0,h]$ for $h \in \{0.1, 0.2, \dots, 3.0\}$. The mean initial polarization between a left and a right agent is $h/2-(-h/2)=h$. We simulate the system $N = 10000$ times. If $m$ instances are consensus, we compute the 95\% confidence interval as $p \pm 2 \sqrt{p(1-p)/N}$ where $p = m/N$. Figure~\ref{fig:init-pol-prob} shows that consensus probability is increasing in the initial polarization $h$ for the stochastic block model and the political blogs network.

\begin{figure} 
\begin{subfigure}{.48\textwidth}
  \centering
  \includegraphics[width=1.0\linewidth]{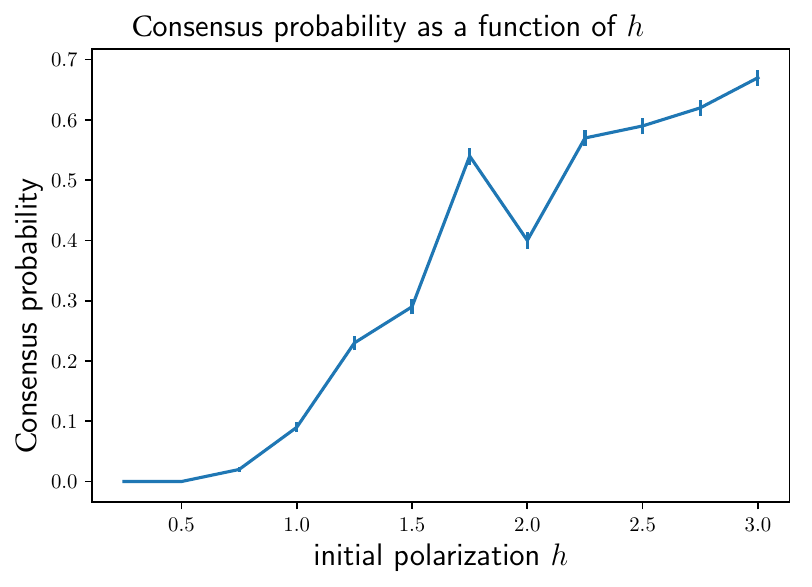}
  \caption{Political blogs network}
\end{subfigure}%
\hfill
\begin{subfigure}{.48\textwidth}
  \centering
  \includegraphics[width=0.95\linewidth]{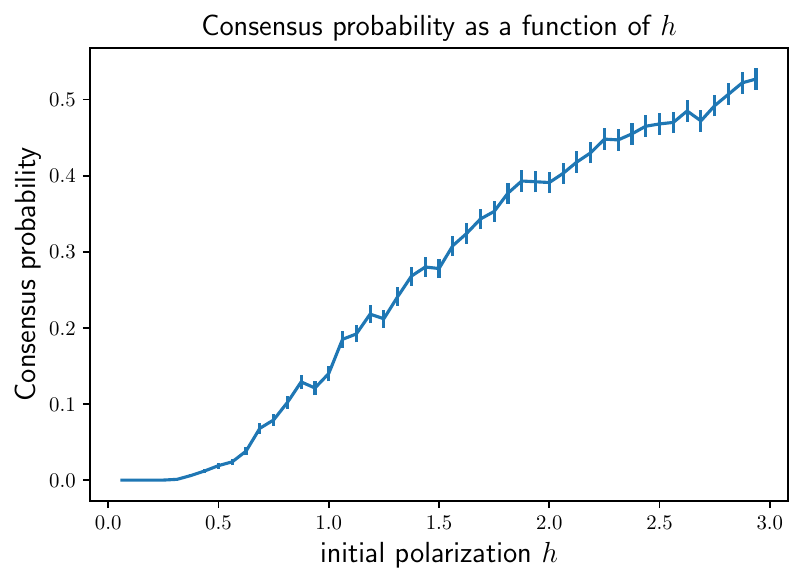}
  \caption{Stochastic block model}
\end{subfigure}
\caption{Consensus probability is increasing in the initial polarization. Recall that in the political blog network the agents' connections and their political view (left-leaning or right-leaning) are labeled but the initial opinion (actual number) is not.  }
\label{fig:init-pol-prob}
\end{figure}

In the model with two media outlets, consensus corresponds to zero polarization. As suggested by Figure \ref{fig:init-pol-prob}, a higher initial polarization often results in reduced final polarization. In the three media outlet model presented in Section~\ref{sec:three}, the relationship between consensus and polarization is more nuanced. 
In particular, partial consensus such as alignment between two of the three groups (e.g., moderates and right-leaning agents) does not necessarily imply a reduction in overall polarization. Depending on the  configuration, such partial alignment may coincide with increased divergence from the third group (e.g., left-leaning agents), leading to higher overall final polarization despite lower overall clustering.\footnote{For example, when $c$ is very small, the system behaves similarly to the two-agent case, and high platform influence ($b$) may lead to final opinions that are highly polarized between two groups.}  
Nonetheless, our simulations reveal a pattern similar to the two-group setting: as initial polarization increases, final polarization tends to decrease, though the effect is modest.

 Figure \ref{fig:initial-3block} illustrates the impact of initial opinions on final polarization. We use the parameters $c = 0.35$ and $b = 0.4$, with 50 left agents whose initial opinions are drawn uniformly at random from $[-h, -c]$, 50 moderate agents from $[-c, c]$, and 50 right agents from $[c, h]$ and consider $h \in \{0.5,1,1.5,2\}$.\footnote{An alternative approach to varying initial polarization would be to split the population into two groups with initial opinions drawn from $[-h, 0]$ and $[0, h]$, as done in the two-block model. However, applying this in the three-outlet setting would implicitly reduce the expected number of moderate agents, since $c$ is fixed and fewer agents would fall into the $[-c, c]$ interval. To isolate the effect of initial polarization without altering the group composition, we chose instead to fix the proportion of moderate, left, and right agents, and increase the extremity of initial opinions for the left and right groups.

\begin{figure}[H]
    \centering
    \begin{subfigure}[b]{0.4\textwidth}
        \includegraphics[width=\textwidth]{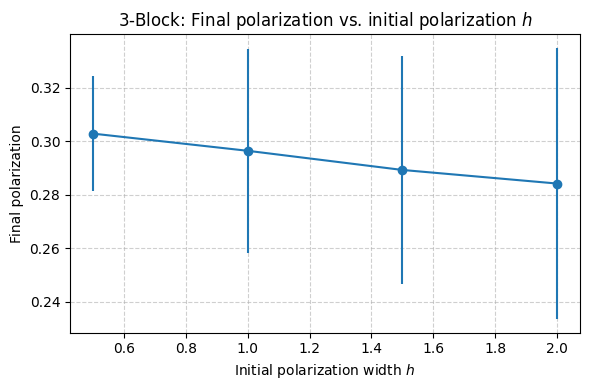}
        \caption{Final vs. initial polarization}
        \label{fig:polarization_vs_h}
    \end{subfigure}
    \caption{
    Simulation results for the three-block model with varying levels of initial polarization $h$. 
    For each value of $h \in \{0,5,1,1.5,2\}$, we run 200 simulations with a stochastic block model consisting of 150 agents (50 left, 50 moderate, 50 right). The figure summarizes the relationship between initial polarization width $h$ and the average final polarization. For each $h$, we compute the mean and standard deviation of the final polarization metric across the 200 runs. The results indicate a modest decreasing trend, suggesting that as initial polarization increases, the final system-wide polarization tends to decrease. 
    }
      \label{fig:initial-3block}
\end{figure}


}

\subsection{Extremism}\label{subsec:extremism}

In the stochastic block model simulation we consider the parameters $(n, p_s, p_d) = (32, 1/4, 1/8)$. The initial opinions are drawn independently form the distributions $\mathcal{D}_L \equiv \text{Unif}[-2,0]$ and $\mathcal{D}_R \equiv \text{Unif}[0,2]$. We define the extremism as $\frac{1}{|V|} \sum_{i \in V} |x_i|$, where $V = L \cup R$ is the set of agents. This extremism measure is the equal weight $l_1$ distance of opinions.\footnote{We can use $l_2$ or other distance measures in place of $l_1$ in the definition. The results are analogous.}

In Figure~\ref{fig:change-in-extremism}, for $b \in \{0.01,0.05,0.1,0.5,1,5,10\}$, we calculate the final extremism (i.e., the extremism of the limiting opinions). We plot the 25th, 50th, and 75th percentiles (interquartile range) in blue from $N=100$ random samples. We see that for low $b$ and for high $b$, the extremism is high while for intermediate $b$, the extremism is moderate.\footnote{The interquartile range for extremism at values of $b$ where  the limiting opinion is either a consensus equilibrium or a persistent disagreement equilibrium, and both options have high probability is much wider than that of other values of $b$ where this does not happen.} This insight agrees with our discussion in Section \ref{subsec:2-agent} that when the platform is either very weak or very strong, society becomes more extreme, whereas when the platform's strength is intermediate, society becomes less extreme. 


\begin{figure} 
\begin{subfigure}{.48\textwidth}
  \centering
  \includegraphics[width=1.0\linewidth]{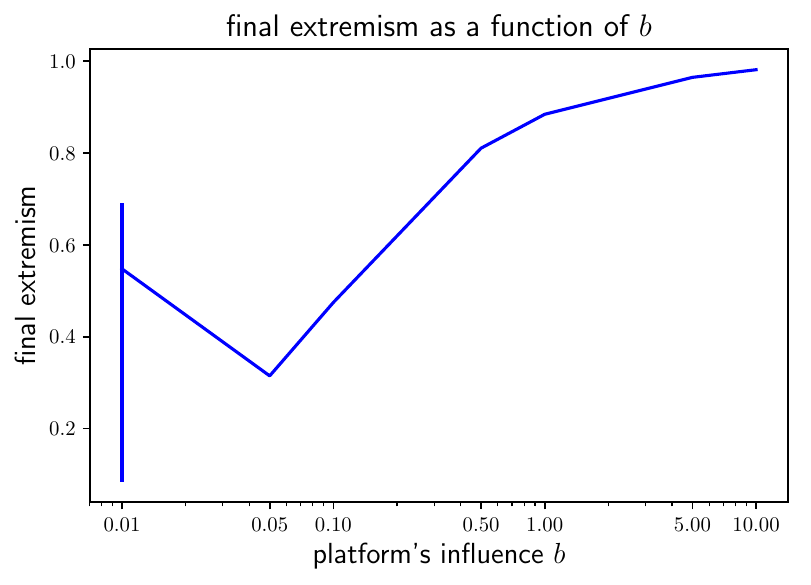}
  \caption{Political blogs network}
\end{subfigure}%
\hspace{1em}
\begin{subfigure}{.48\textwidth}
  \centering
  \includegraphics[width=1.06\linewidth]{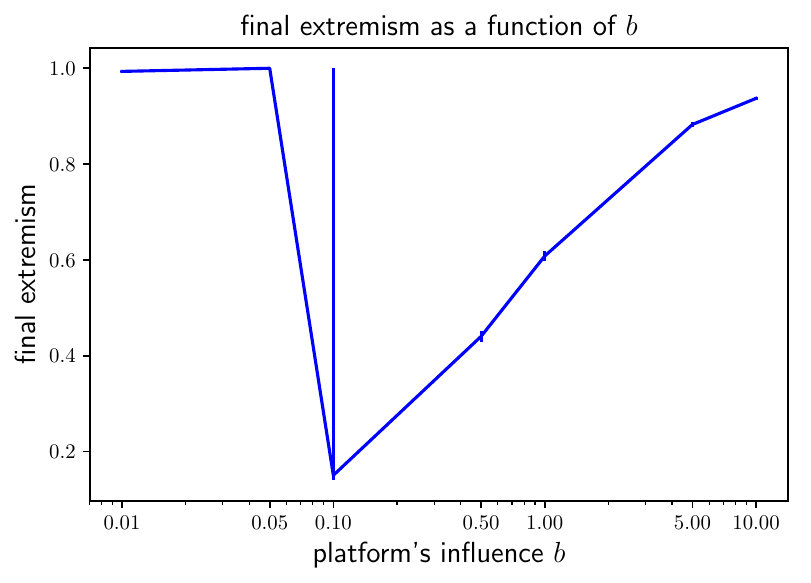}
  \caption{Stochastic block model}
\end{subfigure}
\caption{The effect of the platform's strength $b$ on extremism.}
\label{fig:change-in-extremism}
\end{figure}

Figure~\ref{fig:extremism-two-sources} also supports this argument (see Table \ref{table:extremism-reasons} for a summary). In blue, we plot the consensus probability as a function of $b$. In red, we plot the mean final extremism conditional on reaching persistent disagreement (rather than consensus). In black, we plot the unconditional mean final extremism. We see that consensus probability is decreasing in $b$, while mean final extremism conditional on persistent disagreement is increasing in $b$. For intermediate and large $b$, the consensus probability is near zero so the extremism conditional on persistent disagreement is essentially the same as unconditional extremism. We see this visually when the black and red lines agree. When $b$ is very small, on the other hand, the consensus probability is high, so even when the extremism is low if we reach persistent disagreement, the probability of that happening is low, and most of the time we get consensus with maximum extremism. We then see that in the intermediate regime, the consensus probability is near zero so the final extremism is based on the persistent disagreement equilibrium, which is still low because $b$ is relatively low. This is the ``good'' regime. The ``bad'' regimes are very low $b$ and very high $b$ with high extremism, but they come from different sources: from consensus, and from extreme disagreement, respectively.

\begin{figure} 
\begin{subfigure}{.48\textwidth}
  \centering
  \includegraphics[width=1.0\linewidth]{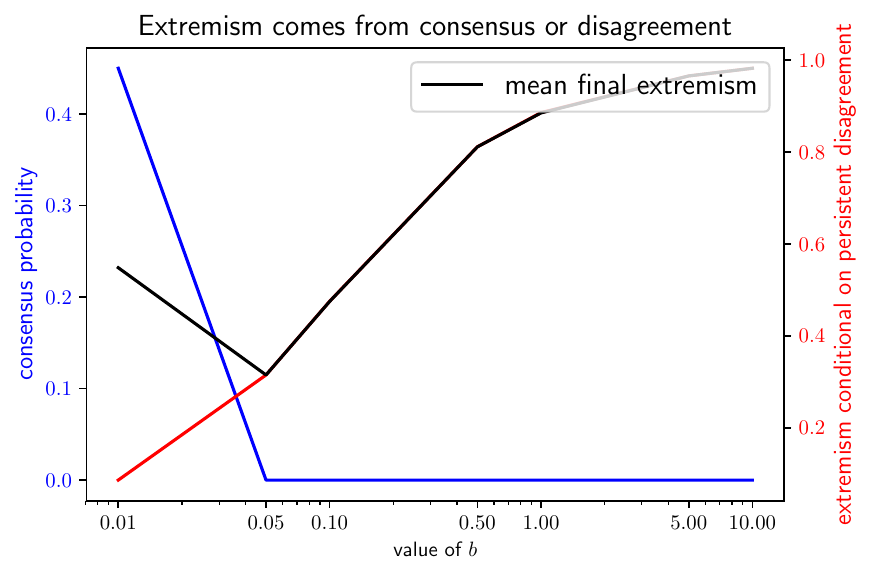}
  \caption{Political blogs network}
\end{subfigure}%
\hspace{1em}
\begin{subfigure}{.48\textwidth}
  \centering
  \includegraphics[width=1.06\linewidth]{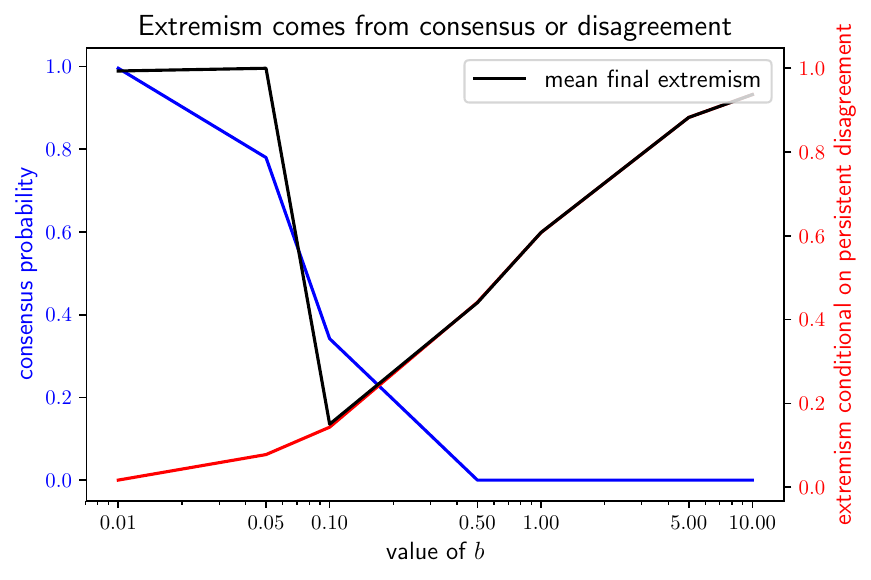}
  \caption{Stochastic block model}
\end{subfigure}
\caption{The sources of extremism: extremism from consensus (low $b$) or disagreement (high $b$).}
\label{fig:extremism-two-sources}
\end{figure}




\begin{table}[h!]
\centering
\begin{tabular}{||c | c | c||} 
 \hline 
 Platform's Strength & Extremism & Reason  \\ [0.5ex] 
 \hline\hline
 weak & high & consensus \\ 
 intermediate & low & moderate persistent disagreement \\ 
 strong & high & near-extreme persistent disagreement \\ [1ex] 
 \hline
\end{tabular}
\caption{The platform reduces extremism when its influence is intermediate.}
\label{table:extremism-reasons}
\end{table}

In the model with two media outlets, extremism peaks when platform influence is minimal because consensus results in unequivocally extreme final opinions with all agents converging to an opinion of either $1$ or $-1$. On the other hand, in the model incorporating three media outlets (see Section \ref{sec:three}), when the moderation parameter 
$c$ is high and the platform's influence 
$b$ is low, the system often steers towards a consensus equilibrium where all agents adopt a neutral opinion of $0$. Figure \ref{fig:finalext3block} shows the final extremism across varying moderation parameters 
$c$. We see that a high 
$b$ tends to amplify final extremism, echoing the behavioral dynamics seen in the two media outlets model where society becomes more extreme as the final polarization is high. Conversely, with high moderation parameters 
$c$ and small $b$, the resulting extremism is low with a higher likelihood of agents converging towards the neutral opinion of $0$. Nonetheless, even in this context, one could posit that society exhibits a form of extremism, given the uniformity in opinions, mirroring the characteristics observed in the two media outlets model.

\begin{figure}
\includegraphics[width=0.95\textwidth,height=0.3\textheight]{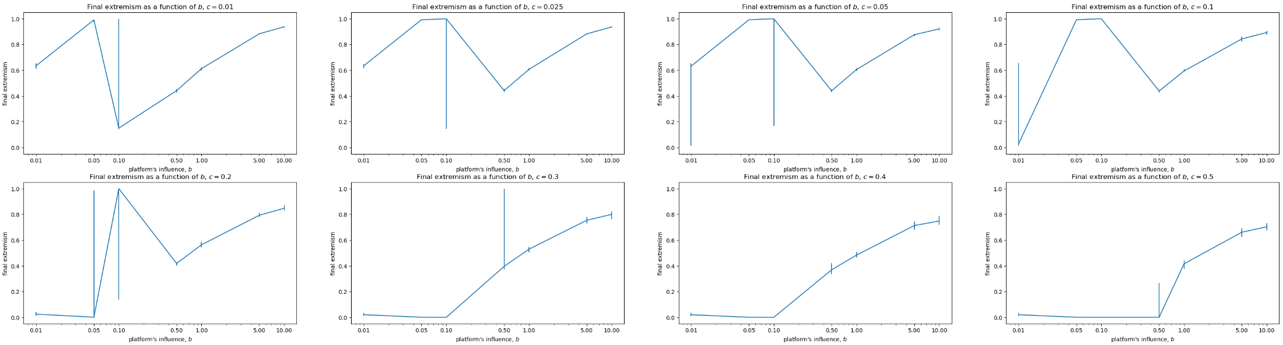}
  \caption{Final extremism as a function of the platform's influence parameter $b$ for different moderation parameters $c$.}
  \label{fig:finalext3block}
\end{figure}

\section{Conclusions}\label{sec:conclude}

In this paper, we investigate the opinion dynamics of agents connected in a social network under the influence of a platform's personalized content recommendation. There are right-leaning and left-leaning media outlets, the platform shows each agent the content from the media outlet closest to that agent's opinion, and each agent's opinion evolves by taking the weighted average of the platform's shown content and her neighbors' opinions.  We first focus on a two-agent system and find that
\begin{itemize}
    \item[(1)] a consensus equilibrium is less likely when the initial opinions are more balanced and the platform's influence is strong.  
    \item[(2)] the extremism of agents' limiting opinions is non-monotonic in the platform's influence, i.e., extremism is low when the platform's influence is intermediate, but extremism is high when either the platform's influence is very weak (leading to consensus) or very strong (leading to a high degree of disagreement).
\end{itemize}

We then show that for large and dense enough networks, the agents' opinions in the  general stochastic block model with two blocks can be approximated by the agents' opinions in the two-agent system. Thus, the findings mentioned above generalize to a stochastic block model with two blocks.  We extended our analysis and approximation result to a model   with three media outlets and  tested numerically our results in a more complicated  model where the agents' initial opinions are correlated with the connection probabilities. 
Lastly, we show that agents' opinions converge to an equilibrium when the influence matrix is symmetric. 

There are a few avenues for future research. In the context of our model, future work can mathematically analyze the stochastic block model with more than three blocks or more complicated structure of media outlets' opinions. 
Further, our model abstracts away issues related to welfare, incentives, and misinformation because our model, like most naive social learning models, is purely descriptive, non-strategic, and opinions are not based on ground truth. We will need a new model, perhaps based on the model proposed in this paper, to explore the welfare implications of the platforms' engagement maximization behavior and design appropriate and incentive-aligned interventions to mitigate  the possible downsides of such behavior.

%
%
%
%

\newpage
\bibliographystyle{ecta}
\bibliography{references}

\newpage


\appendix

\section{Proofs in Section \ref{sec:theory}}

\subsection{Proof of Theorem \ref{thm:equal-a-init-char}}

\begin{proof}{Proof of Theorem \ref{thm:equal-a-init-char}.}

We describe the initial opinions that lead to a persistent disagreement equilibrium. It is clear that if the initial opinion is in the $(+,+)$ quadrant it will lead to a consensus equilibrium at $\mathbf{1} = (1, 1)$. Similarly, if the initial opinion is in the $(-,-)$ quadrant it will lead to a consensus equilibrium at $-\mathbf{1} = (-1,-1)$. It remains to consider initial opinions in the $(+,-)$ and $(-,+)$ quadrants. By symmetry we will consider the $(-,+)$ quadrant, as the other case is analogous.

We first need to show the following two lemmas. The proofs of these lemmas are given in the next subsection, Appendix~\ref{ap:subsec:supporting-lemmas}.

\begin{lemma}\label{lem:stay-in-quadrant}
Let $\mathbf{x}(0) = (x_1(0), x_2(0))$ be an initial condition of the dynamical system
\begin{align*}
    \dot{x}_1 &= a(x_2-x_1) + b(-1 - x_1) \\
    \dot{x}_2 &= a(x_1-x_2) + b(1 - x_2).  
\end{align*}
such that $x_1(0) < 0, x_2(0) > 0$.
If $(a,b,\mathbf{x}(0))$ satisfies condition PD, then $x_1(t) < 0$ and $x_2(t) > 0$ for all $t \geq 0$. If $(a,b,\mathbf{x}(0))$ satisfies condition CO and $x_1(0)x_2(0) < 0$, then there exists a $t' > 0$ such that $x_1(t) < 0, x_2(t) > 0$ for all $0 \leq t < t'$ and either $x_1(t')=0,x_2(t') > b/a$ or $x_1(t') < -b/a,x_2(t')=0$. An analogous symmetric statement holds with the system 
\begin{align*}
    \dot{x}_1 &= a(x_2-x_1) + b(1 - x_1) \\
    \dot{x}_2 &= a(x_1-x_2) + b(-1 - x_2). 
\end{align*}
with $x_1(0)>0,x_2(0) < 0$
\end{lemma}

\begin{lemma}\label{lem:cross-eps-band}
Let $\mathbf{x}(0) = (x_1(0), x_2(0))$ be an initial condition of the dynamical system
\begin{align*}
    \dot{x}_1 &= a(x_2-x_1) + b(x_1/\epsilon - x_1) \\
    \dot{x}_2 &= a(x_1-x_2) + b(1 - x_2)   
\end{align*}
such that $x_1(0)=-\epsilon$ and $x_2(0) \geq b/a$. 
Then for all sufficiently small $\epsilon > 0$, there exists a $t' > 0$ such that $x_2(t) > \epsilon$ for all $t \in [0,t']$, $x_1(t)$ is increasing on $[0,t']$ and $x_1(t')=\epsilon$.
\end{lemma}

Lemma \ref{lem:cross-eps-band} says that if the trajectory enters the $\epsilon$-band of the positive y-axis at point $x_2(0) \geq b/a$, then it will enter the $(+,+)$ quadrant.


Consider $x_1(0) < 0, x_2(0) > 0$. Let  $\tilde{\mathbf{x}}$ be the solution trajectory when $\sgn_{\epsilon}(x_1)$ and $\sgn_{\epsilon}(x_2)$ are replaced by $-1$ and $1$ respectively, and $\mathbf{x}$ be the solution trajectory of the original system with $\sgn_{\epsilon}$. Then $\tilde{\mathbf{x}}$ is described by Lemma~\ref{lem:stay-in-quadrant}. If $(a,b,\mathbf{x}(0))$ satisfies condition PD, then by Lemma~\ref{lem:stay-in-quadrant}, $\tilde{\mathbf{x}}$ stays strictly in the $(-,+)$ quadrant with $\tilde{x}_1(t) < 0, \tilde{x}_2(t) > 0$ for all $t$. For any $\epsilon$ that is less than the minimum distance between the $\tilde{\mathbf{x}}$ trajectory and the two axes, the solution $\mathbf{x}$ coincides with $\mathbf{x}$ as $\sgn_{\epsilon}(x_1) = -1, \sgn_{\epsilon}(x_2) = 1$ always holds, and thus $x_1(t) < 0, x_2(t) >0$ for all $t$, proving the first part of the theorem. 

Now assume that $(a,b,\mathbf{x}(0))$ satisfies condition CO. If $x_1(0)x_2(0) > 0$ then the trajectory trivially converges to a consensus equilibrium, so we will assume $x_1(0)x_2(0) < 0$ from now on, and without loss of generality also assume $x_1(0) < 0, x_2(0) > 0$. By Lemma~\ref{lem:stay-in-quadrant}, the trajectory $\tilde{\mathbf{x}}$ does not stay in the $(-,+)$ quadrant and either crosses the positive y-axis at $x_1=0,x_2 > b/a$ or the negative x-axis at $x_2=0,x_1 < -b/a$. By symmetry, and without loss of generality, assume the former. Since the trajectory is continuous, for sufficiently small $\epsilon > 0$, the trajectory $\tilde{\mathbf{x}}$ will still touch the $\epsilon$-band around the positive y-axis at $x_1=0, x_2 > b/a$. The trajectories $\tilde{\mathbf{x}}$ and $\mathbf{x}$ coincide up until the trajectory touches the $\epsilon$-band. After that, the trajectory $\mathbf{x}$ is described by Lemma \ref{lem:cross-eps-band}, and the lemma shows that the trajectory must cross into the region with $x_1 = \epsilon, x_2 > \epsilon$, after which it converges to the $(1,1)$ consensus equilibrium, as stated in the theorem.
\end{proof}

\subsection{Proofs of Supporting Lemmas in the Proof of Theorem~\ref{thm:equal-a-init-char}}\label{ap:subsec:supporting-lemmas}

\begin{proof}[Proof of Lemma~\ref{lem:stay-in-quadrant}]
We will only consider the first case with $x_1(0)< 0,x_2(0)>0$ and the first system of equations; the second one is entirely analogous. Let $x_1(0)=-u,x_2(0)=v$ with $u,v > 0$. The above system is a standard linear system of ODE with initial conditions, which can be solved to get
\begin{align}
    x_1(t) &= -\frac{b}{2a+b} - e^{-bt} \frac{(u-v)}{2} - e^{-(2a+b)t} \frac{(2a+b)(u+v)-2b}{2(2a+b)} \label{eqn:x1-u-v}\\
    x_2(t) &= \frac{b}{2a+b} - e^{-bt} \frac{(u-v)}{2} + e^{-(2a+b)t} \frac{(2a+b)(u+v)-2b}{2(2a+b)}. \label{eqn:x2-u-v}
\end{align}

With algebra and careful case analysis, we can show the following:
\begin{lemma}\label{lem:alg-pd-ec}
Let $x_1(t)$ and $x_2(t)$ be given as (\ref{eqn:x1-u-v}) and (\ref{eqn:x2-u-v}). Then, under condition~PD, $\max_{t \geq 0} x_1(t) < 0$ and $\min_{t \geq 0} x_2(t) > 0$, while under condition CO, either $\max_{t \geq 0} x_1(t) > 0$ or $\min_{t \geq 0} x_2(t) < 0$.
\end{lemma}

The proof of Lemma~\ref{lem:alg-pd-ec} is given at the end of this section. Lemma~\ref{lem:alg-pd-ec} implies that under condition~PD, $x_1(t) < 0$ and $x_2(t) > 0$ for all $t \geq 0$. Similarly, under condition CO, either $\max_{t \geq 0} x_1(t) > 0$ or $\min_{t \geq 0} x_2(t) < 0$. If $\max_{t \geq 0} x_1(t) > 0$, then the trajectory crosses the positive y-axis from the $(-,+)$ to the $(+,+)$ quadrant and at the first time $t'$ that it crosses, the $x_1$-derivative there must be strictly positive in order to reach the interior of the $(+,+)$ quadrant. There, $x_1(t')=0$ and $\dot{x}_1(t') = a(x_2(t')-x_1(t'))+b(-1-x_1(t')) = a x_2(t') - b > 0$, so $x_2(t') > b/a$. Similarly, if $\min_{t \geq 0} x_2(t) < 0$, then the trajectory crosses the negative x-axis from the $(-,+)$ to the $(-,-)$ quadrant and at the first time $t'$ that it crosses, the $x_2$-derivative there must be strictly positive in order to reach the interior of the $(-,-)$ quadrant. There, $x_2(t')=0$ and $\dot{x}_2(t') = a(x_1(t')-x_2(t'))+b(1-x_2(t')) = a x_1(t') + b < 0$, so $x_2(t') < - b/a$.
\end{proof}

\begin{proof}[Proof of Lemma \ref{lem:cross-eps-band}]
Without loss of generality we assume $a=1$ throughout (as we can scale the solution by $a$ and replace $b$ with $b/a$).
We solve the linear ODE and get
\begin{align*}
    \mathbf{x}(t) = c_{+} e^{\lambda_{+}t} \mathbf{v}_{+} + c_{-} e^{\lambda_{-}t} \mathbf{v}_{-} + \mathbf{v}_{0},
\end{align*}
where
\begin{align*}
    \mathbf{v}_{0} = \frac{1}{b+1-b\epsilon-2\epsilon} \begin{pmatrix}
    -\epsilon \\ b-\epsilon-b\epsilon
    \end{pmatrix}, \quad \lambda_{\pm} = \frac{b-2\epsilon-2b\epsilon\pm \sqrt{b^2+4\epsilon^2} }{2\epsilon}, \quad \mathbf{v}_{\pm} = \begin{pmatrix}  (b \pm \sqrt{b^2+4\epsilon^2})/2\epsilon  \\ 1 \end{pmatrix}.
\end{align*}
$(\lambda_{+}, \mathbf{v}_{+}), (\lambda_{-}, \mathbf{v}_{-})$ are eigenvalue-eigenvector pairs of the matrix
\begin{align*}
    \begin{pmatrix}
    b/\epsilon - b - 1 & 1 \\
    1 & -b-1
    \end{pmatrix}.
\end{align*}
We get $c_{\pm}$ from solving the system of equations
\begin{align*}
    c_{+}\left( \frac{b+\sqrt{b^2+4\epsilon^2}}{2\epsilon} \right) + c_{-} \left( \frac{b-\sqrt{b^2+4\epsilon^2}}{2\epsilon} \right) &= -\epsilon+\frac{\epsilon}{b+1-b\epsilon-2\epsilon} \\
    c_{+} + c_{-} &= x_2(0) - \frac{b-\epsilon-b\epsilon}{b+1-b\epsilon-2\epsilon},
\end{align*}
which gives
\begin{align*}
    c_{+} &= -\frac{\epsilon}{\sqrt{b^2+4\epsilon^2}} \left( \epsilon - \frac{\epsilon}{1+b-3\epsilon} + \frac{\sqrt{b^2+4\epsilon^2}-b}{2\epsilon} \left( \frac{b-\epsilon-b\epsilon}{1+b-2\epsilon-b\epsilon} -x_2(0) \right) \right) \\
    c_{-} &= -\frac{b-\epsilon-b\epsilon}{1+b-2\epsilon-b\epsilon} + \frac{\epsilon}{\sqrt{b^2+4\epsilon^2}} \left( \epsilon - \frac{\epsilon}{1+b-3\epsilon} + \frac{\sqrt{b^2+4\epsilon^2}-b}{2\epsilon} \left( \frac{b-\epsilon-b\epsilon}{1+b-2\epsilon-b\epsilon} -x_2(0) \right)  \right)+x_2(0).
\end{align*}
The coefficients of $x_2(0)$ in $c_{+}$ and $c_{-}$ are $\frac{\sqrt{b^2+4\epsilon^2}-b}{2\sqrt{b^2+4\epsilon^2}} > 0$ and $\frac{\sqrt{b^2+4\epsilon^2}+b}{2\sqrt{b^2+4\epsilon^2}} > 0$ respectively, so they are lower-bounded by when $x_2(0)$ is replaced by $b$ because $x_2(0) > b$. We then compute the Taylor expansions of the lower bound around $\epsilon$ at $\epsilon = 0$ to get
\begin{align*}
    c_{+} &> -\frac{\epsilon}{\sqrt{b^2+4\epsilon^2}} \left( \epsilon - \frac{\epsilon}{1+b-3\epsilon} + \frac{\sqrt{b^2+4\epsilon^2}-b}{2\epsilon} \left( \frac{b-\epsilon-b\epsilon}{1+b-2\epsilon-b\epsilon} -b \right) \right) \\
    &= \frac{1+3b}{b^2(1+b)^2} \epsilon^3 + O(\epsilon^4)
    \\
    c_{-} &> -\frac{b-\epsilon-b\epsilon}{1+b-2\epsilon-b\epsilon} + \frac{\epsilon}{\sqrt{b^2+4\epsilon^2}} \left( \epsilon - \frac{\epsilon}{1+b-3\epsilon} + \frac{\sqrt{b^2+4\epsilon^2}-b}{2\epsilon} \left( \frac{b-\epsilon-b\epsilon}{1+b-2\epsilon-b\epsilon} -b \right)  \right)+b \\
    &= \frac{b^2}{1+b} + O(\epsilon)
\end{align*}
The leading coefficients of these expressions (those with the smallest $\epsilon$ power) are positive, so $c_{+}, c_{-} > 0$ for all sufficiently small $\epsilon > 0$. Let $\mathbf{v}_{+} = (v_{+,1},v_{+,2})^\top, \mathbf{v}_{-} = (v_{-,1},v_{-,2})^\top, \mathbf{v}_{0} = (v_{0,1},v_{0,2})^\top$. We also note that $ \lambda_{+},v_{+,1} > 0$, $\lambda_{-}, v_{-,1} < 0$ and for sufficiently small $\epsilon$, $v_{0,1} < 0, v_{0,2} > \epsilon$.
Therefore, for those $\epsilon$,
\begin{align*}
    x_2(t) &=  c_{+} e^{\lambda_{+}t}  + c_{-} e^{\lambda_{-}t}  + v_{0,2} > \epsilon \\
    \dot{x}_1(t) &= c_{+} v_{+,1} \lambda_{+} e^{\lambda_{+}t}  + c_{-} v_{-,1} \lambda_{-} e^{\lambda_{-}t}  > 0
\end{align*}
Therefore, throughout the trajectory, $x_2$ is always bounded below by $\epsilon$ and $x_1$ always increases until time $t'$ such that $x_1(t') = \epsilon$, and we are done.
\end{proof}


\begin{proof}[Proof of Lemma~\ref{lem:alg-pd-ec}]

We start with the expressions (\ref{eqn:x1-u-v}) and (\ref{eqn:x2-u-v}) for the trajectory stated in the main proof:
\begin{align*}
    x_1(t) &= -\frac{b}{2a+b} - e^{-bt} \cdot \frac{u-v}{2}  - e^{-(2a+b)t} \cdot \frac{ (2a+b)(u+v)-2b }{2(2a+b)} \quad &(\ref{eqn:x1-u-v})
    \\
    x_2(t) &= \frac{b}{2a+b} - e^{-bt} \cdot \frac{u-v}{2}  + e^{-(2a+b)t} \cdot \frac{  (2a+b)(u+v)-2b}{2(2a+b)}. \quad &(\ref{eqn:x2-u-v}) 
\end{align*}

Note that this solution implies $x_1(0) = -u < 0, x_2(0) = v > 0, x_1(\infty) = -\frac{b}{2a+b} < 0, x_2(\infty) = \frac{b}{2a+b} > 0$.

We can compute the time derivatives
\begin{align*}
    \dot{x}_1(t) =  e^{-bt} \cdot \frac{b(u-v)}{2}  + e^{-(2a+b)t} \cdot \frac{  (2a+b)(u+v)-2b }{2}
    \\
    \dot{x}_2(t) = e^{-bt} \cdot \frac{b(u-v)}{2}  - e^{-(2a+b)t} \cdot \frac{  (2a+b)(u+v)-2b }{2}
\end{align*}

Now we consider 4 cases. These cases are exhaustive up to sets of measure zero.

\textbf{Case 1} $(2a+b)(u+v)-2b < 0$, $u-v < 0$



Then $\dot{x}_1 \leq 0$ always, so $x_1(t) \leq x_1(0) < 0$ for all $t$. 

Using the expression for $\dot{x}_2(t)$, we see that
\begin{align*}
    \dot{x}_2(t) > 0 \Longleftrightarrow t < \frac{1}{2a} \log \left( \frac{2b-(2a+b)(u+v)}{b(v-u)}  \right) \equiv t^*
\end{align*}
So $x_2$ increases up to $t^*$ and decreases afterward, which implies $x_2(t) \geq \min(x_2(0), x_2(\infty)) > 0$. So we always reach persistent disagreement in this case.

\textbf{Case 2} $(2a+b)(u+v)-2b < 0$, $u-v > 0$

Then $\dot{x}_2(t) \geq 0$ always, so $x_2(t) \geq x_2(0) > 0$.


Using the expression for $\dot{x}_1(t)$, we see that
\begin{align*}
    \dot{x}_1(t) > 0 \Longleftrightarrow t > \frac{1}{2a} \log \left( \frac{2b-(2a+b)(u+v)}{b(u-v)} \right)
\end{align*}
Similarly to Case 1, $x_1(t) < \max(x_1(0),x_1(\infty)) < 0$, so we always reach persistent disagreement in this case.

\textbf{Case 3} $(2a+b)(u+v)-2b > 0$, $u-v < 0$


The expression of $x_2$ shows that $x_2(t) > 0$ for all $t$. We have

\begin{align*}
    \dot{x}_1(t) > 0 \Longleftrightarrow t < \frac{1}{2a} \log \left( \frac{(2a+b)(u+v)-2b}{b(v-u)} \right) \equiv t^*
\end{align*}

\textit{Case 3.1} If $(2a+b)(u+v)-2b < b(v-u)$, then $t^* < 0$, and $\dot{x}_1(t) < 0 $ for all $t \geq 0$ so we always reach persistent disagreement. 

\textit{Case 3.2} If $(2a+b)(u+v)-2b > b(v-u)$, then $t^* > 0$, and the condition $\max_{t} x_1(t) < 0$ is equivalent to $x_1(t^*) < 0$.

We first compute
\begin{align*}
    x_1(t^*) = -\frac{b}{2a+b} + e^{-(2a+b)t^*} \left( \frac{b}{2a+b} - \frac{u+v}{2} + \frac{1}{2}(v-u)e^{2at^*} \right)
\end{align*}

We then plug in $e^{2at^*} b(v-u) = (2a+b)(u+v)-2b$ to get

\begin{align*}
    x_1(t^*) = -\frac{b}{2a+b} + e^{-(2a+b)t^*} \cdot \frac{a((2a+b)(u+v)-2b)}{b(2a+b)}
\end{align*}

By our assumption the second term is positive. We write $e^{-(2a+b)t^*} = (e^{2at^*})^{-(2a+b)/(2a)}$. After some algebraic manipulations, we get that $x_1(t^*) < 0$ if and only if
\begin{align*}
    (2a+b)(u+v)-2b > b^{1-2a/b} a^{2a/b} (v-u)^{1+2a/b}
\end{align*}



\textbf{Case 4} $(2a+b)(u+v)-2b > 0$, $u-v > 0$


The expression of $x_1$ shows that $x_1(t) < 0$ for all $t$ so we are done here.

Similarly to Case 3,

\begin{align*}
    \dot{x}_2(t) > 0 \Longleftrightarrow t > \frac{1}{2a} \log \left( \frac{(2a+b)(u+v)-2b}{b(u-v)} \right) \equiv t^*
\end{align*}

\textit{Case 4.1} If $(2a+b)(u+v)-2b < b(u-v)$ then $t^* < 0$ and $\dot{x}_2(t) > 0$ for all $t \geq 0$ so $x_2(t) > x_2(0) > 0$ for all $t \geq 0$ and we always reach persistent disagreement. 

\textit{Case 4.2} If $(2a+b)(u+v)-2b > b(u-v)$ then $t^* > 0$ and the condition $\min_{t} x_2(t) > 0$ is equivalent to $x_2(t^*) > 0$. Similarly to Case 3, this is equivalent to
\begin{align*}
    (2a+b)(u+v)-2b > b^{1-2a/b} a^{2a/b} (u-v)^{1+2a/b}
\end{align*}

Combining all cases, we get that under condition PD (Case 1, 2, 3.1, 4.1), that is, either $(2a+b)(u+v) < b|u-v|$ or $(2a+b)(u+v)-2b > b^{1-2a/b} a^{2a/b} |u-v|^{1+2a/b}$, we have $\max_{t} x(t) < 0$ and $\min_{t} x_2(t) > 0$. Under condition CO (Case 3.2, 4.2), that is, $b|u-v| < (2a+b)(u+v)-2b < b^{1-2a/b} a^{2a/b} |u-v|^{1+2a/b}$, we have $\max_{t} x(t) > 0$ or $\min_{t} x_2(t) < 0$. The other case, $x_1(0) > 0$ and $x_2(0) < 0$, leads to the same conditions, and we are done.

Lastly, we note that condition CO can hold only when $b^{1-2a/b} a^{2a/b} |u-v|^{1+2a/b} > b|u-v|$ or $|u-v| > b/a$, which is only possible if $b < a$. 

\end{proof}

\subsection{Proof of Proposition~\ref{prop:polar-monotonic}}

\begin{proof}[Proof of Proposition~\ref{prop:polar-monotonic}]
Let $x_1(0)=-u$ and $x_2(0)=v$ for $u,v > 0$ (the other case is similar). Conditional on reaching persistent disagreement, the trajectory stays in the $(-,+)$ quadrant and the expressions (\ref{eqn:x1-u-v}) and (\ref{eqn:x2-u-v}) for $x_1(t)$ and $x_2(t)$ from the proof of Theorem \ref{thm:equal-a-init-char} holds for all time. 
Let $y(t) = |x_2(t)-x_1(t)| = x_2(t)-x_1(t)$ be the polarization. We get
\begin{align*}
    y(t) = \frac{2b}{2a+b} + e^{-(2a+b)t} \left( y_0 - \frac{2b}{2a+b} \right) > 0
\end{align*}
where $y_0 = u+v$ is the initial polarization. Since $e^{-(2a+b)t}$ is decreasing in $t$, this expression shows that if $y_0 < 2b/(2a+b)$ then polarization $y(t)$ is increasing in $t$, while if $y_0 > 2b/(2a+b)$ then polarization is decreasing in $t$.
\end{proof}

\subsection{Proofs of Theorem \ref{Thm:Main} and Corollary \ref{Cor:polarization} }

We now prove Theorem \ref{Thm:Main}.
We denote the set of left agents by $\mathcal{L}_{n}$ and the set of right agents by $\mathcal{R}_{n}$. Let $L_{n,i}$ be the number of left connections that agent $i$  has and let $R_{n,i}$ be the number of right connections that agent $i$ has. 
 
 
 We need the following two lemmas to prove the theorem.
\begin{lemma} \label{Lemma:Concent}
 Consider the sets 
$$ S_{\delta,n} =  [(1-\delta) p(n) n , (1+\delta) p(n) n]\text{ and }   D_{\delta,n} =  [(1-\delta)q(n) n , (1+\delta) q(n) n ]. $$
Then 
$$\lim _{n \rightarrow \infty} \mathbb{P}(R_{n,i} \notin D_{\delta,n} , \ \forall i \in \mathcal{L}_{n} )= \lim _{n \rightarrow \infty} \mathbb{P}(L_{n,i} \notin S_{\delta,n}, \ \forall i \in \mathcal{L}_{n} ) = 0  $$
and 
$$\lim _{n \rightarrow \infty} \mathbb{P}(R_{n,i} \notin S_{\delta,n}, \ \forall i \in \mathcal{R}_{n} )= \lim _{n \rightarrow \infty} \mathbb{P}(L_{n,i} \notin D_{\delta,n} , \ \forall i \in \mathcal{R}_{n} ) = 0  $$
for all $\delta \in (0,1)$. 
\end{lemma}

\begin{proof}[Proof of Lemma \ref{Lemma:Concent}]
Let $\delta \in (0,1)$, $n \geq 1$ and consider a left agent $i \in \mathcal{L}_{n}$. Note that the number of left connections $L_{n,i}$ that agent $i$ has is a binomial random variable with a probability of success $p (n)$ and $n$ trials. Applying Theorem 12.6 in \cite{blum2020foundations} we have 
\begin{equation}
    \mathbb{P}(L_{n,i} \notin [(1-\delta) p(n) n , (1+\delta) p(n) n] )\leq 3 \exp (-\delta ^{2} p (n) n /8 ).
\end{equation}
Hence, using the union bound, we have
$$\mathbb{P}(L_{n,i} \notin [(1-\delta)p(n) n , (1+\delta) p(n) n ] , \ \forall i \in \mathcal{L}_{n} ) \leq 3n \exp (-\delta ^{2} p(n) n  /8 ) .$$ 
The fact that $p(n)$ is $\omega(\ln (n) / n)$ implies that $3n \exp (-\delta ^{2} p(n) n  /8 )$ converges to $0$ as $n \rightarrow \infty$.  To see this note that 
$$n \exp (-\delta ^{2} p(n) n  /8 ) = \exp ( \ln(n) (1 - \delta ^{2} p(n) n  /8 \ln(n) )) = n^{1 - \delta ^{2}p(n)n/8\ln(n) } $$ 
converges to $0$ when $n$ and $p(n) n / \ln(n)$ converge to $\infty$. 

Similarly, we can deduce that 
$$\mathbb{P}(R_{n,i} \notin [(1-\delta)q(n) n , (1+\delta) q(n) n ] , \ \forall i \in \mathcal{L}_{n} ) \leq 3n \exp (-\delta ^{2} q(n) n  /8 ) .$$ 
converges to $0$ as $n \rightarrow \infty$. The proof of the lemma for the right agents follows from the same arguments as the above arguments for the left agents. 
\end{proof}

Define the event
 \begin{equation} \label{Eq: set} C_{\delta,n} = \{A_{n} \in \boldsymbol{A}_{n}: R_{n,i} \in D_{\delta,n}, L_{n,i} \in S_{\delta,n}, \ \forall i \in \mathcal{L}_{n}, \text{ and } R_{n,j} \in S_{\delta,n}, L_{n,j} \in D_{\delta,n}, \ \forall j \in \mathcal{R}_{n}\} .
 \end{equation}
 That is, $C_{\delta,n}$ consists of all adjacency matrices such that every agent has between $(1-\delta) p(n)n$ and $(1+\delta) p(n)n$  connections of the same type, and between $(1-\delta) q(n)n$ and $(1+\delta) q(n)n$ connections of a different type. 
 The proof of Lemma \ref{Lemma:Concent} and the union bound imply that for every $\delta \in (0,1)$ we have $\lim_{n \rightarrow \infty} \mathbb{P} (C_{\delta,n} )= 1$.

For $0<\delta<1$ and an integer $n$, consider the following differential equations:
 \begin{equation} \label{eq:bound}
     \begin{aligned}
       \dot {\bar{ x}}_L (t) &=  
       a(\bar{x}_R (t) - \bar{ x}_L(t)) \frac{(1+\delta)q(n)}{(1-\delta)(p(n)+q(n))} + 
       b(sgn_\epsilon(\bar{ x}_L(t)) - \bar{ x}_L(t))  \\
       \underline{\dot x}_L(t) &=  
       a(\underline{x}_R(t) - \underline{x}_L(t)) \frac{(1-\delta)q(n)}{(1+\delta)(p(n)+q(n))} + b(sgn_\epsilon(\underline{x}_L(t)) - \underline{x}_L(t))  \\
\dot{\bar{x}}_R(t) &= 
       a(\bar{x}_L(t) - \bar{x}_R(t)) \frac{(1-\delta)q(n)}{(1+\delta)(p(n)+q(n))} + 
       b(sgn_\epsilon(\bar{x}_R(t)) - \bar{x}_R(t)) \\
      \dot{\underline{x}}_R(t) &= 
       a(\underline{x}_L(t) - \underline{x}_R(t)) \frac{(1+\delta)q(n)}{(1-\delta)(p(n)+q(n))} + b(sgn_\epsilon(\underline{x}_R(t)) - \underline{x}_R(t))  
      \end{aligned}
 \end{equation}
  with the initial conditions $\bar{x}_{i} (0) = \max X_{i}$ and $\underline{x}_{i} (0) = \min X_{i}$ for $i=L,R$, i.e., $\bar{x}_R(0)$ ($\bar{x}_{L}(0)$) is the highest right (left) possible initial opinion and $\underline{x}_{R} (0)$ ($\underline{x}_{L}(0)$) is the lowest right (left) possible initial opinion.

Note that the existence of a solution $\boldsymbol{x}_{n,A_{n},\boldsymbol{x}_{0}}(t)$ for any realized adjacency matrix $A_{n}$ and initial opinion vector $\boldsymbol{x}_{0}$ follows from standard differential equation theory (e.g., \cite{cortes08-discontinuous-dynamical-systems-tutorial}). 


Let  $$a_{1} = \frac{a(1-\delta)q(n)} {(1+\delta)(p(n)+q(n))} \text{ and } a_{2}=\frac{a(1+\delta)q(n)} {(1-\delta)(p(n)+q(n))}.$$ 

The following lemma shows that the solutions of the differential equations given in (\ref{eq:bound}) bound the  agents' opinions $\boldsymbol{x}_{n,A_{n}, \boldsymbol{x}_{0}}(t)$ for adjacency matrices $A_{n}$ that belong to $C_{\delta,n}$ and initial opinions that belong to $(X_{L} \times X_{R})^{n}$.

\begin{lemma} \label{Lemma: bounds} Suppose that $  {\bar{ x}}_L(t)  < {\underline{ x}}_R(t) $ for $t \geq 0$. 

(i) We have ${\underline{ x}}_R(t) \leq  {\bar{ x}}_R(t)$ and ${\underline{ x}}_L(t) \leq {\bar{ x}}_L(t)$ for all $t \geq 0  $. 

(ii)  For all $A_{n} \in C_{\delta,n}$, $\boldsymbol{x}_{0} \in (X_{L} \times X_{R})^{n}$ and all $t \geq 0 $ we have
 \begin{equation} \label{Eq: Lemma2}
 {\underline{ x}}_L(t) \leq  x_{i,n,A_{n}, \boldsymbol{x}_{0}}(t) \leq   {\bar{ x}}_L(t), \ \forall i \in \mathcal{L}_{n}, \text{ and }  {\underline{ x}}_R(t) \leq  x_{j,n,A_{n} , \boldsymbol{x}_{0}}(t) \leq   {\bar{ x}}_R(t),  
\ \forall j \in \mathcal{R}_{n}.
\end{equation}
\end{lemma} 

\begin{proof}
(i) Assume in contradiction that there exists a $t \geq 0$ such that   ${\underline{ x}}_R(t) > {\bar{ x}}_R(t)$ or ${\underline{ x}}_L(t) > {\bar{ x}}_L(t)$. Note that $t>0$.

Let $t_{1} = \inf \{t \in [0,\infty): {\underline{ x}}_R(t) > {\bar{ x}}_R(t) \text{ or }{\underline{ x}}_L(t) > {\bar{ x}}_L(t) \}$. By  the contradiction assumption and the continuity of the solutions of the ordinary differential equations given in Equation (\ref{eq:bound}),  $t_{1}$ is finite, and we have  $ {\underline{ x}}_R(t_{1}) = {\bar{ x}}_R(t_{1})$ and ${\underline{ x}}_L(t_{1}) \leq {\bar{ x}}_L(t_{1})$ or ${\underline{ x}}_L(t_{1}) = {\bar{ x}}_L(t_{1})$ and  $ {\underline{ x}}_R(t_{1}) \leq {\bar{ x}}_R(t_{1})$. 
Assume without loss of generality that $ {\underline{ x}}_R(t_{1}) = {\bar{ x}}_R(t_{1})$ and ${\underline{ x}}_L(t_{1}) \leq {\bar{ x}}_L(t_{1})$. We have
\begin{align*}
    \dot{\bar{x}}_R(t_{1}) &= 
       a(\bar{x}_L(t_{1}) - \bar{x}_R(t_{1})) \frac{(1-\delta)q(n)}{(1+\delta)(p(n)+q(n))} + 
       b(sgn_\epsilon(\bar{x}_R(t_{1})) - \bar{x}_R(t_{1})) \\
      & >
       a(\underline{x}_L(t_{1}) - \underline{x}_R(t_{1})) \frac{(1+\delta)q(n)}{(1-\delta)(p(n)+q(n))} + b(sgn_\epsilon(\underline{x}_R(t_{1})) - \underline{x}_R(t_{1}))  =  \dot{\underline{x}}_R(t_{1})
\end{align*}
so that $\underline{ x}_R(t) < {\bar{ x}}_R(t)$ for $t > t_{1}$ that is close enough to $t_{1}$. 

From continuity, if ${\underline{ x}}_L(t_{1}) < {\bar{ x}}_L(t_{1})$ then ${\underline{ x}}_L(t) < {\bar{ x}}_L(t)$ for $t > t_{1}$  that is close enough to $t_{1}$. If  ${\underline{ x}}_L(t_{1}) = {\bar{ x}}_L(t_{1})$ then 
\begin{align*}
    \dot{\bar{x}}_L(t_{1}) &= 
       a(\bar{x}_R(t_{1}) - \bar{x}_L(t_{1})) \frac{(1+\delta)q(n)}{(1-\delta)(p(n)+q(n))} + 
       b(sgn_\epsilon(\bar{x}_R(t_{1})) - \bar{x}_R(t_{1})) \\
      & >
        a(\underline{x}_R(t_{1}) - \underline{x}_L(t_{1})) \frac{(1-\delta)q(n)}{(1+\delta)(p(n)+q(n))} + 
       b(sgn_\epsilon(\bar{x}_R(t_{1})) - \underline{x}_R(t_{1})) =  \dot{\underline{x}}_L(t_{1})
\end{align*}
so that $\underline{ x}_L(t) < {\bar{ x}}_L(t)$ for $t > t_{1}$ that is close enough to $t_{1}$. 
Hence, $\underline{ x}_R(t) < {\bar{ x}}_R(t)$ and ${\underline{ x}}_L(t) < {\bar{ x}}_L(t)$ for $t > t_{1}$ that is close enough to $t_{1}$, which is a contradiction to the definition of $t_{1}$.

(ii) Let $A_{n} \in C_{\delta,n}$ and $\boldsymbol{x}_{0} \in (X_{L} \times X_{R} ) ^{n}$ . The claim of the lemma  holds for $t=0$ as $ x_{i,n,A_{n},\boldsymbol{x}_{0} }(0) $ is in $X_{L}$ or in $X_{R}$. Suppose that the claim holds for some $t>0$. Suppose that agent $i$ has the highest expected opinion at time $t$ (hence, agent $i$ is a right agent), i.e., $x_{i,n,A_{n}, \boldsymbol{x}_{0}
}(t) \geq x_{j,n,A_{n}, \boldsymbol{x}_{0}
}(t) $ for all $j \in \mathcal{L}_{n} \cup \mathcal{R}_{n}$. We have
\begin{align*}
    \dot{x}_{i,n,A_{n}, \boldsymbol{x}_{0}}(t) &= \frac{a}{|N(i)|}\sum_{j\in N(i)} (x_{j,n,A_{n}, \boldsymbol{x}_{0}}(t)-x_{i,n,A_{n}, \boldsymbol{x}_{0}}(t)) +b \left (  sgn_{\epsilon}(x_{i,n,A_{n},\boldsymbol{x}_{0}}(t))-x_{i,n,A_{n}, \boldsymbol{x}_{0}}(t) \right) \\
    & \leq  \frac{a}{|N(i)|}\sum_{j\in N(i) \cap \mathcal{L}_{n}} (x_{j,n,A_{n}, \boldsymbol{x}_{0}}(t)-x_{i,n,A_{n}, \boldsymbol{x}_{0}}(t)) +b \left (  sgn_{\epsilon}(x_{i,n,A_{n},\boldsymbol{x}_{0}}(t))-x_{i,n,A_{n}, \boldsymbol{x}_{0}}(t) \right) \\
    & \leq  \frac{a|N(i)\cap \mathcal{L}_{n} | }{|N(i)|}(\bar{ x}_L(t)-x_{i,n,A_{n}, \boldsymbol{x}_{0}}(t)) +b \left ( sgn_{\epsilon}(x_{i,n,A_{n},\boldsymbol{x}_{0}}(t))-x_{i,n,A_{n}, \boldsymbol{x}_{0}}(t) \right)\\
    & \leq a \frac{(1-\delta)q(n)}{(1+\delta)(p(n)+q(n))}({\bar{ x}}_L(t)-x_{i,n,A_{n}, \boldsymbol{x}_{0}}(t)) + 
     b \left ( sgn_{\epsilon}(x_{i,n,A_{n},\boldsymbol{x}_{0}}(t))-x_{i,n,A_{n}, \boldsymbol{x}_{0}}(t) \right).
    \end{align*} 
    The first inequality follows because $x_{i',n,A_{n}, \boldsymbol{x}_{0}}(t)-x_{i,n,A_{n}, \boldsymbol{x}_{0}}(t) \leq 0 $ for every right agent $i' \in \mathcal{R}_{n}$. The second inequality follows because $x_{j,n,A_{n}, \boldsymbol{x}_{0}}(t) \leq \bar{x}_{L}(t)$ for every left agent $j \in \mathcal{L}_{n}$. The third inequality follows because
    $A_{n} \in C_{\delta,n}$ implies that the cardinality of $N(i) \cap \mathcal{L}_{n}$ is at least $(1-\delta)q(n)n$ and the cardinality of $N(i)$ is at most $(1+\delta)(p(n)+q(n))n$ and because ${\bar{ x}}_L(t) < x_{i,n,A_{n}}(t) $. 
    
    We conclude that $\dot{x}_{i,n,A_{n},\boldsymbol{x}_{0}}(t) \leq \bar{U}(t,x_{i,n,A_{n}, \boldsymbol{x}_{0}}(t))$ where 
    $$  \bar{U}(t, y)= a(\bar{x}_L(t) - y) \frac{(1-\delta)q(n)}{(1+\delta)(p(n)+q(n))} +
       b(sgn_\epsilon(y) - y) .$$
    
     Theorem 4.1 in \cite{Hartman2002} implies that $x_{i,n,A_{n}, \boldsymbol{x}_{0}}(t) \leq \bar{ x}_R(t)$ on $[t, t+\delta']$ for some $\delta' >0$ that does not depend on $t$. An analogous argument proves that the other inequalities in the statement of the lemma (see Equation (\ref{Eq: Lemma2}))  hold on an interval $[t, t+\delta'']$ for some $\delta'' >0$. 
     We conclude that the inequalities in the statement of Lemma (\ref{Eq: Lemma2}) hold for all $t  \geq 0$, and the proof is complete. 
\end{proof}

We now show that the above two lemmas imply Theorem \ref{Thm:Main}. 

Suppose first that the two-agent system $(a \beta ,b, (x_{L},x_{R}) )$ satisfies the PD condition, has a solution $(x_{1}(t),x_{2}(t))$, and an equilibrium $\boldsymbol{e}=(e_{L},e_{R})$ and that $X_{L} \times X_{R} \subseteq X_{PD}(a\beta , b)$. From the proof of Theorem 1 we have $\boldsymbol{e}=(e_{L},e_{R}) = (-b/(b+2a\beta)), b/(b+2a\beta) )$
 
 By the Theorem's assumption, we can find a small $\delta>0$ and a large $N$ such that for all $n \geq N$  and all $\epsilon'' >0$ we have $a_{1}:= a(1-\delta) q(n)/(1+\delta)(p(n)+q(n) \geq a \beta - \epsilon''$ and $a_{2} := a(1+\delta) q(n)/(1-\delta)(p(n)+q(n) \leq a \beta +\epsilon''$. Note that this implies that $X_{L} \times X_{R} \subseteq X_{PD}(a_{i} , b)$ for $i=1,2$ for small enough $\delta$ and large $n$. 
 
 Consider the following differential equations:
 \begin{equation} \label{eq:bd:positive}
     \begin{aligned}
       \dot {\bar{ x}}_L' (t) &=  
       a_{2}(\bar{x}_R '(t) - \bar{ x}_L'(t))  + 
       b(-1 - \bar{ x}_L'(t))  \\
\dot{\bar{x}}_R'(t) &= 
       a_{1}(\bar{x}_L'(t) - \bar{x}_R'(t))  + 
       b(1 - \bar{x}_R'(t))  
      \end{aligned}
 \end{equation}
 with the initial conditions $\bar{x}_L'(0)=\bar{x}_{L}(0)$
 and
 $\bar{x}_R'(0)=\bar{x}_{R}(0)$ so $(\bar{x}_{L}(0) , \bar{x}_{R}(0)) \in X_{PD} (a_{i},b)$ for $i=1,2$. From standard differential equation theory, the equilibrium of the differential equations given in Equation (\ref{eq:bd:positive})  above does not depend on the initial opinions and is given by $(e_{L}',e_{R}')$ and is continuous in the parameters $(a_{1},a_{2},b)$. In addition, the solution of the differential equation converges to the equilibrium. Thus, for all $\delta' > 0$, we can find a $T' >0$ such that $ |\bar{ x}_L' (t) - e_{L}' | \leq \delta' / 2$ and $ |\bar{ x}_R' (t) - e_{R}' | \leq \delta' /2$ for all $t \geq T'$, and a small $\epsilon''$ (i.e., a large $n$ and small $\delta$) such that $|e_{L}-e_{L}'| \leq \delta'/2$, $|e_{R}-e_{R}'| \leq \delta'/2$. Hence, $ |\bar{ x}_L' (t) - e_{L} | \leq \delta' $ and $ |\bar{ x}_R' (t) - e_{R} | \leq \delta' $ for all $t \geq T'$ and small enough $\epsilon''$. 
 
 
  The PD condition and the proof of Theorem \ref{thm:equal-a-init-char} imply that the trajectories $\bar{ x}_L (t)$ and $\bar{ x}_R (t)$  are bounded away from zero for all $t$. We conclude that  $\bar{ x}_L (t)$ and $\bar{ x}_R (t)$ satisfy the  differential equations given in Equation (\ref{eq:bd:positive}) for a small $\epsilon>0$.  Hence, $ |\bar{ x}_L (t) - e_{L} | \leq \delta' $ and $ |\bar{ x}_R (t) - e_{R} | \leq \delta' $ for all $t \geq T'$ and small enough $\epsilon''$.  A similar argument (using Equation (\ref{eq:bd:positive}) with initial opinions $\underline{x}_{L}(0)$ and $\underline{x}_{R}(0)$) shows that  $ |\underline{ x}_L (t) - e_{L} | \leq \delta' $ and $ |\underline{ x}_R (t) - e_{R} | \leq \delta' $ for all $t \geq T'$ and small enough $\epsilon''$.

We can now apply Lemma \ref{Lemma: bounds}  to conclude that for all $\epsilon' >0$, we can find an $N$, $T'$, and $\delta>0$ such that for all $t \geq T'$ and all $A_{n} \in C_{\delta,n}$, $\boldsymbol{x}_{0} \in (X_{L} \times X_{R})^{n}$ with $n \geq N$, we have $|x_{i,n,A_{n}, \boldsymbol{x}_{0}}(t) - e_{L}| \leq \epsilon'$ for every left agent $i$ and $|x_{j,n,A_{n}, \boldsymbol{x}_{0}}(t) - e_{R}| \leq \epsilon'$ for every right agent $j$.


Now suppose that the two-agent system $(a \beta ,b, (x_{L},x_{R}) )$ satisfies the CO condition and has  a solution $(x_{1}(t),x_{2}(t))$ and an equilibrium $\boldsymbol{e}=(e_{L},e_{R}) = (-1,-1)$  and that $X_{L} \times X_{R} \subseteq X_{CON}(a\beta , b)$. 

The proof of Theorem \ref{thm:equal-a-init-char} and the continuity of the solutions of the differential equations imply that  $0 > \bar{x}_{R} (t^{*}) > \bar{x}_{L}(t^{*})$ for some $t^{*}$ and small enough $\delta$ and large $n$. Using the proof of Theorem \ref{thm:equal-a-init-char} again, we have $\bar{x}_{R}(t) < 0$, $\bar{x}_{L} (t) < 0$ for $ t \geq t^{*}$ and we can assume that $\bar{x}_{R}(t) > \bar{x}_{L}(t)$ for $t \leq t^{*}$. As long as $x_{i,n,A_{n},\boldsymbol{x_{0}}}(t) \geq \bar x_{L} (t)$ we can use the same arguments as the arguments in Lemma \ref{Lemma: bounds} part (ii) to show that $x_{i,n,A_{n},\boldsymbol{x_{0}}}(t) \leq \bar{x}_{R} (t)$ for any agent $i$. We conclude that $x_{i,n,A_{n},\boldsymbol{x_{0}}}(t'') < 0$ for any agent for some $t''$ which implies a consensus equilibrium where the opinions of all agents are $-1$. 

Thus, we can find an $N'$, $T$, and $\delta>0$ such that for all $t \geq T$ and all $A_{n} \in C_{\delta,n}$ with $n \geq N'$, we have $|x_{i,n,A_{n},\boldsymbol{x}_{0}}(t) - e_{L}| \leq \epsilon'$ for every left agent $i$ and $|x_{j,n,A_{n}}(t) - e_{R}| \leq \epsilon'$ for every right agent $j$.

The proof for the case that $\boldsymbol{e}=(e_{L},e_{R}) = (1,1)$  and that $X_{L} \times X_{R} \subseteq X_{COP}(a\beta , b)$ follows from a similar argument to the argument above. This completes the proof of the theorem.


\begin{proof}[Proof of Corollary \ref{Cor:polarization}]
Part (i) follows immediately from the proof of Theorem \ref{Thm:Main}. Proposition \ref{prop:polar-monotonic} implies that $x_{1}(t)-x_{2}(t)$ is strictly increasing on $(0,\infty)$, and hence,  $x_{1}(t_{2})-x_{2}(t_{2}) = x_{1}(t_{1}) - x_{2}(t_{1}) + \delta ''$ for some $\delta '' >0$. Choosing $0<\delta ' \leq \delta '' /4$ and using part (i) imply that for every left agent $i$ and every right agent $j$ we have 
\begin{align*} 
x_{i,n,A_{n} } (t_{2}) -  x_{j,n,A_{n} }(t_{2}) & \geq x_{1}(t_{2})-x_{2}(t_{2}) - 2\delta'  \geq x_{1}(t_{1}) - x_{2}(t_{1}) + 2\delta' \geq x_{i,n,A_{n} } (t_{1}) -  x_{j,n,A_{n} }(t_{1}) 
\end{align*}
which proves part (ii). 
\end{proof}

\subsection{Proof of Theorem \ref{thm:opinions-converge} }

Before we give a proof of Theorem \ref{thm:opinions-converge}, we first 
introduce some useful concepts and 
state LaSalle's invariance theorem (adapted from \cite[p.~128]{khalil01-nonlinear-systems}).
A  set $S$ is said to be an invariant set with respect to dynamics $\dot{\mathbf{x}}=f(\mathbf{x})$ if $\mathbf{x}(0)\in S$ implies $\mathbf{x}(t)\in S$ for all $t\in \mathbb{R}$. Similarly, $S$ is said to be a positively invariant set if 
$\mathbf{x}(0)\in S$ implies $\mathbf{x}(t)\in S$ for all $t\in \mathbb{R}_+$.
We say that $\mathbf{x}(t)$ converges to a set $S$ if for every $\epsilon>0$ there is a $T>0$ such that
$\mathrm{dist}(\mathbf{x}(t),S)<\epsilon$ for all $t>T$, where 
$\mathrm{dist}(\mathbf{x},S)=\inf_{\mathbf{y}\in S} \|\mathbf{x}-\mathbf{y} \|_2$.

\begin{theorem}[LaSalle's Invariance Theorem \citep{khalil01-nonlinear-systems}] \label{prop:khalil}
Let $\Omega$ be a compact set that is positively invariant with respect to   the dynamics $\dot{\mathbf{x}}=f(\mathbf{x})$.
Let $V:\Omega \to \mathbb{R}$ be a continuously differentiable function such that $\dot{V}(\mathbf{x})\leq 0$ in $\Omega$. Let $E$ be the set of all points in $\Omega$ where $\dot{V}(\mathbf{x})=0$.
Let $E'$ be the largest invariant set in $E$.
Then,  every solution to $\dot{\mathbf{x}}=f(\mathbf{x})$ starting in $\Omega$ approaches  $E'$ as $t\rightarrow \infty$.
\end{theorem}


\begin{proof}[Proof of Theorem \ref{thm:opinions-converge}]
Since $s$ is continuous, the vector field on the right hand side is continuous, so a standard theorem from ordinary differential equations (see e.g. \cite{cortes08-discontinuous-dynamical-systems-tutorial}) guarantees that the continuous solution $\mathbf{x}(t)$ exists and is unique. Moreover, at $x_i=K$, both terms of the time derivative $\dot{x}_i(t) = \sum_{j} a_{ij} (x_j - x_i) + b_i(s(x_i) - x_i)$ are non-positive, so we always have $x_i(t) \leq K$. Similarly, $x_i(t) \geq -K$ for all $t$. We conclude that $\mathbf{x}(t) \in \Omega$ for all $t\geq 0$ 
and $\Omega$ is a positively invariant set.


Let $M=L+B$. Note that $M$ is symmetric. Consider the Lyapunov function
\begin{align*}
    V(\mathbf{x}) = \frac{1}{2} \mathbf{x}^\top M \mathbf{x} - \sum_{j=1}^{n} b_j \int_{0}^{x_j} s(t) dt.
\end{align*}
Observe that 
\begin{align*}
  \frac{\partial V}{\partial x_i} = (M \mathbf{x})_i - b_i s(x_i),
\end{align*}
and  $\nabla V(\mathbf{x}) = M \mathbf{x} - B s(\mathbf{x}) = -\dot{\mathbf{x}}$.
These observations imply that the time derivative of the Lyapunov function with respect to the trajectory is given by
\begin{align*}
    \dot{V}(\mathbf{x})= \nabla V(\mathbf{x}) \cdot \dot{\mathbf{x}} = -||M \mathbf{x} - B s(\mathbf{x})||^2 \leq 0.
\end{align*}
Note that since $s(\cdot)$ is a continuous function, it also follows that $V(\mathbf{x})$ is continuously differentiable.
Thus, Theorem \ref{prop:khalil} implies that
$\mathbf{x}(t)$ converges to a subset of $E=\{x|M \mathbf{x} = B s(\mathbf{x}) \} $ for any $\mathbf{x}_0\in \Omega$.
Since $E$ is a finite set, this implies that for any $\mathbf{x}_0\in \Omega$, $\mathbf{x}(t)$ converges to a point in $E$, and hence
the limit $\mathbf{x}_\infty \equiv \lim_{t \to \infty} \mathbf{x}(t)$ exists and belongs to $E$.
\end{proof}

\section{Additional Simulations} \label{Appendix:Simulations}

In this section we provide simulations for a non-normalized stochastic block model. In the non-normalized stochastic block model all edges have weight 1, namely, in the influence matrix, the $(i,j)$ entry equals $1$ if there is an edge between two agents $i$ and $j$, and $0$ otherwise. As in Section \ref{sec:sim} each block has $n$ agents. Each pair of agents from the same block is connected with probability $p_s$, and each pair of agents from different blocks is connected with probability $p_d$. Each agent in the $L$ block ($R$ block) has an initial opinion drawn independently from a known distribution $\mathcal{D}_L$ ($\mathcal{D}_R$). Opinions evolve according to $ \dot{x}_i(t) = \sum_{j \neq i} a_{ij} (x_j(t) - x_i(t)) + b (\sgn_{\epsilon}(x_i(t)) - x_i(t))$ for sufficiently small $\epsilon$ as in (\ref{eqn:main-dynamics-b}).

We consider the same set of simulations as we considered in Section \ref{sec:sim}. These simulations suggest that our theoretical findings and predictions apply also for the non-normalized stochastic block model. 

\begin{figure}
\begin{subfigure}{.47\textwidth}
  \centering
  \includegraphics[width=1.0\linewidth]{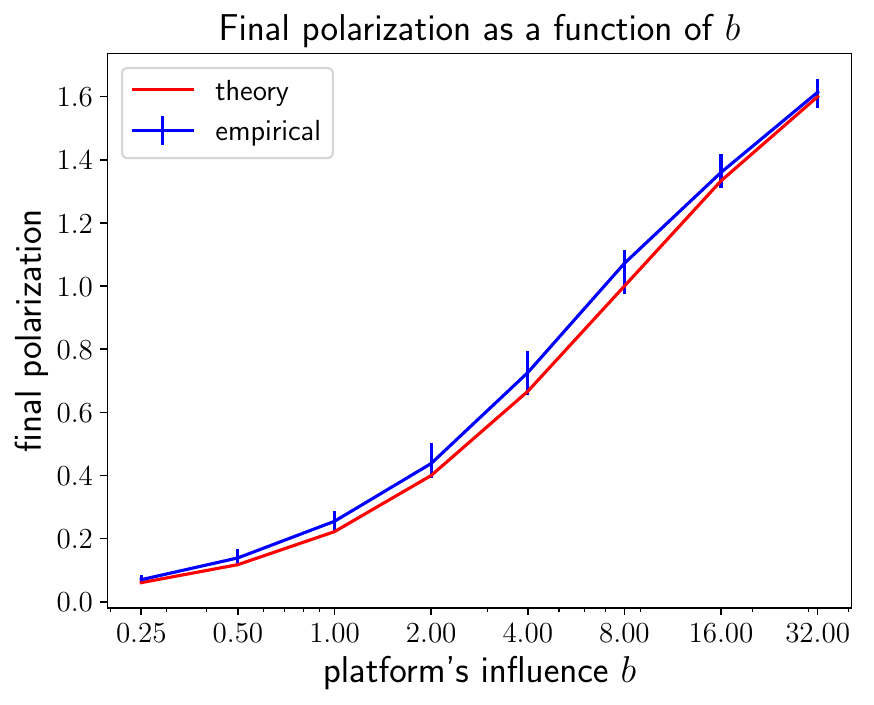}
  \subcaption[1.25\linewidth]{{The distributions of final polarization are tightly concentrated around theoretical values.}}
\end{subfigure}%
\hfill
\begin{subfigure}{.47\textwidth}
  \centering
  \includegraphics[width=1.02\linewidth]{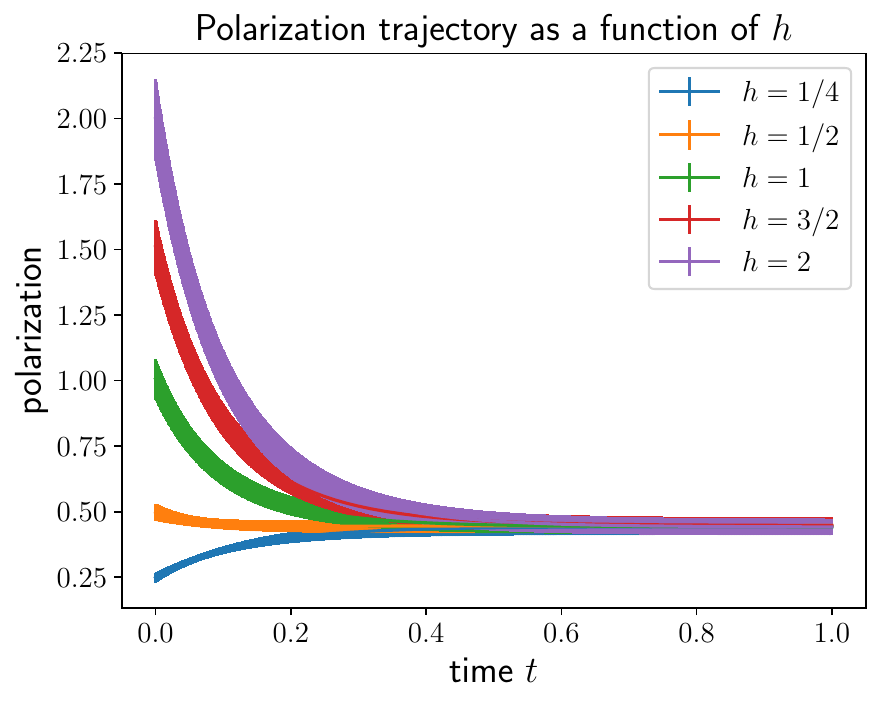}
  \subcaption{Polarization trajectories are monotonic and converge to the same limit polarization (for $b=2$).}
\end{subfigure}
\caption{The (random) empirical limit polarization is close to the theoretical prediction $2b/(b+2np_d)$ with high probability, and polarization trajectory converges monotonically to it for $(n, p_s, p_d) = (32, 1/4, 1/8)$, $\mathcal{D}_L \equiv \text{Unif}[-2,0]$, and $\mathcal{D}_R \equiv \text{Unif}[0,2]$.}
\end{figure}

\begin{figure}
  \centering
  \includegraphics[width=0.4\linewidth]{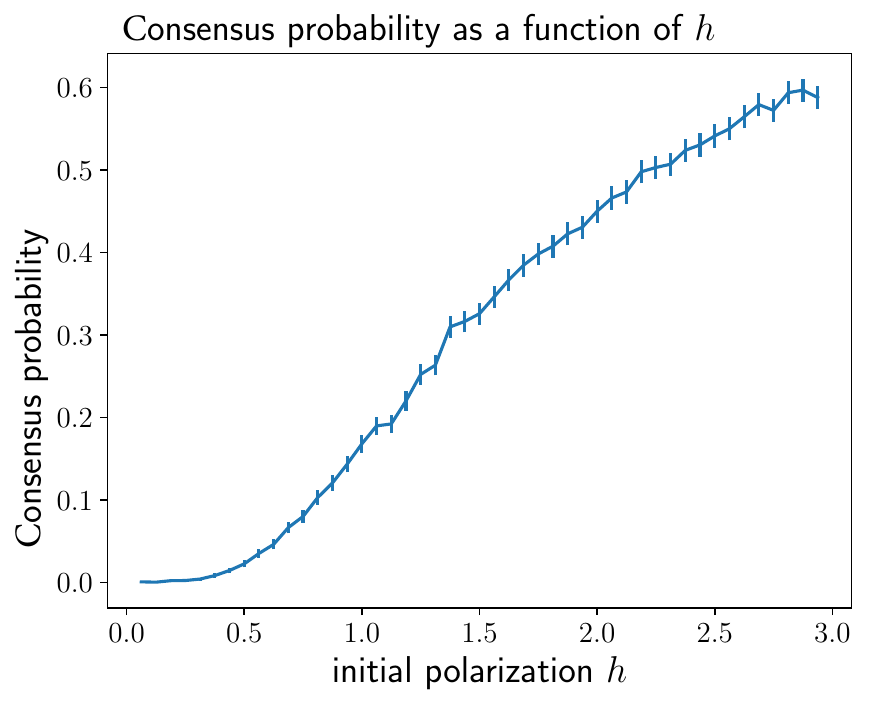}
  \caption{Consensus probability is increasing in mean initial polarization $h$.}
\caption{The effect of initial polarization on consensus for $(n,p_s,p_d,b) = (32, 1/4, 1/8,3/4)$ and $\mathcal{D}_L \equiv \text{Unif}[-h,0], \mathcal{D}_R \equiv \text{Unif}[0,h]$.}
\end{figure}
\begin{figure}
\begin{subfigure}{.48\textwidth}
  \centering
  \includegraphics[width=1.0\linewidth]{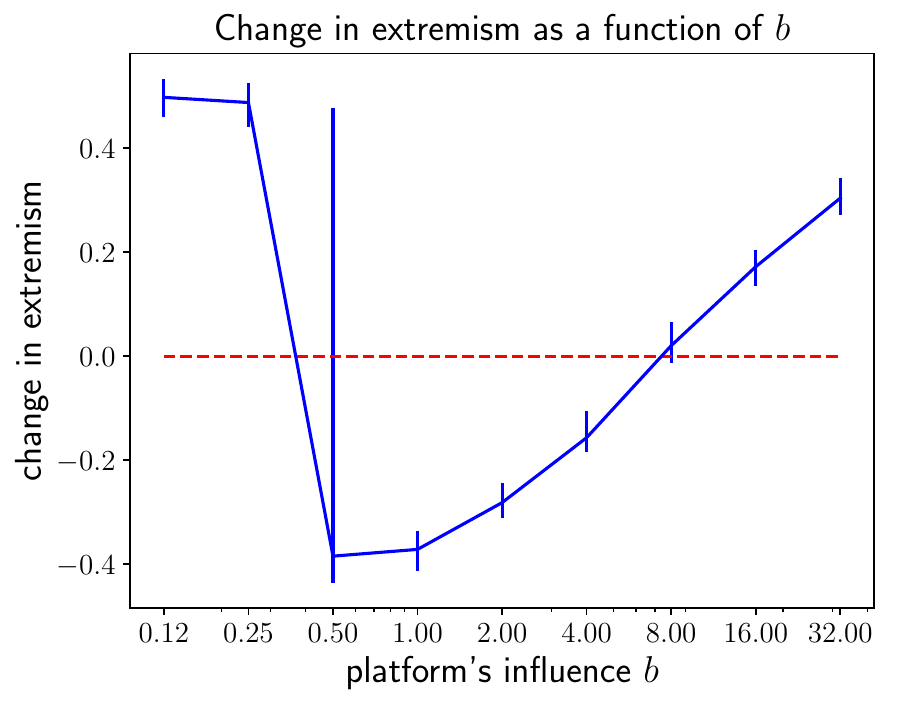}
  \caption{Interquartile range of change in extremism as a function of platform's strength $b$}
\end{subfigure}%
\hspace{1em}
\begin{subfigure}{.48\textwidth}
  \centering
  \includegraphics[width=1.06\linewidth]{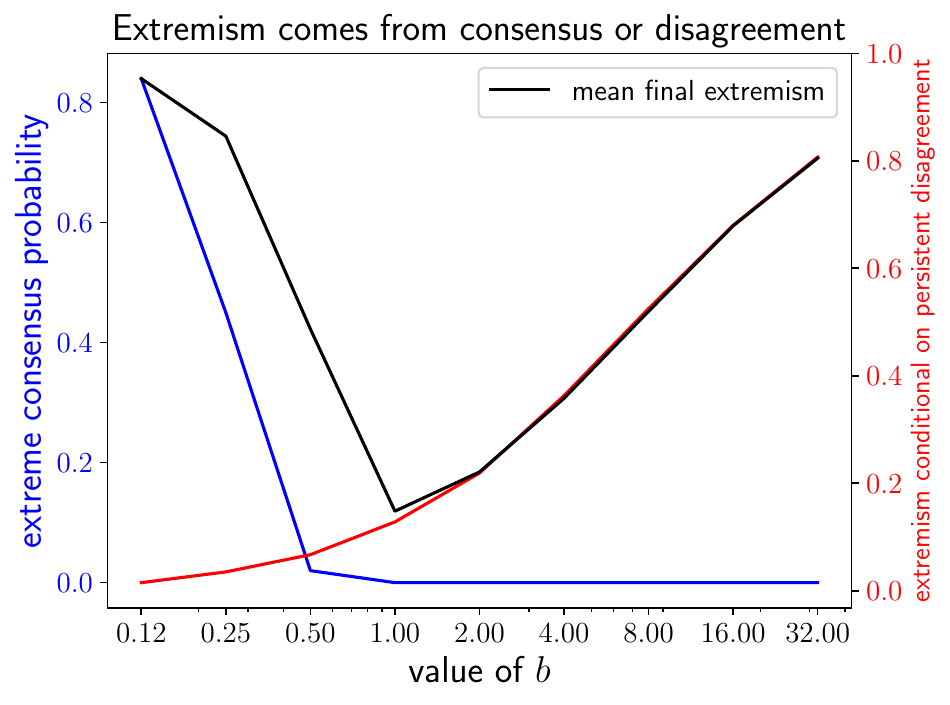}
  \caption{Extremism from consensus (low $b$) or disagreement (high $b$)}
\end{subfigure}
\caption{The effect of platform's strength $b$ on extremism for $(n, p_s, p_d) = (32, 1/4, 1/8)$, $\mathcal{D}_L \equiv \text{Unif}[-2,0]$ and $\mathcal{D}_R \equiv \text{Unif}[0,2]$.}
\end{figure}


\newpage

\section{Stochastic Block Model with Three Blocks} \label{Sec:Appendix3BSM}
We now present the stochastic block model with three blocks, which extends the two-block version. With a slight abuse of notation, we adopt similar notation to that used in the two-block case, and we provide the full model details below for completeness. 
In this model there are three sets of agents: A set of left agents, a set of moderate agents, and a set of right agents. Each set consists of $n$ agents. The probability of a connection between the same type of agents is given by $p(n)$ and the probability of a connection between different types of agents is given by $q(n)$.\footnote{Our results can be extended to the case where a connection between left and left, moderate and moderate, and right and right agents have probabilities of $p_{L}(n)$, $p_{M}(n)$, and $p_{R}(n)$ respectively where $p_{L}(n) \neq p_{M}(n) \neq p_{R}(n)$. To ease exposition, we assume that these probabilities are the same.} We denote by $A_{n}$  a typical adjacency matrix,  
 where an adjacency matrix satisfies $A_{n,ij}=1$ if there is a connection between agent $i$ and agent $j$ and zero otherwise, and by $\boldsymbol{A}_{n}$ the set of all possible adjacency matrices when there are $n$ left agents, $n$ moderate agents, and $n$ right agents. We denote by $\mathbb{P}_{n}(A_{n})$ the probability that a matrix $A_{n}$ is realized. With slight abuse of notation we denote $\mathbb{P}_{n}$ by $\mathbb{P}$.  
 The initial opinion of each left agent is a random variable on $X_{L} \subset \mathbb{R}_{-} = (-\infty,c)$, the initial opinion of each moderate agent is a random variable on $X_{M} \subset \mathbb{R}_{-} = (-c,c)$,  and the initial opinion of each right agent is a random variable on $X_{R} \subset \mathbb{R}_{+} = (c,\infty)$. Let $\mu_{0}$ be the probability measure that describes the initial opinions' distribution where we again omit the index $n$ for notational convenience. For simplicity we assume that the support of $\mu_{0}$ is given by a finite set $\boldsymbol{X}_{0}^{n}$  in $\mathbb{R}^{3n}$ where $\boldsymbol{X}_{0} = X_{L} \times X_{M} \times X_{R}$. We denote by  $(\boldsymbol{A}_{n} \times \boldsymbol{X}_{0}^{n},2^{ \boldsymbol{A}_{n}} \times 2^{\boldsymbol{X}_{0}^{n}},\mathbb{P} \otimes \mu_{0})$ the probability space that is generated by the stochastic block model described above where $\otimes$ denotes the product measure between two probability measures.

 Given a realized adjacency matrix $A$ and a realized initial opinion vector $\boldsymbol{x}_{0}$, each agent's opinion is dynamically evolving and is influenced by the platform and by the opinions of her connections as we described in Section \ref{sec:three}. In order to study a mean field regime where the number of agents tends to infinity we normalize the influence matrix.   We assume that if agent $i$ has $d_{i}$ connections then each connection has an influence weight of $a/d_{i}$, i.e.,  $a_{ij} = a/d_{i}$ is the influence of agent $j$ on agent $i$ if agent $j$ is a connection of agent $i$ and $0$ otherwise. 

 Hence, the opinions of agent $i$ evolve according to the equation
    \[
    \dot{x}_{i,n}(t) = \frac{a}{|N(i)|}\sum_{j\in N(i)} (x_{j,n}(t)-x_{i,n}(t)) +b (sgn_{\epsilon,c}(x_{i,n}(t))-x_{i,n}(t))
    \]
where $N(i)$ is the set of agents that are connected to agent $i$ and $|N(i)|$ is the cardinality of $N(i)$, i.e., the total number of connections that agent $i$ has.\footnote{If agent $i$ has no connections, i.e., $N(i)$ is an empty set then the other agents do not influence agent $i$'s opinions.} Note that the opinion of agent $i$ in some period $t$ is a random variable that depends on the realized adjacency matrix and realized initial opinions.

As in the two blocks case, we first analyze a deterministic three agent system that is used to approximate the dynamics of the stochastic block model with three blocks.

\subsection{Three Agent System}

We consider the three agent system $(a,b,x_{L},x_{M},x_{R})$ defined by
\[
\begin{cases} \label{Eq:3ODE}
\dot x_{1}(t)=a\bigl(x_{2}(t)-x_{1}(t)\bigr)+a\bigl(x_{3}(t)-x_{1}(t)\bigr)+b\bigl(-1-x_{1}(t)\bigr),\\[6pt]
\dot x_{2}(t)=a\bigl(x_{1}(t)-x_{2}(t)\bigr)+a\bigl(x_{3}(t)-x_{2}(t)\bigr)-b\,x_{2}(t),\\[6pt]
\dot x_{3}(t)=a\bigl(x_{1}(t)-x_{3}(t)\bigr)+a\bigl(x_{2}(t)-x_{3}(t)\bigr)+b\bigl(1-x_{3}(t)\bigr),
\end{cases}
\qquad
x_{1}(0)=x_L,\;x_{2}(0)=x_M,\;x_{3}(0)=x_R ,
\]

and we define the following constants 
\[
\mu \;:=\;\frac{b}{3a+b},\qquad
S \;:=\;\frac{x_L+x_M+x_R}{3},\qquad
\begin{aligned}
\Delta_1 &:=\frac{2x_L - x_M - x_R}{3},\\
\Delta_2 &:=\frac{-x_L + 2x_M - x_R}{3},\\
\Delta_3 &:=\frac{-x_L - x_M + 2x_R}{3}\;=\;-\bigl(\Delta_1+\Delta_2\bigr).
\end{aligned}
\]

The closed-form solution for the linear ODE system in Equation \ref{Eq:3ODE} is given by 
\[
\begin{aligned}
x_{1}(t) &= -\mu \;+\; S\,e^{-b t} \;+\; \bigl(\Delta_{1}+\mu\bigr)\,e^{-(3a+b)t},\\[6pt]
x_{2}(t) &= S\,e^{-b t} \;+\; \Delta_{2}\,e^{-(3a+b)t},\\[6pt]
x_{3}(t) &= \;\mu \;+\; S\,e^{-b t} \;+\; \bigl(\Delta_{3}-\mu\bigr)\,e^{-(3a+b)t}.
\end{aligned}
\]
with the unique equilibrium $(-\mu,0,\mu)$. 

We define 
\(
\mathcal R_c:=(-\infty,-c)\times(-c,c)\times(c,\infty)
\). The next theorem provides conditions that imply that $\mathcal R_c$ is invariant, that is, it provides conditions on the primitives of the three agent system such that $(x_1(t),x_2(t),x_3(t))_{t\ge0}$  is in $\mathcal{R}_{c}$ for all $t \geq 0$. It is immediate from the analysis below that these conditions are necessary for persistent disagreement, up to boundary cases where the inequalities in conditions $C_{1},C_{2},C_{3}$ hold with equality, as in the two-agent case.\footnote{Providing a full characterization of the three-outlet model would be   more complex than the two agent system as expressions like those in Lemma \ref{lem:cross-eps-band} would be algebraically very messy. Nevertheless, our simulations with large number of agents suggest that the mean-field approximation continues to hold qualitatively in these regimes as well. }

\begin{theorem} \label{Thm:3-agentsystem}
Fix $c>0$ and parameters $a,b>0$.  
The solution of the three agent system $(a,b,x_{L},x_{M},x_{R})$ with $x_{L} < -c$, $x_{M} \in (-c,c)$, $x_{R} > c$, $\mu > c$,  given by $(x_1(t),x_2(t),x_3(t))$ remains in $\mathcal R_c$
for every $t\ge0$ if the PD3 condition holds:
\begin{equation} \label{eq:PD3}
C_1\land C_2\land C_3,
\tag{PD3}
\end{equation}
where
\[
\begin{aligned}
C_1:\;&
\Bigl[S\,(\Delta_1+\mu) > 0\Bigr]
\;\lor\;
\Bigl[S\,(\Delta_1+\mu)<0
      \;\land\;
      \bigl(S<0 \lor r_1\ge1
            \;\lor\;
            (0<r_1<1
             \;\land\;
             r_1^{\,\tfrac b{3a}}
             <
             \tfrac{(3a+b)(\mu-c)}{3a\,S})\bigr)\Bigr],\\
&\qquad
 r_1=  \dfrac{-bS}{(3a+b)(\Delta_1+\mu)};\\[8pt]
C_2:\;&
\Bigl[S\,\Delta_2 > 0\Bigr]
\;\lor\;
\Bigl[S>0 \land \Delta_{2} <0 
      \;\land\;
      \bigl(r_2\ge1
            \;\lor\;
            (0<r_2<1
             \;\land\;
             r_2^{\frac{b}{3a}}
             <
             \frac{c(3a+b)}{3aS})\bigr)\Bigr] \\ 
             & \;\lor\;
\Bigl[S<0 \land \Delta_{2} >0 
      \;\land\;
      \bigl(r_2\ge1
            \;\lor\;
            (0<r_2<1
             \;\land\;
             r_2^{\frac{b}{3a}}
             >
             \frac{c(3a+b)}{-3aS})\bigr)\Bigr],
 r_2=\dfrac{-\,b\,S}{(3a+b)\,\Delta_2};\\[8pt]
C_3:\;&
\Bigl[S\,(\Delta_3-\mu) > 0\Bigr]
\;\lor\;
\Bigl[S\,(\Delta_3-\mu) < 0
      \;\land\;
      \bigl(S > 0 \lor r_3\ge1
            \;\lor\;
            (0<r_3<1
             \;\land\;
             r_3^{\,\tfrac{b}{3a}}
             < 
             \tfrac{(3a+b) (\mu-c)}{-3aS } )\bigr)\Bigr],\\
&\qquad
 r_3=\dfrac{-\,b\,S}{(3a+b)\,(\Delta_3-\mu)}.
\end{aligned}
\]

\end{theorem}
\vspace{-0.6em}
 
\begin{proof}
To prove the result 
we compute the time derivatives
\begin{align*}
    \dot{x}_1(t) & =  -bS e^{-bt}  - (3a+b) (\Delta_{1} + \mu) e^{-(3a+b)t} 
    \\
    \dot{x}_2(t) & = -bS e^{-bt}  - (3a+b) \Delta_{2} e^{-(3a+b)t} \\ 
     \dot{x}_3(t) & = - bS e^{-bt}   - (3a+b) (\Delta_{3} - \mu) e^{-(3a+b)t} 
\end{align*}

Each $C_i$ guarantees that $x_{i}$ remains in $\mathcal{R}_{c}$ where the first disjunct guarantees this through monotonicity, i.e., the two exponentials pull in the same direction, so the only potential extremum is at $t=0$ and $t=\infty$, while the second disjunct computes the extremum (minimum or maximum) and bounding it so that $x_{i}$  remains in $\mathcal{R}_{c}$. 
Exactly one disjunct is relevant for every concrete parameters.

First let $x_j(t) = D_0 + D_1 e^{-bt} + D_2 e^{-(3a+b)t}$. Then 
$$x_j'(t) = -b D_1 e^{-bt} - (3a+b) D_2 e^{-(3a+b)t}.$$

The approach to proving Theorem~\ref{Thm:3-agentsystem} is to observe that if $D_{1}$ and $D_{2}$ have the same sign, then $x_{j}$ is monotonic, and hence the extremum occurs at either the initial or terminal point. If they have opposite signs, we analyze the interior extremum point.

Assume now $D_{1}$ and $D_{2}$ have opposite signs. 

If $x_j'(t_{crit})=0$, 
then $(3a+b)D_2 e^{-(3a+b)t_{crit}} = -bD_1 e^{-bt_{crit}}$. 
The second derivative at such a critical point is 
$$x_j''(t_{crit}) = b^2 D_1 e^{-bt_{crit}} + (3a+b)^2 D_2 e^{-(3a+b)t_{crit}} = b^2 D_1 e^{-bt_{crit}} + (3a+b)(-bD_1 e^{-bt_{crit}}) = -3ab D_1 e^{-bt_{crit}}.$$ 
For $x_1, x_2, x_3$, the coefficient $D_1$ is $S$.
So, $x_j''(t_{crit}) = -3abS e^{-bt_{crit}}$.

Hence, if $S>0$ we have a local maximum and if $S<0$ we have a local minimum. 

We will now use this logic to analyze three different representative cases and the rest of the cases follow from a symmetric argument. 

\textbf{Case 1} $S > 0$, $\Delta _{2} > 0, \Delta_{1} + \mu > 0 $

Note that $\Delta_{3} - \mu = -\Delta_{1} - \Delta_{2} - \mu < 0$ so $C_1\land C_2\land C_3$ holds. 

First note $\dot{x}_1 (t) \leq 0$ for all $t$ so $x_1(t) \leq x_1(0) = x_{L} < -c$ for all $t$. 

In addition, $\dot{x}_2 (t) \leq 0$ for all $t$ so $x_{2}(t)$ is decreasing from $x_{2}(0) = x_{M} \in (-c,0]$ to $0$ and hence remains in the interval $(-c,c)$. 

Using the expression for $\dot{x}_3(t)$, we see that
\begin{align*}
    \dot{x}_3(t) > 0 \Longleftrightarrow t < \frac{1}{3a} \log \left( \frac{(\mu - \Delta_{3})(3a+b)}{bS}  \right) \equiv t^*
\end{align*}
So $x_3$ increases up to $t^*$ and decreases afterward, which implies $x_3(t) \geq \min(x_3(0), x_3(\infty)) > c$. Thus, in this case, we have $x_{1}(t) < -c$, $x_{2}(t) \in [-c,c]$, and $x_{3}(t) > c$ as required.

\textbf{Case 2}  $S>0,\quad
\Delta_2>0,\quad
\Delta_1+\mu<0,\quad
\alpha:=\Delta_3-\mu<0.$

With $S>0$ and $\beta:=\Delta_1+\mu<0$ we need to analyze the maximum of $x_1(t)$.
We have a critical point in $(0,\infty)$ exactly when $r_{1}<1$ and it is given by $$ t^{*} = \frac{1}{3a} \ln \frac{1}{r_{1}}.$$ When $r_{1} \geq 1$, $x_{1}(t)$ is monotone, and hence, we have $x_{1} (t) < -c$ for all $t \geq 0$ as $x_{1}(0) < -c$ and $-\mu < -c$. So assume $r_{1} \in (0,1)$. The critical point satisfies $-bS e^{-bt^{*} }  - (3a+b) \beta e^{-(3a+b)t^{*}} = 0 $. 

Plugging this into $x_{1}(t)$ we get $$x_{1}(t^{*}) = - \mu +S e ^{-bt^{*} } - \frac{bS}{3a+b} e ^{-bt^{*} }$$ so $x_{1}(t^{*}) < - c$ if and only if 
$$ r_{1}^{b/(3a)} < \frac{(3a+b)(\mu-c)}{3a\,S}$$
as required by $C_{1}$.

Now $S>0$ and $\Delta_2>0$ give $S\Delta_2>0$, so $C_2$ is satisfied via its first disjunct and
$ x_2(t) \in (-c,c)$ for all $t \geq 0$ as the path is monotone and $x_{2}(0) = x_M\in (-c,c)$, $x_{2}(\infty) = 0$.

Note that $C_3$ holds and $x_{3}(t) > c$ for all $t \geq 0 $ follows from the same argument as in Case 1. 

\textbf{Case 3}  $S>0,\quad
\Delta_2<0,\quad
\Delta_1+\mu > 0,\quad
\alpha:=\Delta_3-\mu<0.$ 

$C_{1}$ holds and $x_{1}(t)$ is monotone which implies that $x_{1}(t) < -c$ for all $t \geq 0$. 

For $x_{2}(t)$, $S>0$ implies that $x_{2}(t)$ has a maximum and we need to verify that the value of the maximum is less than $c$. Note that the maximum $t^{*}$ satisfies $-\Delta_{2}(3a+b) e ^{-(3a+b)t^{*} } = Sbe^{-bt^{*} }. $ 

Plugging this into $x_{2}$ and noting that $t^{*} = \frac{1}{3a} \ln \frac{1}{r_{2}}$ when $r_{2} \in (0,1)$ yield 
$$x_{2}(t^{*}) = S r_{2} ^{b/3a} - \frac{Sb}{3a+b} r_{2} ^{b/3a}. $$ 
Thus, $x_{2}(t^{*}) < c$ if and only if 
$$r_{2} ^{b/3a} < \frac{c(3a+b) }{3aS} $$
as required by condition $C_{2}$. If $r_{2} \geq 1$ then $x_{2}(t) $ is monotone on $[0,\infty)$, and hence, remains in $(-c,c)$. 

 $C_3$ holds and $x_{3}(t) > c$ for all $t \geq 0 $ follows from the same argument as in Case 1. 
\end{proof}

\subsection{Mean Field Approximation Theorem}
 We now state the approximation theorem for the stochastic block model with three blocks.  Let $X_{PD3}(a,b,c)$   be the set in $\mathbb{R}^{3}$ that includes the initial opinions that satisfy the PD3 condition defined in Theorem \ref{Thm:3-agentsystem} for fixed $a,b,c$.

\begin{theorem} \label{Theorem: 3BSMmain}
    Suppose that $p(n)$ and $q(n)$ are $\omega(\ln (n) / n)$. Fix $c>0$ and assume that $\epsilon$ is small so the characterization of  Theorem \ref{Thm:3-agentsystem} holds.  Let  $\rho = \lim_{n \rightarrow \infty} q(n)/(2q(n)+p(n)) $.

 Consider a three agent system $(a \rho ,b, (x_{L},x_{M},x_{R}) )$ that satisfies the PD3  condition with an equilibrium $\boldsymbol{e} = (e_{L},e_{M},e_{R})$.  If $X_{L} \times X_{M} \times  X_{R} \subseteq X_{PD3}(a\rho,b,c)$,  
 then for every $\epsilon' >0$, there exist $\delta>0$, $N>0$, $T>0$, and sets $C_{\delta,n} \subset \boldsymbol{A}_{n}$ such that  $\lim  _{n \rightarrow \infty} \mathbb{P}(C_{\delta,n}) = 1$ and for all $n \geq N$, all $t \geq T$, all $A_{n} \in C_{\delta,n}$, all $\boldsymbol{x}_{0} \in (X_{L} \times X_{M} \times X_{R})^{n}$, we have 
 \begin{equation}\label{Eq:MainThmStatement3BSM}
|  x_{i,n,A_{n},\boldsymbol{x}_{0} }(t) - e_{L}| \leq \epsilon' \text {, } |  x_{k,n,A_{n},\boldsymbol{x}_{0} }(t) - e_{M}| \leq \epsilon' \text{ and } |  x_{j,n,A_{n},\boldsymbol{x}_{0}}(t) - e_{R}| \leq \epsilon'
\end{equation}
 for every left agent $i$, every moderate agent $k$, and every right agent $j$. 
\end{theorem}

Let $\delta \in (0,1)$, $n \geq 0$ and denote  the set of left, moderate and right agents by 
\(\mathcal L_n,\ \mathcal M_n,\ \mathcal R_n\), respectively.  Similarly to the proof of Theorem \ref{Thm:Main} we define for each agent $i$, 
\[
L_{n,i}:=|N(i)\cap\mathcal L_n|,\quad
M_{n,i}:=|N(i)\cap\mathcal M_n|,\quad
R_{n,i}:=|N(i)\cap\mathcal R_n|,
\]

and the sets 
\[
S_{\delta,n}:=\bigl[(1-\delta)p(n)n,\,(1+\delta)p(n)n\bigr],\qquad
D_{\delta,n}:=\bigl[(1-\delta)q(n)n,\,(1+\delta)q(n)n\bigr]
\]
and
\[
C_{\delta,n}
=\Bigl\{A_n:\;
\begin{array}{l}
\text{for all } i\in\mathcal L_n:\;
          L_{n,i}\in S_{\delta,n},\;
          M_{n,i},R_{n,i}\in D_{\delta,n},\\[2pt]
\text{for all } i\in\mathcal M_n:\;
          M_{n,i}\in S_{\delta,n},\;
          L_{n,i},R_{n,i}\in D_{\delta,n},\\[2pt]
\text{for all } i\in\mathcal R_n:\;
          R_{n,i}\in S_{\delta,n},\;
          L_{n,i},M_{n,i}\in D_{\delta,n}
\end{array}
\Bigr\}.
\]
From the proof of Lemma \ref{Lemma:Concent} and the union bound we have 
$\mathbb{P} (C_{\delta,n} )\xrightarrow[n\to\infty]{}1$.

Put
\[
a_-:=\frac{a(1-\delta)q(n)}{(1+\delta)\bigl[p(n)+2q(n)\bigr]},
\qquad
a_+:=\frac{a(1+\delta)q(n)}{(1-\delta)\bigl[p(n)+2q(n)\bigr]},
\qquad a_-<a_+ .
\]

We define the following system of differential equations: 

\begin{equation} \label{Eq:ODE3}
\begin{aligned} 
\dot{\bar x}_L &= a_+\!\!\bigl[(\bar x_M-\bar x_L)+(\bar x_R-\bar x_L)\bigr] + b(\sgn_{\epsilon,c}(\bar x_{L}) -\bar x_L),\\
\dot{ \underline x}_L      &= a_-\!\!\bigl[(\underline x_M- \underline x_L)+(\underline x_R- \underline x_L)\bigr]              + b(\sgn_{\epsilon,c}(\underline x_{L}) - \underline x_L),\\[4pt]
\dot{\bar x}_M &= a_-\!(\bar x_L-\bar x_M)+a_+(\bar x_R-\bar x_M) + b(\sgn_{\epsilon,c}(\bar x_{M}) -\bar x_M),\\
\dot{\underline x}_M      &= a_+\!(\underline x_L- \underline x_M)+a_-(\underline x_R- \underline x_M) + b( \sgn_{\epsilon,c}(\underline x_{M}) - \underline x_M),\\[4pt]
\dot{\bar x}_R &= a_-\!\!\bigl[(\bar x_L-\bar x_R)+(\bar x_M-\bar x_R)\bigr] + b(\sgn_{\epsilon,c}(\bar x_{R})-\bar x_R),\\
\dot{\underline x}_R      &= a_+\!\!\bigl[(\underline x_L- \underline x_R)+( \underline x_M-\underline x_R)\bigr]              + b(\sgn_{\epsilon,c}(\underline x_{R})- \underline x_R)
\end{aligned}
\end{equation}
with the initial conditions $\bar x_{g}(0) = \max X_{g} $ and $\underline x_{g}(0) = \min X_{g} $  for $g \in \{ L,M,R \}$.

Note that the choice of the coefficients $a_+$ and  $a_-$ is more subtle than in the two block setting as each agent now interacts with two out-groups, so the coefficients choice  need to guarantee that we can prove the following version of Lemma \ref{Lemma: bounds} that shows that we can bound the  agents' opinions $\boldsymbol{x}_{n,A_{n}, \boldsymbol{x}_{0}}(t)$ for adjacency matrices $A_{n}$ that belong to $C_{\delta,n}$ and initial opinions that belong to $(X_{L} \times X_{M} \times  X_{R})^{n}$ by the auxiliary system of differential equations above.

\begin{lemma} \label{Lemma: bounds3} Suppose that $  {\bar{ x}}_L(t)  < {\underline{ x}}_M(t) $ and  $  {\bar{ x}}_M(t)  < {\underline{ x}}_R(t) $  for $t \geq 0$. 

(i) We have ${\underline{ x}}_R(t) \leq  {\bar{ x}}_R(t)$, ${\underline{ x}}_M(t) \leq  {\bar{ x}}_M(t)$, and ${\underline{ x}}_L(t) \leq {\bar{ x}}_L(t)$ for all $t \geq 0  $. 

(ii)  For all $A_{n} \in C_{\delta,n}$, $\boldsymbol{x}_{0} \in (X_{L} \times X_{M} \times X_{R})^{n}$ and all $t \geq 0 $ we have
 \begin{align*} \label{Eq: 3SBMineq}
& {\underline{ x}}_L(t) \leq  x_{i,n,A_{n}, \boldsymbol{x}_{0}}(t) \leq   {\bar{ x}}_L(t), \ \forall i \in \mathcal{L}_{n}, \text {  } {\underline{ x}}_M(t) \leq  x_{k,n,A_{n}, \boldsymbol{x}_{0}}(t) \leq   {\bar{ x}}_M(t), \ \forall k \in \mathcal{M}_{n}, \\   & {\underline{ x}}_R(t) \leq  x_{j,n,A_{n} , \boldsymbol{x}_{0}}(t) \leq   {\bar{ x}}_R(t),  
\ \forall j \in \mathcal{R}_{n}.
\end{align*}
\end{lemma} 

\begin{proof}
(i) Assume in contradiction that there exists a $t \geq 0$ such that   ${\underline{ x}}_R(t) > {\bar{ x}}_R(t)$ or ${\underline{ x}}_M(t) > {\bar{ x}}_M(t)$ or ${\underline{ x}}_L(t) > {\bar{ x}}_L(t)$. Note that $t>0$.

Let $t_{1} = \inf \{t \in [0,\infty): {\underline{ x}}_R(t) > {\bar{ x}}_R(t) \text{ or } {\underline{ x}}_M(t) > {\bar{ x}}_M(t) \text{ or }  {\underline{ x}}_L(t) > {\bar{ x}}_L(t) \}$. By  the contradiction assumption and the continuity of the solutions of the ordinary differential equations  $t_{1}$ is finite, and we have  $ {\underline{ x}}_g(t_{1}) \leq  {\bar{ x}}_g(t_{1})$ for $g \in \{L,M,R\}$ with at least one of the inequalities as equality. 
We will only consider the case $ {\underline{ x}}_R(t_{1}) \leq {\bar{ x}}_R(t_{1})$,  $ {\underline{ x}}_M(t_{1}) = {\bar{ x}}_M(t_{1})$ and ${\underline{ x}}_L(t_{1}) \leq {\bar{ x}}_L(t_{1})$ as the rest of the cases follow from an analogous argument. We have
\begin{align*}
    \dot{\bar{x}}_M(t_{1}) &= 
       a_{-}(\bar{x}_L(t_{1}) - \bar{x}_M(t_{1})) +  a_{+}(\bar{x}_R(t_{1}) - \bar{x}_M(t_{1})) + b( \sgn_{\epsilon,c}(\bar x_{M}(t_{1}) )-  \bar{x}_M(t_{1}))
      \\
      & >
       a_{+}(\underline{x}_L(t_{1}) - \underline{x}_M(t_{1})) +  a_{-}(\underline{x}_R(t_{1}) - \underline{x}_M(t_{1})) + b( \sgn_{\epsilon,c}(\underline x_{M}(t_{1}) )-  \underline{x}_M(t_{1}))   =  \dot{\underline{x}}_M(t_{1})
\end{align*}
so that $\underline{ x}_M(t) < {\bar{ x}}_M(t)$ for $t > t_{1}$ that is close enough to $t_{1}$. 

For $g \in \{L,R \}$, from continuity, if ${\underline{ x}}_g(t_{1}) < {\bar{ x}}_g(t_{1})$  then ${\underline{ x}}_g(t) < {\bar{ x}}_g(t)$ for $t > t_{1}$  that is close enough to $t_{1}$. If  ${\underline{ x}}_g(t_{1}) = {\bar{ x}}_g(t_{1})$ then we can continue exactly as the inequalities above, e.g., if  ${\underline{ x}}_L(t_{1}) = {\bar{ x}}_L(t_{1})$ then 
\begin{align*}
    \dot{\bar{x}}_L(t_{1}) &= 
       a_{+}(\bar{x}_M(t_{1}) - \bar{x}_L(t_{1}))  +   a_{+}(\bar{x}_R(t_{1}) - \bar{x}_L(t_{1})) 
    +
       b(\sgn_{\epsilon,c}(\bar x_{L}(t_{1}) ) - \bar{x}_L(t_{1})) \\
      & >
         a_{-}(\underline{x}_M(t_{1}) - \underline{x}_L(t_{1}))  +   a_{-}(\underline{x}_R(t_{1}) - \underline{x}_L(t_{1})) 
    +
       b(\sgn_{\epsilon,c}(\underline x_{L}(t_{1}) ) - \underline{x}_L(t_{1})) = \underline{x}_L(t_{1})
\end{align*}
so that $\underline{ x}_L(t) < {\bar{ x}}_L(t)$ for $t > t_{1}$ that is close enough to $t_{1}$. 
Hence, $\underline{ x}_R(t) < {\bar{ x}}_R(t)$, $\underline{ x}_M(t) < {\bar{ x}}_M(t)$, and ${\underline{ x}}_L(t) < {\bar{ x}}_L(t)$ for $t > t_{1}$ that is close enough to $t_{1}$, which is a contradiction to the definition of $t_{1}$.

(ii) Let $A_{n} \in C_{\delta,n}$ and $\boldsymbol{x}_{0} \in (X_{L} \times X_{M} \times X_{R} ) ^{n}$ . The claim of the lemma  holds for $t=0$ by construction. Suppose that the claim holds for some $t>0$. 

Consider the case where agent $i$ is a moderate agent that has the highest opinion at time $t$ among all moderate agents. We have
\begin{align*}
    \dot{x}_{i,n,A_{n}, \boldsymbol{x}_{0}}(t) &= \frac{a}{|N(i)|}\sum_{j\in N(i)} (x_{j,n,A_{n}, \boldsymbol{x}_{0}}(t)-x_{i,n,A_{n}, \boldsymbol{x}_{0}}(t)) + b(\sgn_{\epsilon,c} (x_{i,n,A_{n}, \boldsymbol{x}_{0}}(t) )-  x_{i,n,A_{n}, \boldsymbol{x}_{0}}(t) ) \\
    & \leq  \frac{a}{|N(i)|}\sum_{j\in N(i) \cap \mathcal{L}_{n} } (x_{j,n,A_{n}, \boldsymbol{x}_{0}}(t)-x_{i,n,A_{n}, \boldsymbol{x}_{0}}(t)) \\ 
    & +  \frac{a}{|N(i)|}\sum_{j\in N(i)  \cap \mathcal{R}_{n}} (x_{j,n,A_{n}, \boldsymbol{x}_{0}}(t)-x_{i,n,A_{n}, \boldsymbol{x}_{0}}(t))  + b(\sgn_{\epsilon,c} (x_{i,n,A_{n}, \boldsymbol{x}_{0}}(t) )-  x_{i,n,A_{n}, \boldsymbol{x}_{0}}(t) )\\
    & \leq  \frac{a|N(i)\cap \mathcal{L}_{n} | }{|N(i)|}(\bar{ x}_L(t)-x_{i,n,A_{n}, \boldsymbol{x}_{0}}(t))  + \frac{a|N(i)\cap \mathcal{R}_{n} | }{|N(i)|}(\bar{ x}_R(t)-x_{i,n,A_{n}, \boldsymbol{x}_{0}}(t)) \\
    & + b(\sgn_{\epsilon,c} (x_{i,n,A_{n}, \boldsymbol{x}_{0}}(t) )-  x_{i,n,A_{n}, \boldsymbol{x}_{0}}(t) ) \\
    & \leq a_{-} ({\bar{ x}}_L(t)-x_{i,n,A_{n}, \boldsymbol{x}_{0}}(t)) + a_{+} ({\bar{ x}}_R(t)-x_{i,n,A_{n}, \boldsymbol{x}_{0}}(t))  + b(\sgn_{\epsilon,c} (x_{i,n,A_{n}, \boldsymbol{x}_{0}}(t) )-  x_{i,n,A_{n}, \boldsymbol{x}_{0}}(t) ).
    \end{align*} 
    The first inequality follows because $x_{i',n,A_{n}, \boldsymbol{x}_{0}}(t)-x_{i,n,A_{n}, \boldsymbol{x}_{0}}(t) \leq 0 $ for every moderate agent $i' \in \mathcal{M}_{n}$. The second inequality follows because $x_{i,n,A_{n}, \boldsymbol{x}_{0}}(t) \leq \bar{x}_{L}(t)$ for every left agent $i \in \mathcal{L}_{n}$ and $x_{j,n,A_{n}, \boldsymbol{x}_{0}}(t) \leq \bar{x}_{R}(t)$ for every right agent $j \in \mathcal{R}_{n}$. The third inequality follows because
    $A_{n} \in C_{\delta,n}$ implies that the cardinality of $N(i) \cap \mathcal{L}_{n}$ is between $(1-\delta)q(n)n$ and $(1+\delta)q(n)n$, and the cardinality of $N(i)$ is between $(1-\delta)(p(n)+2q(n))n$  and $(1+\delta)(p(n)+2q(n))n$. 
    
    We conclude that $\dot{x}_{i,n,A_{n}, \boldsymbol{x}_{0} }(t) \leq \bar{U}(t,x_{i,n,A_{n}, \boldsymbol{x}_{0}}(t))$ where 
    $$  \bar{U}(t, y)= a_{-} (\bar{x}_L(t) - y) + a_{+}(\bar{x}_R(t) - y) +b(\sgn_{\epsilon,c} (y) -y) .$$
    
     Theorem 4.1 in \cite{Hartman2002} implies that $x_{i,n,A_{n}, \boldsymbol{x}_{0}}(t) \leq \bar{ x}_M(t)$ on $[t, t+\delta']$ for some $\delta' >0$ that does not depend on $t$. An analogous argument proves that the other inequalities in the statement of the lemma  hold on an interval $[t, t+\delta'']$ for some $\delta'' >0$. 
     We conclude that the inequalities in the statement of the Lemma hold for all $t  \geq 0$, and the proof is complete. 
\end{proof}

Now suppose that the three-agent system $(a \rho ,b, (x_{L},x_{M}, x_{R}) )$ satisfies the PD3 condition, has a solution $(x_{1}(t),x_{2}(t),x_{3}(t))$,  an equilibrium $\boldsymbol{e}=(e_{L},e_{M},e_{R})$ and that $X_{L} \times X_{M} \times X_{R} \subseteq X_{PD3}(a\rho , b,c)$. From the proof of Theorem \ref{Thm:3-agentsystem} we have $\boldsymbol{e}= (-\mu,0, \mu) $.
 
From the Theorem's assumption, for a small fixed $\delta>0$ we can find a large $N$ such that for all $n \geq N$  and all $\epsilon'' >0$ we have $a_{-} \geq a \rho - \epsilon''$ and $a_{+} \leq a \rho +\epsilon''$. Note that this implies that $X_{L} \times X_{M} \times X_{R} \subseteq X_{PD3}(a_{-} , b,c)$ and $X_{L} \times X_{M} \times X_{R} \subseteq X_{PD3}(a_{+} , b,c)$ for small enough $\delta$ and large $n$. 
 
 Consider the following differential equations:
 \begin{equation} \label{eq:bd:positive3}
     \begin{aligned}
       \dot {\bar{ x}}_L' (t) &=  
       a_{+}(\bar{x}_M '(t) - \bar{ x}_L'(t))  +  a_{+}(\bar{x}_R '(t) - \bar{ x}_L'(t)) + 
       b(-1 - \bar{ x}_L'(t))  \\
         \dot {\bar{ x}}_M' (t) &=  
       a_{-}(\bar{x}_L '(t) - \bar{ x}_M'(t))  +    a_{+}(\bar{x}_R '(t) - \bar{ x}_M'(t))  
       - b \bar{ x}_M'(t)  \\
\dot{\bar{x}}_R'(t) &= 
       a_{-}(\bar{x}_L'(t) - \bar{x}_R'(t))  +   a_{-}(\bar{x}_M'(t) - \bar{x}_R'(t))+ 
       b(1 - \bar{x}_R'(t))  
      \end{aligned}
 \end{equation}
 with the initial conditions $\bar{x}_L'(0)=\bar{x}_{L}(0)$, $\bar{x}_M'(0)=\bar{x}_{M}(0)$
 and
 $\bar{x}_R'(0)=\bar{x}_{R}(0)$.
 Now we can use the same argument from Theorem \ref{Thm:Main} to show that under the PD3 condition 
 to  conclude that  $\bar{ x}_L (t)$,  $\bar{ x}_M (t)$,  and $\bar{ x}_R (t)$ satisfy the  differential equations given in Equation (\ref{eq:bd:positive3}) for a small $\epsilon>0$ and $ |\bar{ x}_L (t) - e_{L} | \leq \delta' $, $ |\bar{ x}_M (t) - e_{M} | \leq \delta' $, $ |\bar{ x}_R (t) - e_{R} | \leq \delta' $ for all $t \geq T'$ and small enough $\epsilon''$.  A similar argument (using Equation (\ref{eq:bd:positive3}) with initial opinions $\underline{x}_{L}(0)$, $\underline{x}_{M}(0)$,  $\underline{x}_{R}(0)$) shows that  $ |\underline{ x}_L (t) - e_{L} | \leq \delta' $, $ |\underline{ x}_M (t) - e_{M} | \leq \delta' $ and $ |\underline{ x}_R (t) - e_{R} | \leq \delta' $ for all $t \geq T'$ and small enough $\epsilon''$.

We can now apply Lemma \ref{Lemma: bounds3}  to conclude that for all $\epsilon' >0$, we can find an $N$, $T'$, and $\delta>0$ such that for all $t \geq T'$ and all $A_{n} \in C_{\delta,n}$, $\boldsymbol{x}_{0} \in (X_{L} \times X_{M} \times X_{R})^{n}$ with $n \geq N$, we have $|x_{i,n,A_{n}, \boldsymbol{x}_{0}}(t) - e_{L}| \leq \epsilon'$ for every left agent $i$,  $|x_{k,n,A_{n}, \boldsymbol{x}_{0}}(t) - e_{M}| \leq \epsilon'$ for every moderate agent $k$, and $|x_{j,n,A_{n}, \boldsymbol{x}_{0}}(t) - e_{R}| \leq \epsilon'$ for every right agent $j$.

\end{document}